\documentclass[12pt]{article}
\usepackage{enumerate}
\usepackage[round]{natbib}
\usepackage[resetlabels]{multibib} 
\newcites{oa}{References (Online Appendices)}
\usepackage{url} 
\usepackage[T1]{fontenc}
\usepackage[utf8]{inputenc}
\usepackage{fancyhdr}
\usepackage[english]{babel}
\usepackage[pdftex]{graphicx}
\usepackage{pdfpages}
\usepackage{comment}
\usepackage{changepage}
\usepackage{graphicx}
\usepackage{paralist}
\usepackage{enumitem,kantlipsum}
\usepackage{listings}
\usepackage{nth}
\usepackage[english, ruled,vlined,resetcount]{algorithm2e}

\usepackage{algorithm2e}
\SetKwComment{Comment}{/* }{ */}
\SetAlgoCaptionSeparator{}
\usepackage{amsmath, amsthm, mathtools, amsfonts,mathrsfs, upgreek, mathabx, wasysym}
\usepackage{adjustbox}
\usepackage{here}
\usepackage{rotating}
\usepackage{multirow}
\usepackage{booktabs}
\usepackage{lettrine}
\usepackage{dsfont}
\usepackage{textcomp}
\usepackage[tikz]{bclogo}
\usepackage{fancybox}
\usepackage{pifont}
\usepackage{color, xcolor}
\usepackage{bbm}
\usepackage{multicol}
\usepackage{footnote}
\usepackage{supertabular, longtable, caption, afterpage, tabularx, threeparttable}
\usepackage{colortbl}
\usepackage{dcolumn}
\newcolumntype{d}[1]{D..{#1}} 
\usepackage{siunitx}
\usepackage{lscape}

\usepackage[colorlinks=true,urlcolor= blue,citecolor=blue]{hyperref}
\usepackage{titling}
\usepackage{xpatch}
\usepackage{titlesec}
\titleformat{\section}{\normalsize\bfseries}{\thesection}{1em}{}
\titleformat{\subsection}{\normalsize\bfseries}{\thesubsection}{1em}{}
\titlespacing*{\section}{0pt}{0.5\baselineskip}{0.2\baselineskip}
\titlespacing*{\subsection}{0pt}{0.5\baselineskip}{0.2\baselineskip}
\titlespacing*{\subsubsection}{0pt}{0.1\baselineskip}{0.2\baselineskip}
\usepackage{ragged2e}
\makeatletter
\newtheorem{theorem}{Theorem}[section]
\newtheorem{proposition}{Proposition}[section]
\newtheorem{lemma}{Lemma}[section]
\newtheorem{definition}{Definition}[section]
\newcommand{\neutralize}[1]{\expandafter\let\csname c@#1\endcsname\count@}
\newtheorem{assumption}{Assumption}[section]

\newtheorem{example}{Example}[section]

\newtheorem{remark}{Remark}[section]
\newcommand\Tstrut{\rule{0pt}{2.6ex}}         
\newcommand\Bstrut{\rule[-0.9ex]{0pt}{0pt}}  
\let\endtitlepage\relax

\usepackage{authblk}

\DeclareMathOperator{\plim}{plim}
\DeclareMathOperator{\diag}{diag}

\newenvironment{mytitlepage}
{\begin{titlepage}\def\@thanks{}}
	{\end{titlepage}}
\xpatchcmd\titlepage{\setcounter{page}\@ne}{}{}{}
\xpatchcmd\endtitlepage{\setcounter{page}\@ne}{}{}{}
\newcommand{\ostar}{\mathbin{\mathpalette\make@circled\star}}
\newcommand{\make@circled}[2]{
	\ooalign{$\m@th#1\smallbigcirc{#1}$\cr\hidewidth$\m@th#1#2$\hidewidth\cr}
}
\newcommand{\smallbigcirc}[1]{
	\vcenter{\hbox{\scalebox{0.77778}{$\m@th#1\bigcirc$}}}
}
\date{November 2025}
\usepackage{setspace}
\definecolor{colari}{rgb}{0.7, 0, 0.7}
\definecolor{coland}{rgb}{0, 0.7, 0.4}

\makeatletter
\usepackage{authblk}
\newcommand{\TITLE}{Count Data Models with Heterogeneous Peer Effects under Rational Expectations}
\title{\vspace{-2cm}\TITLE}
\author{Aristide Houndetoungan\footnote[1]{\fontsize{10pt}{10pt} I am thankful to Frank Windmeijer, the editor, and two anonymous referees for their helpful comments that improved the manuscript; I am also grateful to Vincent Boucher, Bernard Fortin, Yann Bramoullé, Ismael Mourifié, Stéphane Bonhomme, Elie Tamer, Michael Vlassopoulos, Arnaud Dufays, Luc Bissonnette, and Marion Goussé for their valuable suggestions, insights, and discussions. \\
		This research uses data from Add Health, a program directed by Kathleen Mullan Harris and designed by J. Richard Udry, Peter S. Bearman, and Kathleen Mullan Harris at the University of North Carolina at Chapel Hill, and funded by Grant P01-HD31921 from the Eunice Kennedy Shriver National Institute of Child Health and Human Development, with cooperative funding from 23 other federal agencies and foundations. Special acknowledgment is given to Ronald R. Rindfuss and Barbara Entwisle for assistance in the original design. Information on how to obtain Add Health data files is available on the Add Health website (\href{http://www.cpc.unc.edu/addhealth}{www.cpc.unc.edu/addhealth}). No direct support was received from Grant P01-HD31921 for this research.\\ I provide an easy-to-use R package---named CDatanet---for implementing the model and methods used in this paper. The package is located at \href{https://CRAN.R-project.org/package=CDatanet}{https://CRAN.R-project.org/package=CDatanet}.}}
\affil{\normalsize\emph{Laval University}\\
	\href{mailto:ahoundetoungan@ecn.ulaval.ca}{ahoundetoungan@ecn.ulaval.ca}\vspace{-2pt}}

\begin{document}
	\setlength{\abovedisplayskip}{4pt}
	\setlength{\belowdisplayskip}{4pt}
	\begin{mytitlepage}
		\maketitle
		
		\vspace{-0.8cm}
		\begin{abstract}
			\noindent 
			{\linespread{1.1}\selectfont
				This paper develops a peer effect model for count responses under rational expectations. The model accounts for heterogeneity in peer effects across groups based on observed characteristics. Identification is based on the linear model condition that requires the presence of friends of friends who are not direct friends. I show that this identification condition extends to a broad class of nonlinear models. Parameters are estimated using a nested pseudo-likelihood approach. An empirical application to students’ extracurricular participation reveals that females are more responsive to peers than males. An easy-to-use R package, \texttt{CDatanet}, is available for implementing the model.
                
			}
		\end{abstract}
	\end{mytitlepage}

	\newpage
	\setlength{\abovedisplayskip}{3pt}
	\setlength{\belowdisplayskip}{3pt}
	\section{Introduction}
	There is a large and growing literature on peer effects in economics.\footnote{For recent reviews, see \cite{de2017econometrics} and \cite{bramoulle2020peer}.} Recent contributions cover, among others, models for limited dependent variables, including binary \citep[e.g.,][]{brock2001discrete, lee2014binary, lin2024binary}, ordered \citep[e.g.,][]{liu2017social, xu2018social, boucher2022peer}, multinomial \citep[e.g.,][]{guerra2020multinomial}, and censored \citep[e.g.,][]{xu2015maximum} responses. However, count variables have received little attention, despite their prevalence in survey data (e.g., the number of physician visits, the frequency of consumption of a good or service, and the frequency of participation in activities). Peer effects on these variables are often estimated using a linear-in-means (or Tobit) model or a binary model after transforming the outcome into binary data. These approaches may lead to biased estimates because the linear-in-means (or Tobit) models are developed for (censored) continuous outcomes. Moreover, transforming the outcome to a binary format does not allow the interpretation of peer influence in terms of intensive marginal effects.

	This paper develops a microfounded rational-expectations model to estimate peer effects on count outcomes. The model incorporates heterogeneity, allowing for the influence of friends to vary across agents' and friends' observed groups (e.g., gender- or ethnicity-based groups). I establish conditions under which the model has a unique Bayesian Nash equilibrium (BNE). I demonstrate that the model parameters are identified under the well-known identification condition of linear-in-means models, which relies on the presence of friends' friends who are not friends in the network \citep{bramoulle2009identification}. Importantly, I generalize this identification result to a large class of nonlinear models. I apply the model to student participation in extracurricular activities using the National Longitudinal Study of Adolescent to Adult Health (Add Health) data. I find that females are more responsive to their peers than males.

	The proposed model is based on a static game with incomplete information, in which agents' payoffs depend on their friends' expected outcomes. I partition agents into observed groups and assume that peer effects depend on both agents' groups and peers' groups. Unlike standard games with continuous or censored outcomes, where payoffs are typically modeled as a linear-quadratic function \citep[see][]{blume2015linear}, I specify a semiparametric payoff. The assumption of a linear-quadratic payoff is widely employed in the literature, as it yields linear or Tobit specifications that are easy to estimate. I show that imposing this assumption on count outcomes, which is equivalent to estimating a linear spatial autoregressive model (SAR) or SAR-Tobit model \citep[][]{xu2015maximum}, yields an ordered model with evenly spaced thresholds or cutoff points. This assumption is restrictive and may lead to misspecification issues, particularly for count outcomes with large support. By employing a more flexible payoff function, the resulting econometric model imposes fewer restrictions on the cutoff-point distribution, while remaining estimable.

	Identification in peer effects models featuring rational expectations is typically established by imposing a rank condition on the matrix of explanatory variables \citep[see][]{brock2001discrete, lee2014binary, yang2017social, xu2018social}. However, this condition poses challenges for empirical testing because the explanatory variable matrix includes the average of friends' expected outcomes, which is unobservable to the econometrician. I address this problem by deriving easily testable sufficient conditions. Specifically, I show that a key identification condition is that the network includes agents who have friends' friends who are not directly their friends. While this condition has previously been used to identify parameters in linear models \citep{bramoulle2009identification}, I demonstrate its applicability in nonlinear models. Once identification is established, I show that the model parameters can be estimated using the nested pseudo-likelihood (NPL) method of \cite{aguirregabiria2007sequential}. The count outcome is allowed to have a large support.

	I empirically investigate peer effects on the number of extracurricular activities students enroll in. First, I estimate these effects without considering heterogeneity in peer effects. The results indicate that increasing the expected number of activities in which a student's friends are enrolled by 1 implies an increase of 0.08 in the expected number of activities in which the student is enrolled. In contrast, the SAR-Tobit estimate of peer effects is more than four times higher. Subsequently, I introduce heterogeneity in peer effects by gender. I find that female students are more responsive to their peers compared to male students. Using these results, I conduct a counterfactual analysis investigating the effects of school gender composition on student participation in extracurricular activities.
	
	This paper contributes to the large literature on social interaction models \citep[e.g., see][]{manski1993identification, brock2001discrete, lee2014binary, blume2015linear, xu2015maximum, guerra2020multinomial}. I investigate the case of count outcomes using a game of incomplete information. I show that estimating peer effects on these outcomes using the linear SAR or Tobit model is equivalent to estimating an ordered model with evenly spaced thresholds, potentially leading to misspecification. 
	
	The paper contributes to the literature on the identification of peer effects models \citep{bramoulle2009identification, de2010identification,blume2015linear}. Necessary conditions to avoid the reflection problem \citep{manski1993identification} are readily testable for linear-in-means models. However, for nonlinear models with rational expectations, identification relies on rank conditions that pose empirical testing challenges \citep{yang2017social}. I demonstrate that these models can be identified if the network includes friends' friends who are not direct friends. This identification condition is similar to the restriction for linear models.
	
	I also contribute to the literature on the heterogeneity of peer effects. Most papers allowing heterogeneity in peer effects consider a linear-in-means model with a continuous outcome \citep{peng2019heterogeneous, beugnot2019gender, comola2022heterogeneous}. The proposed model in this paper allows for the influence of friends on count outcomes to vary across both agents' and friends' observed groups.

	The remainder of the paper is organized as follows. Section \ref{sec_game} presents the microeconomic foundation of the model. Section \ref{sec_econ_model} addresses the identification and estimation of the model parameters. Section \ref{sec_mc} documents the Monte Carlo experiments. Section \ref{sec_appli} presents the empirical results and Section \ref{sec_conclu} concludes this paper. Proofs of propositions are provided in Appendix \ref{app_proof}.
	
	\section{Microeconomic Foundations}
	\label{sec_game}
	This section presents the model's microfoundations. Let $\mathcal P = \{1,\dots,n\}$ be a population of $n$ agents. The outcome of agent $i \in \mathcal P$ is denoted by $y_i$, representing an integer variable known as a \textit{count variable}. Let $\mathbb{N}_{R} = \{0,1,\dots,R\}$ be the support of $y_i$, where $R$ is a finite strictly positive integer.\footnote{The assumption of a finite $R$ implies a bounded outcome, which differs from outcomes in standard count data models such as the Poisson model. I make this assumption for ease of exposition and treat the case of an infinite $R$ in Online Appendix \ref{append_unbounded}. Specifically, when $R = \infty$, many summations in the paper involve infinitely many terms and care must be taken to ensure that the resulting sums converge.}

	\subsection{Incomplete Information Network Game}
	The model is based on a game of incomplete information. To introduce heterogeneity in peer effects, I partition $\mathcal P$ into $M$ observed groups (e.g., groups defined by gender, ethnicity, or the interaction of gender and ethnicity), denoted $\mathcal G_1$, \dots, $\mathcal G_M$, where $M$ is finite. Let $g_i$ be an observable variable indicating the group of individual $i$, with support $G = \{1,\dots,M\}$. For any $g \in G$, $g_i = g$ means that $i \in \mathcal G_g$. Agents act noncooperatively and interact through a directed network, in which peer influence depends on both their own and their peers’ group memberships.

	\textbf{Network}. The network is characterized by a set of $n \times n$ adjacency matrices,  $\mathcal{A} = \{\mathbf{A}^{gg^{\prime}}: g, g^{\prime} \in G\}$. For any $g, g^{\prime} \in G$,  $\mathbf{A}^{gg^{\prime}} = [a_{ij}^{gg^{\prime}}]$,  where $a_{ij}^{gg^{\prime}} = 1$ if $j$ is a friend of $i$, $g_i = g$, and  $g_j = g^{\prime}$, and $a_{ij}^{gg^{\prime}} = 0$ otherwise. For example, if  $\mathcal G_1$ and $\mathcal G_2$ denote the groups of males and females, then  $a_{ij}^{12} = 1$ if $j$ is a friend of $i$, $i$ is male, and $j$ is female.  Agents do not interact with themselves; i.e., $a_{ii}^{gg^{\prime}} = 0$ for all $i\in\mathcal P$. The network is directed, thus $a_{ij}^{gg^{\prime}}$ may not equal $a_{ji}^{gg^{\prime}}$.\footnote{Later, I assume that the network consists of many independent subnetworks (e.g., schools), such that two agents from different subnetworks have no connection \citep[see][]{calvo2009peer}. However, for notational simplicity, I present the microeconomic model within a single subnetwork.} I introduce heterogeneity in peer effects through agents' groups: agents are influenced by their peers depending on their own and their peers’ groups.

	\textbf{Individual characteristics}. Every agent $i$ is described by a vector of characteristics $(\phi_i,\varepsilon_i)^{\prime}$,  where $\phi_i \in \mathbb R$ is a characteristic observable to all agents and $\varepsilon_i \in \mathbb R$ is a private characteristic observable only to agent $i$. Agent $i$ observes  $\varepsilon_i$ and $\phi_j$ for all $j\in\mathcal P$, but does not observe $\varepsilon_j$ for $j\ne i$.  Following the standard approach in the literature, I assume that the private characteristics $\varepsilon_i$'s are independently and identically distributed (i.i.d.) and that this distribution is common knowledge to all agents \citep[see][]{brock2001discrete}.

	\textbf{Preferences}. Let $\mathbf{W}^{gg^{\prime}} = [w_{ij}^{gg^{\prime}}]$ be the matrix obtained by row-normalizing $\mathbf{A}^{gg^{\prime}}$; i.e., $w_{ij}^{gg^{\prime}} = a_{ij}^{gg^{\prime}}/\sum_{j = 1}^n a_{ij}^{gg^{\prime}}$ if  $a_{ij}^{gg^{\prime}} = 1$ and $w_{ij}^{gg^{\prime}} = 0$ otherwise.\footnote{An alternative way to normalize the network matrix is to define 
		$w_{ij}^{gg^{\prime}} = \frac{a_{ij}^{gg^{\prime}}}{\sum_{g^{\prime\prime}}\sum_{j=1}^n a_{ij}^{gg^{\prime\prime}}}$ 
		if $a_{ij}^{gg^{\prime}} = 1$ \citep[e.g., see][]{comola2022heterogeneous}. 
		However, this approach complicates the interpretation of the model because 
		$\mathbf{W}^{gg^{\prime}} \mathbf{y}$, where $\mathbf{y} = (y_1,\dots,y_n)^{\prime}$, 
		no longer corresponds to the average outcome among friends in group $g^{\prime}$. 
		Consequently, the peer effect parameter ($\alpha^{gg^{\prime}}$ below) cannot be interpreted as the average influence of peers in group $g^{\prime}$. 
		Moreover, having $\alpha^{gg^{\prime}} > \alpha^{gg^{\prime\prime}}$ does not necessarily indicate that peers in group $g^{\prime}$ are more influential, because the influence also depends on the number of friends in each group.} Since agent $i$ does not observe others' private characteristics, they also do not observe their choices. Consequently, agents' decisions depend on their beliefs about the choices of other players. However, given that the distribution of private characteristics is common knowledge, agents form rational expectations; i.e., their beliefs are consistent with the distribution of others' private characteristics \citep[see][]{brock2001discrete, blume2015linear}. Specifically, preferences are described by the following additive discrete payoff function:
	\begin{equation}
		\label{payoff_addex}
		U_i^e(y_i) = (\phi_i + \varepsilon_i)y_i - c_{g_i}(y_i) - \sum_{g^{\prime}\in G}\dfrac{\alpha^{g_ig^{\prime}}}{2}\mathbb{E}\left[(\textstyle y_i - \sum_{j= 1}^n w_{ij}^{g_ig^{\prime}} y_j)^2|\mathcal{A},\boldsymbol \phi\right].
	\end{equation}
	where $\mathbf{y}_{-i} = \left(y_1,y_2,\dots, y_{i-1},y_{i+1},\dots,y_n\right)$ and $\boldsymbol \phi = (\phi_1,\dots,\phi_n)^{\prime}$. Payoff \eqref{payoff_addex} is separable into two components. The first two terms denote the private subpayoff, where $(\phi_i + \varepsilon_i)y_i$ is the benefit and $c_{g_i}(y_i)$ is the private cost. As in \cite{blume2015linear}, the benefit is linear in the agent's own observable characteristic $\phi_i$ and private characteristic $\varepsilon_i$. In the literature, the private cost is generally defined as a quadratic function of $y_i$; i.e., $c_{g_i}(y_i) = \frac{1}{2}y_i^2$ \cite[see][]{ballester2006s, calvo2009peer, blume2015linear}. This is because the game results in an estimable econometric model. Leveraging the counting nature of the outcome, I show that one can relax this restriction. I model cost as a flexible nonparametric function of $y_i$ and agents' groups. I discuss in Remark \ref{rem_general} below that specifying \eqref{payoff_addex} with a quadratic private cost function implies a strong restriction on the econometric model. 
	
	The last term in payoff \eqref{payoff_addex} represents an expected social cost. The expectation is with respect to $\mathbf y_{-i}$, the choices of other players, because these choices are not observed by agent $i$. As the network matrix $\mathbf{W}^{gg^{\prime}}$ is row-normalized, $\sum_{j= 1}^n w_{ij}^{g_ig^{\prime}} y_j$ is the average outcome among friends in $\mathcal G_{g^{\prime}}$. The social cost depends on the gap between agents' choices and the average choices of their friends. The effects of this gap on the preferences can be measured by the peer effect parameter $\alpha^{g_ig^{\prime}}$. There are $M^2$ peer effect parameters: $\alpha^{gg^{\prime}}$ where $g,g^{\prime}\in G$. A higher value for $\alpha^{gg^{\prime}}$ indicates that agents in $\mathcal G_g$ are more responsive to the average choice of friends in $\mathcal G_{g^{\prime}}$. A positive $\alpha^{g_ig^{\prime}}$ suggests pure conformity between agent $i$ and friends in $\mathcal G_{g^{\prime}}$, whereas a negative $\alpha^{g_ig^{\prime}}$ implies substitutability.\footnote{This specification of social cost differs from that implying complementary preferences. The choice between conformist and complementary preferences depends on the outcome being studied. However, it can be shown that both types of preferences yield the same econometric model \citep[see][]{boucher2016some}.} 
	
	For any sequence $b(r)$, I denote by $\Delta b(r)$ its first difference defined by $\Delta b(r) = b(r) - b(r-1)$. I impose the inconsequential restrictions that $c_g(-1) = +\infty$ and $c_g(R + 1) = +\infty$. To show that there exists a unique choice $y_i$ that maximizes the expected payoff \eqref{payoff_addex}, I introduce the following restrictions.
	
	\begin{assumption} \label{ass_cost} The cost function $c_g(.)$ is a strictly convex and strictly increasing function on $\mathbb{N}_{R}$.
	\end{assumption}
	\begin{assumption}\label{ass_pe} For all $g \in G$, $\sum_{g^{\prime}\in G} \alpha^{gg^{\prime}} \geq 0$.
	\end{assumption}
	\begin{assumption}\label{ass_dist_e}
		The private characteristics $\varepsilon_i$'s are i.i.d. with a continuous symmetric distribution. Moreover, $\varepsilon_i$'s are independent of $\{\mathcal{A},\boldsymbol \phi,g_1, \dots,g_n\}$.
	\end{assumption}
	
	\noindent The strict convexity condition in Assumption \ref{ass_cost} implies strictly increasing differences in costs: $\Delta c_{g_i}(r+1) - \Delta c_{g_i}(r) > 0$. This nests the commonly used linear-quadratic payoff function specification. Assumption \ref{ass_pe} requires the overall effects of peers across all groups to be nonnegative. Although the cost function is strictly convex, large negative peer effects can lead to a non-concave payoff, making the game difficult to solve. Assumption \ref{ass_pe} still accommodates scenarios where friends in some groups exert negative effects, provided these effects are offset by positive effects from other groups. Assumption \ref{ass_dist_e} restricts correlation in agents' decisions to occur only through the network and observable characteristics, a standard condition in network models \citep[e.g.,][]{li2009binary, lee2014binary, yang2017social}. The concavity of the payoff and the continuity of the type distribution ensure that the payoff function is maximized at a single point almost surely (a.s.).

	\begin{proposition} \label{prop:xtrema}  Under Assumptions \ref{ass_cost}--\ref{ass_dist_e}, $U_i^e(.)$ has a unique maximizer, $y_i^{\ast} \in \mathbb{N}_R$, almost surely (a.s.). Moreover, $y_i^{\ast} = r$  if and only if 
		$U_i^e\left(r\right) > \max\left\{ U_i^e\left(r- 1\right), U_i^e\left(r + 1\right)\right\}$.
	\end{proposition}

	\noindent The condition $U_i^e\left(r\right) > \max\left\{ U_i^e\left(r- 1\right), U_i^e\left(r + 1\right)\right\}$ is also equivalent to: $$
	\sum_{g^{\prime}\in G}\alpha^{g_ig^{\prime}} \boldsymbol{w}_{i}^{g_ig^{\prime}} \mathbb{E}(\mathbf{y}|\mathcal{A},\boldsymbol{\phi})+ \phi_i - \gamma_{g_i}(r+ 1) < -\varepsilon_i < \sum_{g^{\prime}\in G}\alpha^{g_ig^{\prime}} \boldsymbol{w}_{i}^{g_ig^{\prime}} \mathbb{E}(\mathbf{y}|\mathcal{A},\boldsymbol{\phi})+ \phi_{i} - \gamma_{g_i}(r),$$ where $\boldsymbol{w}_{i}^{g_ig^{\prime}} = (w_{i1}^{g_ig^{\prime}},\dots,w_{in}^{g_ig^{\prime}})$ is the $i$-th row of $\mathbf{W}^{g_ig^{\prime}}$,  $\mathbf{y} = (y_1,\dots,y_n)^{\prime}$, and $\textstyle \gamma_{g_i}(r) = \Delta c_{g_i}(r) +  \sum_{g^{\prime}\in G}(r - 1/2)\alpha^{g_ig^{\prime}}$. Given that $\varepsilon_i$ has a symmetric distribution, Proposition \ref{prop:xtrema} implies that the probability of $\{y_i = r\}$, conditional on  $\mathcal{A}$ and $\boldsymbol{\phi}$, can be expressed as: 
	\begin{align}
		\begin{split}
			\textstyle \mathbb{P}(y_i = r|\mathcal{A},\boldsymbol{\phi}) & \textstyle= F_{\varepsilon}\big(\sum_{g^{\prime}\in G}\alpha^{g_ig^{\prime}} \boldsymbol{w}_{i}^{g_ig^{\prime}} \mathbb{E}(\mathbf{y}|\mathcal{A},\boldsymbol{\phi})+ \phi_{i} - \gamma_{g_i}(r)\big) -  \\
			& \textstyle \quad~ F_{\varepsilon}\big(\sum_{g^{\prime}\in G}\alpha^{g_ig^{\prime}} \boldsymbol{w}_{i}^{g_ig^{\prime}} \mathbb{E}(\mathbf{y}|\mathcal{A},\boldsymbol{\phi})+ \phi_i - \gamma_{g_i}(r+ 1)\big),
		\end{split}\label{REE_nomatrix}
	\end{align}
	where $F_{\varepsilon}$ is the cumulative distribution function (cdf) of $\varepsilon_i$. Equation \eqref{REE_nomatrix} is similar to an ordered variable model with social interactions \citep[see][]{liu2017social}. The probability $\mathbb{P}(y_i = r|\mathcal{A},\boldsymbol{\phi})$ is a function of $\boldsymbol{w}_{i}^{g_ig^{\prime}} \mathbb{E}(\mathbf{y}|\mathcal{A},\boldsymbol{\phi})$ for all $g^{\prime}\in G$, which are expectations of the average outcomes among $i$'s friends in each group, conditional on $\mathcal{A}$ and $\boldsymbol{\phi}$. The distance between the thresholds (cutoff points), $\textstyle \gamma_{g_i}(r) = \Delta c_{g_i}(r) +  \sum_{g^{\prime}\in G}(r - 1/2)\alpha^{g_ig^{\prime}}$, depends closely on the cost function $c_{g_i}$. 
	A key distinction between count data models and ordered models is that, for count data, the magnitude of the outcome matters. The values taken by the outcome influence the choice probability \eqref{REE_nomatrix} through the term $\mathbb{E}(\mathbf{y}|\mathcal{A},\boldsymbol{\phi})+ \phi_{i} - \gamma_{g_i}(r)\big)$, whereas in ordered models, these values are irrelevant; the focus is instead on the relative ranking of the outcomes. 
	
	\begin{remark}\label{rem_general}
		As $\textstyle \gamma_{g_i}(r) = \Delta c_{g_i}(r) +  \sum_{g^{\prime}\in G}(r - 1/2)\alpha^{g_ig^{\prime}}$, a quadratic cost function implies that $\gamma_{g_i}(r + 1) - \gamma_{g_i}(r)$ is constant in $r$; that is, the cut points $\gamma_{g}(r)$ are evenly spaced over $r$. This restriction may yield inconsistent estimators if the cut points are not truly evenly spaced. The assumption of evenly spaced cut points can be tested using a likelihood ratio test by treating the model with evenly spaced cut points as the constrained model. Spatial autoregressive (SAR) and SAR-Tobit models are based on a game similar to that described by payoff \eqref{payoff_addex} with a quadratic cost function \citep[see][]{ballester2006s, xu2015maximum}. Heuristically, applying these models to count data would yield estimates similar to those from the proposed model with evenly spaced cut points.
	\end{remark}

	\subsection{Bayesian Nash Equilibrium}
	Under rational expectations, the belief of any agent $j \ne i$ about $\{y_i = r\}$ equals the \textit{true} conditional probability $\mathbb{P}(y_i = r | \mathcal{A}, \boldsymbol{\phi})$, defined as a function of the expected outcome in Equation \eqref{REE_nomatrix}. The conditional expected outcome of agent $j$ can be expressed as $\mathbb{E}(y_j|\mathcal{A},\boldsymbol{\phi}) =  \sum_{t=1}^R t \mathbb{P}(y_j = t|\mathcal{A},\boldsymbol{\phi})$. Substituting this into \eqref{REE_nomatrix} expresses $\mathbb{P}(y_i = r|\mathcal{A},\boldsymbol{\phi})$ as a function of $\mathbb{P}(y_j = t|\mathcal{A},\boldsymbol{\phi})$, where $j\ne i$ and $t = 1,\dots,R$. This yields a fixed-point equation in $\mathbf{p} =\big \{\mathbb{P}(y_i = t|\mathcal{A},\boldsymbol{\phi}): ~i\in\mathcal P, t = 1,\dots,R\big\}$, which is an $nR$-dimensional vector. This fixed-point equation may have no, one, or multiple solutions.\footnote{When $R = 1$, $\mathbb{E}(y_j|\mathcal{A},\boldsymbol{\phi}) = \mathbb{P}(y_j = 1|\mathcal{A},\boldsymbol{\phi})$ and Equation \eqref{REE_nomatrix} implies that $ \mathbb{P}(y_i = 1|\mathcal{A},\boldsymbol{\phi}) = F_{\varepsilon}\big(\sum_{g^{\prime}\in G}\alpha^{g_ig^{\prime}} \sum_{j = 1}^n w_{ij}^{g_ig^{\prime}} \mathbb{P}(y_j = 1|\mathcal{A},\boldsymbol{\phi})+ \phi_{i} - \gamma_{g_i}(r)\big)$, which generalizes the fixed-point equation of the binary model with a single group \citep[see][]{lee2014binary}.} In this section, I establish conditions ensuring the existence of a unique belief system (outcome distribution) coherent with Equation \eqref{REE_nomatrix}.
	
	The fixed-point equation resulting from \eqref{REE_nomatrix} is defined in $\mathbb R^{nR}$. Solving it may be challenging when $R$ is large. To address this issue, I focus on the conditional expected outcome $ \mathbb{E}(\mathbf{y}|\mathcal{A},\boldsymbol{\phi})$. Equation \eqref{REE_nomatrix} implies that the knowledge of $\mathbb{E}(\mathbf{y}|\mathcal{A},\boldsymbol{\phi})$  is sufficient to compute the underlying belief vector $\mathbf{p}$. Moreover, since $\mathbb{E}(y_i|\mathcal{A},\boldsymbol{\phi}) = \sum_{t = 1}^R t \mathbb{P}(y_i = t|\mathcal{A},\boldsymbol{\phi})$, the knowledge of $\mathbf{p}$ is  sufficient to compute  $\mathbb{E}(\mathbf{y}|\mathcal{A},\boldsymbol{\phi})$. 
	I show that the rational expected outcome also verifies a fixed-point equation.
	\begin{proposition} 
		\label{prop:expby}
		Let $\mathbf{p} =\big \{\mathbb{P}(y_i = t|\mathcal{A},\boldsymbol{\phi}): ~i\in\mathcal P, t = 1,\dots,R\big\}$ be a belief system and $\mathbb{E}(y_i|\mathcal{A},\boldsymbol{\phi}) = \sum_{t = 1}^R t \mathbb{P}(y_i = t|\mathcal{A},\boldsymbol{\phi})$ be the associated expected outcome. If $\mathbf{p}$ and $\mathbb{E}(\mathbf{y}|\mathcal{A},\boldsymbol{\phi})$ are rational; that is, if they satisfy Equation  \eqref{REE_nomatrix}, then 
		\begin{equation}
			\textstyle \mathbb{E}(y_i|\mathcal{A},\boldsymbol{\phi}) = \sum_{t = 1}^{R}F_{\varepsilon}\big(\sum_{g^{\prime}\in G}\alpha^{g_ig^{\prime}} \boldsymbol{w}_{i}^{g_ig^{\prime}} \mathbb{E}(\mathbf{y}|\mathcal{A},\boldsymbol{\phi})+ \phi_{i} - \gamma_{g_i}(t)\big). \label{REEY:nomatrix}
		\end{equation}
	\end{proposition}
	
	\noindent Equation \eqref{REEY:nomatrix} expresses agents’ expected outcomes as a function of the average expected outcomes of their peers. If $\alpha^{g_ig^{\prime}} > 0$, an agent’s expected outcome is increasing in the expected outcome of peers in $\mathcal G_{g^{\prime}}$. Conversely, $\alpha^{g_ig^{\prime}} < 0$ implies that the expected outcomes of agents decrease when those of peers in $\mathcal G_{g^{\prime}}$ increase. If $\alpha^{g_ig^{\prime}} = 0$, then agents are not responsive to peers in $\mathcal G_{g^{\prime}}$. 
	
	To demonstrate that a unique rational belief system is consistent with Equation \eqref{REE_nomatrix}, it is sufficient to show that a unique expected outcome, $\mathbb{E}(\mathbf{y}|\mathcal{A}, \boldsymbol{\phi})$, solves \eqref{REEY:nomatrix}. To achieve this result, I show that the right-hand side of Equation \eqref{REEY:nomatrix} forms a contraction mapping by imposing the following assumption.
	\begin{assumption}
		\label{ass_eqcond}
		For any $g\in G$, $\sum_{g^{\prime}\in G}\alpha^{gg^{\prime}}  <\big(\max_{u \in \mathbb{R}} \sum_{t = 1}^{R}  f_{\varepsilon}\left(u - \gamma_g(t)\right)\big)^{-1}$, where $f_{\varepsilon}$ is the probability density function (pdf) of $\varepsilon_i$.
	\end{assumption}
	\noindent  
	The problem of multiple rational expectations equilibria generally arises when peer effects are strong. Assumption \ref{ass_eqcond} is equivalent to assuming the overall marginal effect of friends' average expected outcome on agents' expected outcome is less than one. Indeed, from Equation \eqref{REEY:nomatrix}, the marginal effect of the average expected outcome among friends in $\mathcal G_{g^{\prime}}$ on the expected outcome of agent $i$ is: 
	$$\textstyle\dfrac{\partial \mathbb{E}(y_i|\mathcal{A},\boldsymbol{\phi})}{\partial  \Bar y_i^{e,g^{\prime}}} =  \alpha^{g_ig^{\prime}}\sum_{t = 1}^{R}f_{\varepsilon}\big(\sum_{\hat g\in G}\alpha^{g_i\hat g}  \Bar y_i^{e,\hat g}+ \phi_{i} - \gamma_{g_i}(t)\big),$$
	where $ \Bar y_i^{e,g^{\prime}} = \boldsymbol{w}_{i}^{g_ig^{\prime}} \mathbb{E}(\mathbf{y}|\mathcal{A},\boldsymbol{\phi})$. The total marginal effect of the average expected outcome among friends in all groups is thus $\sum_{g^{\prime}\in G}\alpha^{g_ig^{\prime}}\sum_{t = 1}^{R}f_{\varepsilon}\big(\sum_{\hat g\in G}\alpha^{g_i\hat g}  \Bar y_i^{e,\hat g}+ \phi_{i} - \gamma_{g_i}(t)\big)$, which is less than one by Assumption \ref{ass_eqcond}. The interpretation of Assumption \ref{ass_eqcond} is that agents do not increase their expected choice beyond the increase in the average expected choice of their peers, \textit{ceteris paribus}. This condition is standard in peer effect models and is verified in many applications \citep{lee2014binary, blume2015linear}. 
	
	\begin{proposition} 
		\label{prop:eunique} 
		Under Assumptions \ref{ass_cost}--\ref{ass_eqcond}, the game of incomplete information with payoff \eqref{payoff_addex} has a unique Bayesian Nash equilibrium (BNE) given by $\mathbf{y}^{\ast} = (y_1^{\ast}, \dots, y_n^{\ast})^{\prime}$ and a unique rational expected outcome $\mathbf y^{e\ast} = (y_1^{e\ast}, \dots, y_n^{e\ast})^{\prime}$, where $y_i^{\ast}$ is the maximizer of the expected payoff $U_i^e(.)$ and $\mathbf y^{e\ast}$ verifies Equation \eqref{REEY:nomatrix}.
	\end{proposition}
	
	\noindent The uniqueness of the BNE is a direct implication of Proposition \ref{prop:xtrema}. Moreover, by the contraction mapping theorem, Assumption \ref{ass_eqcond} ensures a unique expected outcome in Equation \eqref{REEY:nomatrix}, which in turn implies a unique rational belief system $\big\{\mathbb{P}(y_i = t | \mathcal{A}, \boldsymbol{\phi}): i \in \mathcal{P}, t = 1, \dots, R\big\}$ via Equation \eqref{REE_nomatrix}. 
	
	The key condition for the uniqueness is Assumption \ref{ass_eqcond}. This assumption may be violated under strong peer effects, leading to multiple equilibria. The literature offers various solutions to address this issue when the outcome is binary and the model includes a single peer effect parameter \citep[see][]{bajari2010estimating}. Most solutions require enumerating equilibria and using a selection rule for estimation. The problem becomes more complex when $R$ is large or unbounded, as the solution of \eqref{REEY:nomatrix} may diverge, as in the case of SAR models. The econometric approach presented in the next section focuses on the case where Assumption \ref{ass_eqcond} holds and the BNE is unique. I leave the issue of multiple rational expected equilibria for future research.

	\section{Econometric Model \label{sec_econ_model}}
	In this section, I present the econometric model, study parameter identification, and propose an estimation strategy. To handle dependence across agents' strategies, I assume, as is common, that the network consists of $S$ independent subnetworks (e.g., schools) that do not interact \cite[see][]{bramoulle2009identification}. This allows one to define observational equivalence at the subnetwork level (see Definition \ref{def_equivalence}), and therefore observing many subnetworks is theoretically useful \cite[see][]{rothenberg1971identification}. Let $\mathcal{P}_s$ denote the set of agents in the $s$-th subnetwork and $n_s$ the number of agents in $\mathcal{P}_s$. Asymptotically, $S$ goes to infinity, while $n_s$ is bounded.\footnote{Observational equivalence can also be defined at the sample level \citep[e.g.,][]{liu2017social, yang2017social}. However, additional conditions on within-network dependence are required in this case to ensure valid inference. Assuming many independent subnetworks provides a straightforward way to handle this dependence.} Throughout the paper, vectors or matrices defined at the sample level (e.g., $\mathbf y$, $\mathcal A$) carry the subscript $s$ when restricted to subnetwork $s$. However, the notation for individual variables (e.g., $y_i$, $g_i$, $\phi_i$) will not include the subscript $s$ unless necessary for clarity.

	\subsection{Econometric Model \label{sec_econ_model_spec}}
	Individual characteristics are commonly modeled as  $\phi_i = \beta_0 + \boldsymbol{x}_i^{\prime}\boldsymbol{\beta}_1 + \Bar{\boldsymbol{x}}_i^{\prime}\boldsymbol{\beta}_2$ for any $i\in \mathcal{P}_s$, where $\boldsymbol{x}_i \in \mathbb R^K$ is a random vector of observed characteristics (e.g., sex and age),  $\Bar{\boldsymbol{x}}_i = \sum_{j \in \mathcal{P}_s} w_{ij} \boldsymbol{x}_j$ is the average of these characteristics among peers, $\beta_0\in\mathbb R$, and $\boldsymbol{\beta}_1,\boldsymbol{\beta}_2\in\mathbb R^K$ \citep[see][]{calvo2009peer}. The parameter $\boldsymbol{\beta}_1$ captures the effects of own characteristics, whereas $\boldsymbol{\beta}_2$ measures the effects of average characteristics among friends (contextual effects). For simplicity, I do not introduce group-based heterogeneity in the contextual effects. Additionally, the specification of $\phi_i$ does not account for unobserved subnetwork heterogeneity because this poses non-identification problems when $S$ is large \citep[see][]{brock2007identification}.\footnote{Unless $n_s$ is sufficiently large, in which case the incidental-parameter problem does not lead to biased estimates \citep[e.g.,][]{lee2014binary, guerra2020multinomial}.}
	Let $\mathbf{X} = [\boldsymbol{x}_1 ~ \dots ~ \boldsymbol{x}_n]^{\prime}$, $\mathbf{1}_n = (1, \dots, 1)^{\prime}\in\mathbb{R}^n$, $\mathbf{Z} = \left[\mathbf{1}_n,\mathbf{X},\mathbf{W}\mathbf{X}\right]$ and $\boldsymbol{\beta} = \left(\beta_0,\boldsymbol{\beta}_1^{\prime},\boldsymbol{\beta}_2^{\prime}\right)^{\prime}$. For any $i\in\mathcal P_s$, Equation \eqref{REE_nomatrix} can be expressed as:
	\begingroup
	\allowdisplaybreaks
	\begin{align}
		\begin{split}
			\textstyle \mathbb{P}(y_i = r|\mathcal{A}_s,\mathbf Z_s) & \textstyle= F_{\varepsilon}\big(\sum_{g^{\prime}\in G}\alpha^{g_ig^{\prime}} \boldsymbol{w}_{s,i}^{g_ig^{\prime}} \mathbb{E}(\mathbf{y}_s|\mathcal{A}_s,\mathbf Z_s)+ \boldsymbol{z}_i^{\prime}\boldsymbol{\beta} - \gamma_{g_i}(r)\big) -  \\
			& \textstyle \quad~ F_{\varepsilon}\big(\sum_{g^{\prime}\in G}\alpha^{g_ig^{\prime}} \boldsymbol{w}_{s,i}^{g_ig^{\prime}} \mathbb{E}(\mathbf{y}_s|\mathcal{A}_s,\mathbf Z_s)+ \boldsymbol{z}_i^{\prime}\boldsymbol{\beta} - \gamma_{g_i}(r+ 1)\big),
		\end{split}\label{REE_nomatrix_recall}
	\end{align}
	\endgroup
	\noindent  where $\boldsymbol{w}_{s,i}$ is the row corresponding to agent $i$ in $\mathbf W_s$ and $\boldsymbol{z}_i = (1, \boldsymbol x_i^{\prime},\boldsymbol{w}_{s,i} \mathbf X_s)^{\prime}$. Similarly, the fixed-point equation \eqref{REEY:nomatrix} can also be written as:
	\begin{equation}
		\textstyle \mathbb{E}(y_i|\mathcal{A}_s,\mathbf Z_s) = \sum_{t = 1}^{R}F_{\varepsilon}\big(\sum_{g^{\prime}\in G}\alpha^{g_ig^{\prime}} \boldsymbol{w}_{s,i}^{g_ig^{\prime}} \mathbb{E}(\mathbf{y}_s|\mathcal{A}_s,\mathbf Z_s)+ \boldsymbol{z}_i^{\prime}\boldsymbol{\beta} - \gamma_{g_i}(t)\big),\label{REEY_nomatrix_recall}
	\end{equation}
	where $\textstyle \gamma_g(r) = \Delta c_g(r) +  \sum_{g^{\prime}\in G}(r - 1/2)\alpha^{gg^{\prime}}$,  for all $g$. Thus, $\gamma_g(0) = -\infty$, $\gamma_g(R + 1) = +\infty$, and $\gamma_g(r + 1) - \gamma_g(r) > \sum_{g^{\prime}\in G}\alpha^{gg^{\prime}}$.\footnote{In Online Appendix \ref{oa_comparison}, I show that the model does not impose equidispersion, as the Poisson model does, and is more flexible than the negative binomial model due to the large number of cut points.} 
	
	\begin{remark}\label{rem_semi}
		For certain applications, the distribution of the outcome may exhibit a long tail with a large (or unbounded) $R$, yielding a large number of parameters to be estimated. Yet, most observations of $y$ can have comparatively low values. See, for instance, my empirical application in Section \ref{sec_appli}, where $R = 33$ but 99\% of observations of $y_i$ fall below 10. In such cases, it could be challenging to estimate $\gamma_g(r)$ for large $r$. This issue can be addressed by using a semiparametric cost function in payoff \eqref{payoff_addex}. For some integer $\bar R_c\geq 1$, $c_g(r)$ can be defined as a nonparametric function if $r \leq \bar R_c$. However, when $r > \bar R_c$, a quadratic specification (or another strictly convex function) can be employed. Assuming that $c_g(r)$ is quadratic when $r > \bar R_c$ implies that $\gamma_g(r + 1) - \gamma_g(r) = \bar \gamma_g$ for all $r \geq \bar R_c$, where $\bar\gamma_g> \sum_{g^{\prime}\in G}\alpha^{gg^{\prime}}$ and does not depend on $r$ but varies across $g\in G$. This representation of the cost reduces the number of parameters to be estimated while not being overly restrictive if $\bar R_c$ is properly chosen. The model can be estimated over a range of values of $\bar R_c$ (e.g., using a grid of integers starting at $\bar R_c = 1$). The optimal $\bar R_c$ can be determined by minimizing information criteria such as the Bayesian Information Criterion (BIC). When $\bar R_c = 1$, the cost function becomes fully quadratic, resulting in estimates that are similar to those produced by a SAR-Tobit model.
	\end{remark}
	
	\subsection{Identification \label{sec_econ_model_ident}}
	The free objects to be identified in the econometric model are the peer effects $\alpha^{gg^{\prime}}: ~g,g^{\prime}\in G$, the control variable parameter $\boldsymbol\beta$, the thresholds $\gamma_g(t):~g\in G, t = 1,\dots, R$, and the distribution function of agents' private characteristic, $F_{\varepsilon}$. In this section, I assume that $F_{\varepsilon}$ is known and present a general analysis that addresses the identification of $F_{\varepsilon}$ in Online Appendix \ref{SM_Ident_F}. As the intercept $\beta_0$ and the thresholds enter Equations \eqref{REE_nomatrix_recall} only through their difference, they cannot be identified separately. I thus fix the first threshold, as in ordered models: $\gamma_g(1) = 0$ for all $g\in G$.
	
	Let $\boldsymbol{\alpha} = (\alpha^{gg^{\prime}}: ~g,g^{\prime}\in G)^{\prime}$ be the vector of peer effects, $\boldsymbol{\gamma} = (\gamma_g(t):~g\in G, t = 2,\dots,R)^{\prime}$ be the vector of thresholds to be estimated, and $\boldsymbol\theta = (\boldsymbol{\alpha}^{\prime}, \boldsymbol{\beta}^{\prime},\boldsymbol{\gamma}^{\prime})^{\prime}$ be the vector that comprises all parameters. The parameter $\boldsymbol\theta$ is identified if any $\tilde{\boldsymbol{\theta}} \ne \boldsymbol\theta$  is not observationally equivalent to $\boldsymbol\theta$.  Observational equivalence is defined as follows.
	\begin{definition}\label{def_equivalence}
		Two parameters $\boldsymbol\theta$  and $\tilde{\boldsymbol{\theta}}$ are observationally equivalent at $\mathcal{A}_s$ and $\mathbf Z_s$ if 
		$p(\mathbf y_s|\mathcal{A}_s,\mathbf Z_s) = \tilde p(\mathbf y_s|\mathcal{A}_s,\mathbf Z_s)$ for any $s$, where $p(\mathbf y_s|\mathcal{A}_s,\mathbf Z_s)$ and $\tilde p(\mathbf y_s|\mathcal{A}_s,\mathbf Z_s)$ are the distributions of $\mathbf y_s$ conditional on $\mathcal{A}_s$ and $\mathbf Z_s$  at $\boldsymbol\theta$ and $\tilde{\boldsymbol{\theta}}$, respectively.
	\end{definition}
	
	\noindent Let $\bar{\mathbf{Y}}_s = [\mathbf{W}^{gg^{\prime}}_s\mathbf y_s: ~g,g^{\prime}\in G]$  be the matrix of average peers' outcomes. For all $i\in\mathcal{P}_s$, let $\bar{\mathcal{Y}}_i$ be the row of agent $i$ in $\bar{\mathbf{Y}}_s$ and $\tilde{\boldsymbol z}_i = [\mathbb E(\bar{\mathcal{Y}}_i|\mathcal{A}_s,\mathbf Z_s),\boldsymbol{z}_i^{\prime}]^{\prime}$. The condition $p(\mathbf y_s|\mathcal{A}_s,\mathbf Z_s) = \tilde p(\mathbf y_s|\mathcal{A}_s,\mathbf Z_s)$ implies that both distributions $p(\mathbf y_s|\mathcal{A}_s,\mathbf Z_s)$ and $ \tilde p(\mathbf y_s|\mathcal{A}_s,\mathbf Z_s)$ yield the same conditional expected outcome $\mathbb E(\mathbf y_s|\mathcal{A}_s,\mathbf Z_s)$. Consequently, for $\tilde{\boldsymbol{\theta}} \ne \boldsymbol\theta$  not to be observationally equivalent to $\boldsymbol\theta$, it is necessary that the matrix of explanatory variables, $[\mathbb E(\bar{\mathbf{Y}}|\mathcal{A},\mathbf Z),\mathbf Z]$, be full rank. Specifically, I impose the following identification restrictions.
	
	\begin{assumption}\label{ass_fullrank}
		The matrix $\plim \left(\frac{1}{S}\sum_{i = 1}^n \tilde{\boldsymbol z}_i^{\prime} \tilde{\boldsymbol z}_i\right)$ is full rank.\footnote{The $\plim$ symbol denotes the limit in probability as $S$ grows to infinity.}
	\end{assumption}
	
	\begin{assumption} \label{ass_density_e} $\varepsilon_i$ has  a positive pdf almost everywhere on $\mathbb{R}$.
	\end{assumption}
	
	\noindent Assumption \ref{ass_fullrank} is classical in peer effect models with rational expectations \citep[e.g., see][]{brock2007identification, yang2017social,xu2018social}. However, since $\mathbb E(\bar{\mathbf{Y}}_s|\mathcal{A}_s,\mathbf Z_s)$ is not observable to the econometrician, testing Assumption \ref{ass_fullrank} in practice may pose challenges. To address this issue, I establish sufficient conditions for this assumption in Proposition \ref{prop_ident:suff} below. Assumption \ref{ass_density_e} ensures that there is a positive probability that the outcome takes any value in $\mathbb {N}_R$ \cite[see][]{manski1988identification}. Recall that Assumption \ref{ass_dist_e} imposes that $\varepsilon_i$'s are independent of $\{\mathcal{A}, \boldsymbol \phi, g_1, \dots,g_n\}$, which means that $\varepsilon_i$ and $\mathbf X_s$ are independent. This suggests that there is no omission of important regressors in $\mathbf{X}_s$ that can be captured by~$\varepsilon_i$. 
	
	\begin{proposition}
		\label{prop_ident}
		Under Assumptions \ref{ass_cost}--\ref{ass_density_e}, $\boldsymbol \theta$ is identified.
	\end{proposition}
	
	As Assumption \ref{ass_fullrank} is not easily testable, I find sufficient conditions for the cases where $M \leq 2$. These conditions do not involve the unobserved expected outcome and can then be easily tested in practice. Specifically, I demonstrate that the presence of friends' friends who are not directly friends in the network can help identify the model’s parameters. This argument has so far been used for addressing the reflection problem in linear models \citep[see][]{bramoulle2009identification}. I show that it can be applied to a large class of models, where a rank condition, such as Assumption \ref{ass_fullrank}, ensures identification.
	
	Let $\boldsymbol\pi_{i} = (\pi_{i}^{1},\dots,\pi_{i}^{M})^{\prime}$ be an $M$-vector, where $\pi_{i}^{g}$ is a dummy variable that equals one if the following two conditions hold: (1) agent $i$ has friends in group $\mathcal G_{g}$; (2) agent $i$  has friends' friends (in any group) who are not directly their friends. For models without heterogeneity in the peer effects ($M = 1$), $\boldsymbol\pi_{i}$ is simply a dummy variable that equals one if $i$ has friends' friends who are not friends.

	\begin{proposition}\label{prop_ident:suff} Assume that $M \leq 2$ and $\alpha^{gg^{\prime}} \geq 0$ for all $g,g^{\prime}\in G$. Assumption \ref{ass_fullrank} is satisfied if the following conditions hold:\\
		\begin{inparaenum}[\bfseries A.] 
			\item The matrix $\plim \left(\frac{1}{S}\sum_{i=1}^n \boldsymbol{z}_i\boldsymbol{z}_i^{\prime}\right)$ is full rank;  \label{cond_frank}\\
			\item The matrix $\plim \left(\frac{1}{S}\sum_{i=1}^n \boldsymbol\pi_{i}^{\prime} \boldsymbol\pi_{i}\right)$ is full rank;  \label{cond_ffriend}\\
			\item There is a variable $x_{i,\bar \kappa}$ in $\boldsymbol{x}_i$ with $\beta_{1,\bar \kappa} \beta_{2,\bar \kappa} \geq 0$ and $\beta_{2,\bar \kappa}\ne0$, where $\beta_{1,\bar \kappa}$ is the coefficient of $x_{i,\bar \kappa}$ and $\beta_{2,\bar \kappa}$ is the coefficient of $\bar x_{i,\bar \kappa}$, the contextual variable associated with $x_{i,\bar \kappa}$.\label{cond_contextual}
		\end{inparaenum}
	\end{proposition}

	\noindent Proposition \ref{prop_ident:suff} is established by assuming that the peer effect parameters $\alpha^{gg^{\prime}}$ are all nonnegative. I will later discuss cases of negative peer effects. Condition \ref{cond_frank} imposes that the matrix of observed regressors is of full rank. Condition \ref{cond_ffriend} implies that the network includes agents who have friends' friends who are not their friends. Moreover, when $M = 2$, certain of these agents should not have friends in all groups. For example, for gender-based groups, some males who have friends of friends who are not their own friends should have only male or only female friends. Condition \ref{cond_contextual} ensures that the expected outcome is influenced by the contextual variable $\bar x_{i,\bar \kappa}$. The direct effect of changes in friends' exogenous characteristics on an agent’s outcome is given by $\beta_{2,\bar \kappa}$, while the effect on friends’ outcomes is $\beta_{1,\bar \kappa}$. When these parameters share the same sign, the resulting indirect effects through peer interactions reinforce the direct effect, yielding a nonzero total effect on the agent’s outcome.\footnote{A similar restriction is also set by \cite{bramoulle2009identification} in linear models with a single peer effect parameter $\alpha$. In their case, the condition is $\alpha \beta_{1,\bar \kappa} + \beta_{2,\bar \kappa}  \ne 0$.}
	
	Formal proof of Proposition \ref{prop_ident:suff} is presented in Appendix \ref{append_Ident_Suf}. Below, I provide an intuition for the case of a single group. 
	I consider a network of three agents $i_1$, $i_2$, and $i_3$, where $i_3$ is a friend's friend of $i_1$ but not a direct friend (see Figure \ref{fig_ident}). If $\mathbb E(\bar{\mathbf{Y}}_s|\mathcal{A}_s,\mathbf Z_s)$  and $\mathbf{Z}_s$ are not linearly independent, then for all $i$ in any subpopulation $\mathcal P_s$, I can write:
	\begin{equation}\label{eq_barylin}
		\mathbb E(\bar{\mathcal{Y}}_i|\mathcal{A}_s,\mathbf Z_s) =  \check\beta_0 + \boldsymbol{x}_i^{\prime}\check{\boldsymbol{\beta}}_1 + \Bar{\boldsymbol{x}}_i^{\prime}\check{\boldsymbol{\beta}}_2,
	\end{equation}
	where $\check\beta_0\in\mathbb R$, $\check{\boldsymbol{\beta}}_1\in\mathbb R^K$, and $\check{\boldsymbol{\beta}}_2\in\mathbb R^K$ are constants and $\bar{\mathcal{Y}}_i = \boldsymbol w_{s,i} \mathbf{y}_s$ is the average outcome among friends, with $\boldsymbol w_{s,i}$ being the row corresponding to agent $i$ in $\mathbf W_s$. As $i_2$ is the only friend of $i_1$, $\bar{\mathcal{Y}}_{i_1} = y_{i_2}$ and $\Bar{\boldsymbol{x}}_{i_1} = \boldsymbol{x}_{i_2}$. For $i = i_1$,  Equation \eqref{eq_barylin} becomes $\mathbb E(y_{i_2}|\mathcal{A}_s,\mathbf Z_s) =  \check\beta_0 + \boldsymbol{x}_{i_1}^{\prime}\check{\boldsymbol{\beta}}_1 + \boldsymbol{x}_{i_2}^{\prime}\check{\boldsymbol{\beta}}_2$, which means that the expected outcome of $i_2$ is not influenced by $x_{i_3,\bar \kappa}$. 
	This condition cannot hold because $i_3$ is a friend of $i_2$. As implied by Condition \ref{cond_contextual}, the marginal effect of  $x_{i_3,\bar \kappa}$  on $\mathbb E(y_{i_2}|\mathcal{A}_s,\mathbf Z_s)$ cannot be zero. As a result, Equation \eqref{eq_barylin} cannot hold for all agents when many subnetworks include agents who have friends' friends who are not directly friends.
	
	\begin{figure}[!ht]
		\centering
		\footnotesize
		\begin{tikzpicture}[scale=0.6]
			\tikzstyle{node}=[circle,draw,fill=yellow!80,text=blue, thick,minimum height=0.1cm]
			\tikzstyle{link}=[->,>=latex, color=black!80, thick]
			
			\node[node] (i3) at (0,0) {$i_3$};
			\node[node] (i2) at (3,0) {$i_2$};
			\node[node] (i1) at (3,3) {$i_1$};
			
			\draw[link] (i1)--(i2);
			\draw[link] (i2)--(i3);
		\end{tikzpicture}
		\caption{Illustration of the identification}
		\label{fig_ident}
	\end{figure}
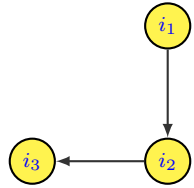
	
	The positive peer effects assumed in Proposition \ref{prop_ident:suff} ensure that the overall effect of $x_{i_3,\bar \kappa}$ on $\mathbb{E}(y_{i_2} | \mathcal{A}_s, \mathbf{Z}_s)$ is nonzero. Increasing $x_{i_3,\bar \kappa}$ generates direct contextual effects since $i_3$ is a friend of $i_2$. However, it may also produce indirect effects through $\mathbb{E}(y_{i_3} | \mathcal{A}_s, \mathbf{Z}_s)$ via peer interactions. If peer effects are negative, direct and indirect effects may have opposite signs, but the overall effect is unlikely to cancel out exactly. Thus, Proposition \ref{prop_ident:suff} extends to negative peer effects as long as direct and indirect effects do not perfectly offset. With many subnetworks, the perfect offset of both effects is highly improbable.

	\subsection{Estimation \label{sec_econ_model_estim}}
	This section presents an approach for estimating the model parameters. Given that the rational expected outcome depends on the cdf $F_{\varepsilon}$, estimating the parameters without specifying this cdf is challenging \citep[e.g., see][]{brock2001discrete, brock2002multinomial, lee2014binary, yang2017social, guerra2020multinomial}.\footnote{An exception where semi-parametric estimation without specifying $F_{\varepsilon}$ is possible is when there is a finite number of agents (e.g., 2 or 3 firms) that are observed multiple times. In this case, the asymptotic analysis focuses on the number of times the game is repeated rather than the number of agents \cite[e.g., see][]{aradillas2010semiparametric, wan2014semiparametric}.} I assume that $\varepsilon_i$ follows a standard normal distribution. The distribution's variance is set to 1 as an identification restriction, as is the case for ordered-response models.
	
	Given $\boldsymbol{\theta}$, Equations \eqref{REE_nomatrix_recall} and \eqref{REEY_nomatrix_recall} can be employed to compute $\mathbb E(\mathbf y_s|\mathcal{A}_s,\mathbf Z_s)$ and $\mathbb{P}(y_i = r|\mathcal{A}_s,\mathbf Z_s)$ for all $r\in\mathbb N_R$. This suggests using the maximum likelihood (ML) approach to estimate $\boldsymbol{\theta}$. However, the estimation approach may be computationally cumbersome for large samples, because one needs to solve the fixed-point equation \eqref{REEY_nomatrix_recall} in $\mathbb R^n$ for every guess of $\boldsymbol{\theta}$. To address this issue, I rely on the nested pseudo-likelihood (NPL) proposed by \cite{aguirregabiria2007sequential}. This method eliminates the need to solve a fixed-point problem.  For any $\boldsymbol u \in \mathbb R^n$, let us consider the pseudo-likelihood function defined as:
	\begin{equation}
		\textstyle\mathcal{L}(\boldsymbol{\theta},\boldsymbol u) = \dfrac{1}{S}\sum_{i = 1}^n\sum_{t = 0}^{R} \mathbbm{1}\{y_i = t\}\log\big(p_{it}(\boldsymbol{\theta},\boldsymbol u_{s(i)})\big), \label{pseudo:llh}
	\end{equation}
	where
	$ p_{it}(\boldsymbol{\theta},\boldsymbol u_{s(i)}) = F_{\varepsilon}\big(\sum_{g^{\prime}\in G}\alpha^{g_ig^{\prime}} \boldsymbol{w}_{s,i}^{g_ig^{\prime}} \boldsymbol u_{s(i)}+ \boldsymbol{z}_i^{\prime}\boldsymbol{\beta} - \gamma_{g_i}(t)\big) -  F_{\varepsilon}\big(\sum_{g^{\prime}\in G}\alpha^{g_ig^{\prime}} \boldsymbol{w}_{s,i}^{g_ig^{\prime}} \boldsymbol u_{s(i)}+ \boldsymbol{z}_i^{\prime}\boldsymbol{\beta} - \gamma_{g_i}(t+ 1)\big)$, $\boldsymbol{u}_{s(i)}$ denotes the vector containing the entries of $\boldsymbol{u}$ associated with the agents in the same subnetwork as $i$, and $\mathbbm{1}\{\cdot\}$ denotes the indicator function. The pseudo-likelihood function $\mathcal{L}(\boldsymbol{\theta},\boldsymbol u)$ depends on an arbitrary vector $\boldsymbol u$ and not on the true expected outcome $\mathbb E(\mathbf y_s|\mathcal{A}_s,\mathbf Z)$. 
	
	Let  $\ell_i(\boldsymbol{\theta},\boldsymbol u_{s(i)}) = \sum_{t = 1}^{R}F_{\varepsilon}\big(\sum_{g^{\prime}\in G}\alpha^{g_ig^{\prime}} \boldsymbol{w}_{s,i}^{g_ig^{\prime}} \boldsymbol u_{s(i)}+ \boldsymbol{z}_i^{\prime}\boldsymbol{\beta} - \gamma_{g_i}(r)\big)$ and $\mathbf L(\boldsymbol{\theta}, \boldsymbol u) = (\ell_1(\boldsymbol{\theta}, \boldsymbol u_{s(1)}),\break \dots, \ell_n(\boldsymbol{\theta}, \boldsymbol u_{s(n)}))^{\prime}$.
	To describe the NPL algorithm, I also define the operators $\textstyle \tilde{\boldsymbol{\theta}}(\boldsymbol u) = \arg \max_{\boldsymbol{\theta}} \mathcal{L}(\boldsymbol{\theta},\boldsymbol u)$ and $\boldsymbol{\phi}(\boldsymbol u) = \mathbf{L}(\tilde{\boldsymbol{\theta}}(\boldsymbol u), \boldsymbol u)$. The NPL algorithm starts with a guess $\boldsymbol u^{0}$ for $\boldsymbol u$ and constructs the sequence of estimators $\mathcal{Q}_k = \{\boldsymbol{\theta}^{k}, \boldsymbol u^{k}\}$, for $k = 1,2,\dots$, where $\boldsymbol{\theta}^{k} = \tilde{\boldsymbol{\theta}}(\boldsymbol u^{k-1})$ and $\boldsymbol u^{k} = \boldsymbol{\phi}(\boldsymbol u^{k-1})$. Specifically, given the guess $\boldsymbol u^{0}$, $\boldsymbol{\theta}^{1} = \tilde{\boldsymbol{\theta}}(\boldsymbol u^{0})$ and $\boldsymbol u^{1} = \boldsymbol{\phi}(\boldsymbol u^{0})$, then  $\boldsymbol{\theta}^{2} = \tilde{\boldsymbol{\theta}}(\boldsymbol u^{1})$ and  $\boldsymbol u^{2} = \boldsymbol{\phi}(\boldsymbol u^{1})$, and so forth.
	Note that each value of $\mathcal{Q}_k$ requires evaluating the mapping $\mathbf{L}$ only once. If $\mathcal{Q}_{k}$  converges, regardless of the initial guess $\boldsymbol u^0$, its limit, which is denoted by $\{\hat{\boldsymbol{\theta}},\hat{\mathbb{E}}(\mathbf{y}|\mathcal{A},\mathbf Z)\}$, is the NPL estimator. This limit satisfies the following two properties: $\hat{\boldsymbol{\theta}}$ maximizes the pseudo-likelihood $\mathcal{L}(\boldsymbol{\theta},\hat{\mathbb{E}}(\mathbf{y}|\mathcal{A},\mathbf Z))$; and $\hat{\mathbb{E}}(\mathbf{y}|\mathcal{A},\mathbf Z) = \mathbf{L}(\hat{\boldsymbol{\theta}},\hat{\mathbb{E}}(\mathbf{y}|\mathcal{A},\mathbf Z))$. 
	
	The NPL algorithm is based on an iterative process to determine both the values of $\boldsymbol{\theta}$ and $\boldsymbol u$ that maximize \eqref{pseudo:llh}. The resulting  estimator $\{\hat{\boldsymbol{\theta}},\hat{\mathbb{E}}(\mathbf{y}|\mathcal{A},\mathbf Z)\}$ is such that $\hat{\boldsymbol{\theta}}$ maximizes the pseudo-likelihood $\mathcal{L}(\boldsymbol{\theta},\hat{\mathbb{E}}(\mathbf{y}|\mathcal{A},\mathbf Z))$; and $\hat{\mathbb{E}}(\mathbf{y}|\mathcal{A},\mathbf Z) = \mathbf{L}(\hat{\boldsymbol{\theta}},\hat{\mathbb{E}}(\mathbf{y}|\mathcal{A},\mathbf Z))$. As demonstrated by \cite{kasahara2012sequential}, a key determinant of the convergence of the NPL algorithm is the contraction property of the fixed point mapping defined by the right-hand side of Equation \eqref{REEY:nomatrix}. Specifically, this mapping is defined for all $\boldsymbol u = (u_1, \dots, u_n)^{\prime} \in \mathbb{R}^{n}$ as $\mathbf{L}^{\ast}(\boldsymbol u) = (\ell_{1}^{\ast}(\boldsymbol u_{s(1)}),\dots,\ell_{n}^{\ast}(\boldsymbol u_{s(n)}))^{\prime}$, with
	$\textstyle\ell_{i}^{\ast}(\boldsymbol u_{s(i)}) = \sum_{t = 1}^{R}F_{\varepsilon}\big(\sum_{g^{\prime}\in G}\alpha^{g_ig^{\prime}} \boldsymbol{w}_{s,i}^{g_ig^{\prime}} \boldsymbol u_{s(i)}+ \phi_{i} - \gamma_{g_i}(t)\big)$. As shown in Proposition \ref{prop:eunique}, this mapping is contracting under Assumption \ref{ass_eqcond}, which guarantees that it has a unique fixed point.
	
	The following proposition establishes the consistency and asymptotic distribution of $\hat{\boldsymbol{\theta}}_n$.

	\begin{proposition}\label{prop:asymp}
		Assume that Assumptions \ref{ass_cost}--\ref{ass_density_e} and the regularity conditions of the NPL estimator hold \citep[see][]{aguirregabiria2007sequential}. The NPL estimator
		$\hat{\boldsymbol{\theta}}$  is consistent and $\sqrt{S}(\hat{\boldsymbol{\theta}}  - \boldsymbol{\theta}_0) \overset{d}{\to} \mathcal{N}\big(0,\mathbb{V}_0(\hat{\boldsymbol{\theta}})\big)$, where $\boldsymbol{\theta}_0$ is the true value of $\boldsymbol{\theta}$ and the asymptotic variance  $\mathbb{V}_0(\hat{\boldsymbol{\theta}})$ is given in Appendix \ref{append:asymptotic}.
	\end{proposition}
	
	\section{Monte Carlo Experiments \label{sec_mc}}
	In this section, I conduct a Monte Carlo study to assess the finite sample performance of the model and the NPL estimator. I consider four data-generating processes (DGPs), denoted as DGPs A--D. For all DGPs, I assume that $\phi_i = 2 + 1.5 x_{1i} - 1.2 x_{2i} + 0.5 \bar{x}_{1i} - 0.9 \bar{x}_{2i}$, where $(x_{1i},x_{2i}) = \boldsymbol{x}_i'$ and $(\bar{x}_{1i},\bar{x}_{2i}) = \bar{\boldsymbol{x}}_i'$. I simulate $x_{1i}$ from a Uniform distribution over $[0,5]$ and $x_{2i}$ from a Poisson distribution with parameter 2. Following the network structure of my empirical application, I assume that each agent $i$ has $\bar n_i$ randomly selected friends, where $\bar n_i$ is chosen randomly between 0 and 10.

	I set $R = 100$, although most observed $y_i$ values are relatively low. Assuming a large $R$ allows for replicating outcomes with distributions exhibiting long tails, similar to those observed in survey data. The DGPs differ in how peer effects and cutoff points (cost function) are specified. For DGPs A and B, there is a single peer effect parameter $\alpha = 0.25$ (no heterogeneity in the peer effects). DGP A considers a quadratic cost function, implying that the cutoff points are equally spaced. I set $\gamma(r+1)-\gamma(r) = 0.55$ for $r \geq 1$. For DGP B, the cost function is semiparametric. I assume that $(\gamma(2), \dots, \gamma(13)) =  (2.050,1.250, 0.850,0.700,0.500,0.400,0.330,0.300,0.290,  ~0.280, 0.270,0.260)$ and $\gamma(r+1)-\gamma(r) = 0.255$ for $r \geq 13$, indicating that the cost is quadratic when $r \geq 13$.

	In DGPs C--D, I introduce heterogeneity in peer effects. I consider two groups: $G = \{1,2\}$, where $g_i = 1$ if $x_{1i} \leq 2.5$ and $g_i = 2$ otherwise. The peer effect parameters are defined as $\boldsymbol\alpha = (0.3,0.15,0.1,0.15)^{\prime}$ for DGP C and $\boldsymbol\alpha = (0.4,-0.1,0.2,  ~0.1)$ for DGP D. In DGP D, peers in $\mathcal G_2$ have negative effects on agents in $\mathcal G_1$. In both $\mathcal G_1$ and $\mathcal G_2$, I define the cost function as in DGP B.

	The parameter values are motivated by my empirical application so that the resulting outcome distributions resemble those observed in the application. Yet, the distribution is quite different for DGP A because the cost function in my application is likely non-quadratic. Figure \ref{mc_data_plot} displays examples of histograms of simulated data for $S=8$ subpopulations, where each subpopulation comprises $n_s = 250$ agents ($n=\text{2,000}$). DGP A exhibits a symmetric distribution left-truncated at zero, whereas the distributions of the other DGPs exhibit long tails.

	\begin{figure}[!ht]
		\centering
		\includegraphics[scale = .8]{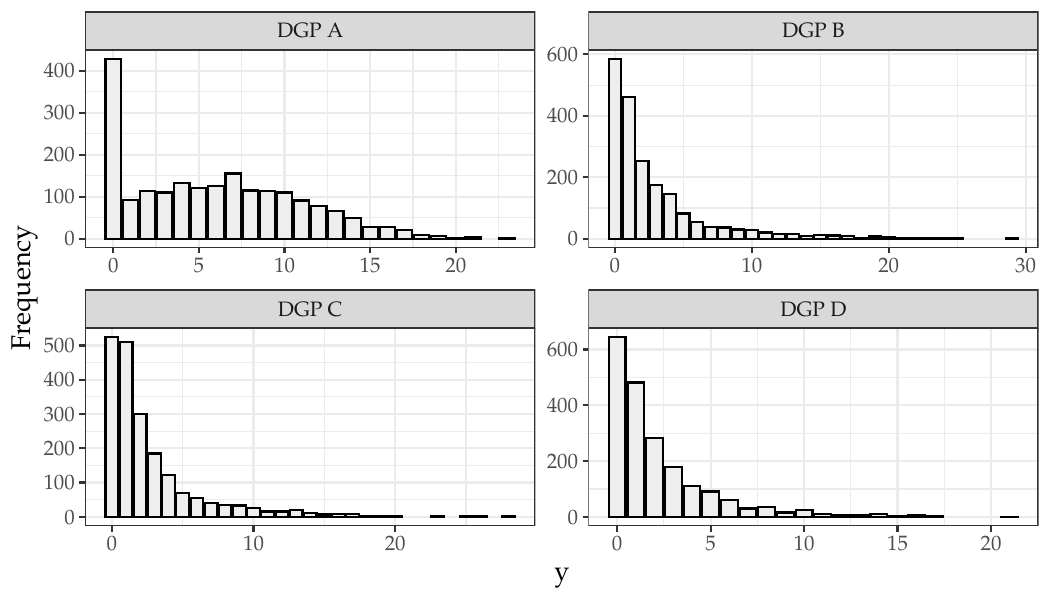}
		
		\vspace{-.4cm}
		\caption{Simulated data using the count data model with social interactions \label{mc_data_plot}}
	\end{figure}

	Table \ref{tab_mc} summarizes the simulation results for 1,000 replications. The number of subpopulations $S \in \{2, 8\}$ and each subpopulation comprises $n_s = 250$ agents; i.e., $n \in \{\text{500, 2,000}\}$. Since the values of the parameters cannot be directly interpreted, I focus on the average \textit{direct} marginal effects (DMEs).\footnote{In peer effects models, changes in exogenous characteristics directly affect the expected outcome while keeping the game equilibrium fixed, as shown in Equation \eqref{REEY_nomatrix_recall}. This is referred to as the direct marginal effect (DME). However, changes in exogenous characteristics also influence the contextual variable (since friends’ characteristics change) and the game’s equilibrium. Consequently, the expected outcome undergoes additional changes, referred to as indirect marginal effects (IME). In this section, I focus only on DMEs, while both effects are presented for the empirical application in Section \ref{sec_appli}.} For DGPs without heterogeneity in peer effects, I compare the estimated direct marginal effects from my model to those provided by a SAR-Tobit model. I estimate the models using a grid of values starting from $\bar R_c = 1$ and select the optimal $\bar R_c$ that minimizes the BIC (see Remark \ref{rem_semi}).

	For DGP A, both semiparametric and quadratic cost estimators perform well. This is because the cost is quadratic in the data. Additionally, the estimated marginal effects from a SAR-Tobit model are similar to those from the count data model. This result aligns with the discussion in Remark \ref{rem_general}, which emphasizes that estimating peer effects on count data using a SAR-Tobit model is similar to estimating an ordered model with equally spaced thresholds.
	
	In the case of DGP B, the quadratic-cost and SAR-Tobit approaches yield similar marginal effects. Yet, these effects are biased upward by an average of 40\% for the peer effect parameters.  This bias occurs because the cost is not quadratic for this DGP. In contrast, the semiparametric cost estimates are accurate. Increasing the sample size from 500 to 2,000 improves the precision of the marginal effect estimates, bringing them closer to the true marginal effects.
	
	Furthermore, the count data model with the semiparametric cost approach performs well when peer effects are heterogeneous. The estimated marginal effects are valid for DGPs C and D. Conversely, the quadratic cost approach biases the marginal effects for both peer effects and control variables. The bias is positive for peer effects on agents in $\mathcal G_1$ and negative for peer effects on agents in $\mathcal G_2$. This occurs because agents in $\mathcal G_1$ have a lower $x_{1i}$ and, thus, a lower expected outcome. The quadratic cost approach tends to overestimate peer effects on agents with lower expected outcomes and to underestimate them on agents with higher expected outcomes.

	\begin{table}[!ht]
		\centering 
		\footnotesize
		\caption{Monte Carlo simulations} 
		\label{tab_mc} 
		\resizebox{\columnwidth}{!}{\begin{threeparttable}
				\begin{tabular}{ld{3}d{3}d{3}d{3}ld{3}d{3}d{3}d{3}ld{3}d{3}d{3}d{3}}
					\toprule
					& \multicolumn{4}{c}{Semiparametric cost}                     &  & \multicolumn{4}{c}{Quadratic cost}                          &  & \multicolumn{4}{c}{Tobit}                                   \\
					& \multicolumn{2}{c}{n = 500} & \multicolumn{2}{c}{n = 2,000} &  & \multicolumn{2}{c}{n = 500} & \multicolumn{2}{c}{n = 2,000} &  & \multicolumn{2}{c}{n = 500} & \multicolumn{2}{c}{n = 2,000} \\
					True Effect               & \multicolumn{1}{c}{Mean}          & \multicolumn{1}{c}{Sd.}         & \multicolumn{1}{c}{Mean}          & \multicolumn{1}{c}{Sd.}          &  & \multicolumn{1}{c}{Mean}          & \multicolumn{1}{c}{Sd.}         & \multicolumn{1}{c}{Mean}          & \multicolumn{1}{c}{Sd.}          &  & \multicolumn{1}{c}{Mean}          & \multicolumn{1}{c}{Sd.}         & \multicolumn{1}{c}{Mean}          & \multicolumn{1}{c}{Sd.}          \\\midrule
					\textbf{DGP A}                  & \multicolumn{4}{l}{}                                        &  & \multicolumn{4}{l}{}                                        &  & \multicolumn{4}{l}{}                                        \\
					PE $= 0.431$                     & 0.431        & 0.041        & 0.430         & 0.020         &           & 0.431        & 0.041        & 0.430         & 0.020         &           & 0.432         & 0.041         & 0.431             & 0.020     \\
					$\boldsymbol x_1 = 2.584$        & 2.583        & 0.061        & 2.584         & 0.031         &           & 2.583        & 0.061        & 2.584         & 0.031         &           & 2.590         & 0.061         & 2.590             & 0.031     \\
					$\boldsymbol x_2 = -2.067$       & -2.066       & 0.061        & -2.067        & 0.030         &           & -2.066       & 0.061        & -2.067        & 0.030         &           & -2.076        & 0.061         & -2.076            & 0.030     \\
					$\bar{\boldsymbol x}_1 = 0.861$  & 0.864        & 0.189        & 0.868         & 0.093         &           & 0.864        & 0.189        & 0.868         & 0.093         &           & 0.866         & 0.190         & 0.870             & 0.093     \\
					$\bar{\boldsymbol x}_2 = -1.550$ & -1.551       & 0.120        & -1.555        & 0.057         &           & -1.551       & 0.120        & -1.555        & 0.057         &           & -1.552        & 0.121         & -1.555            & 0.058  \Bstrut\\\hline 
					\textbf{DGP B}                   & \multicolumn{4}{l}{}                                        &  & \multicolumn{4}{l}{}                                        &  & \multicolumn{4}{l}{}                                        \Tstrut \\
					PE $= 0.265$                     & 0.263        & 0.047        & 0.265         & 0.023         &           & 0.358        & 0.107        & 0.374         & 0.057         &           & 0.356         & 0.106         & 0.371             & 0.055     \\
					$\boldsymbol x_1 = 1.589$        & 1.586        & 0.138        & 1.588         & 0.068         &           & 1.393        & 0.113        & 1.394         & 0.054         &           & 1.402         & 0.113         & 1.404             & 0.062     \\
					$\boldsymbol x_2 = -1.271$       & -1.269       & 0.115        & -1.270        & 0.056         &           & -1.164       & 0.103        & -1.164        & 0.049         &           & -1.170        & 0.104         & -1.171            & 0.057     \\
					$\bar{\boldsymbol x}_1 = 0.530$  & 0.533        & 0.104        & 0.528         & 0.052         &           & 0.384        & 0.178        & 0.364         & 0.095         &           & 0.384         & 0.183         & 0.366             & 0.099     \\
					$\bar{\boldsymbol x}_2 = -0.953$ & -0.961       & 0.114        & -0.953        & 0.057         &           & -0.879       & 0.139        & -0.865        & 0.073         &           & -0.879        & 0.140         & -0.866            & 0.072   \Bstrut\\\hline
					\textbf{DGP C}                   & \multicolumn{4}{l}{}                                        &  & \multicolumn{4}{l}{}                                        &  & \multicolumn{4}{l}{}                                        \Tstrut \\
					PE$^{1,1} = 0.115$               & 0.112        & 0.020        & 0.114         & 0.013         &           & 0.144        & 0.026        & 0.146         & 0.012         &           &               &               &                   &           \\
					PE$^{1,2} = 0.058$               & 0.057        & 0.008        & 0.058         & 0.004         &           & 0.072        & 0.009        & 0.073         & 0.004         &           &               &               &                   &           \\
					PE$^{2,1} = 0.135$               & 0.120        & 0.047        & 0.127         & 0.032         &           & 0.009        & 0.080        & 0.013         & 0.040         &           &               &               &                   &           \\
					PE$^{2,2} = 0.202$               & 0.200        & 0.023        & 0.202         & 0.011         &           & 0.167        & 0.028        & 0.173         & 0.014         &           &               &               &                   &           \\
					$\boldsymbol x_1 = 2.596$        & 2.568        & 0.202        & 2.590         & 0.098         &           & 1.974        & 0.172        & 1.985         & 0.083         &           &               &               &                   &           \\
					$\boldsymbol x_2 = -2.077$       & -2.059       & 0.153        & -2.073        & 0.073         &           & -1.839       & 0.140        & -1.845        & 0.068         &           &               &               &                   &           \\
					$\bar{\boldsymbol x}_1 = 0.865$  & 0.863        & 0.134        & 0.866         & 0.068         &           & 0.779        & 0.161        & 0.771         & 0.078         &           &               &               &                   &           \\
					$\bar{\boldsymbol x}_2 = -1.557$ & -1.547       & 0.151        & -1.557        & 0.076         &           & -1.499       & 0.170        & -1.509        & 0.087         &           &               &               &                   &       \Bstrut\\\cline{1-10}
					\textbf{DGP D}                   & \multicolumn{4}{l}{}                                        &  & \multicolumn{4}{l}{}                                        &  & \multicolumn{4}{l}{}                                        \Tstrut \\
					PE$^{1,1} =0.111$                & 0.108        & 0.024        & 0.109         & 0.012         &           & 0.167        & 0.032        & 0.168         & 0.021         &           &               &               &                   &           \\
					PE$^{1,2} = -0.028$              & -0.028       & 0.007        & -0.028        & 0.003         &           & -0.011       & 0.009        & -0.012        & 0.006         &           &               &               &                   &           \\
					PE$^{2,1} = 0.201$               & 0.186        & 0.058        & 0.194         & 0.032         &           & 0.070        & 0.093        & 0.066         & 0.059         &           &               &               &                   &           \\
					PE$^{2,2} = 0.100$               & 0.098        & 0.021        & 0.099         & 0.011         &           & 0.052        & 0.027        & 0.053         & 0.015         &           &               &               &                   &           \\
					$\boldsymbol x_1 = 1.920$        & 1.915        & 0.136        & 1.918         & 0.064         &           & 1.428        & 0.109        & 1.427         & 0.056         &           &               &               &                   &           \\
					$\boldsymbol x_2 = -1.536$       & -1.532       & 0.099        & -1.532        & 0.049         &           & -1.378       & 0.092        & -1.371        & 0.044         &           &               &               &                   &           \\
					$\bar{\boldsymbol x}_1 =0.640$   & 0.648        & 0.098        & 0.644         & 0.049         &           & 0.624        & 0.126        & 0.623         & 0.064         &           &               &               &                   &           \\
					$\bar{\boldsymbol x}_2 = -1.152$ & -1.159       & 0.109        & -1.154        & 0.052         &           & -1.099       & 0.122        & -1.092        & 0.059         &           &               &               &                   &          \\\bottomrule
				\end{tabular}
				\begin{tablenotes}[para,flushleft]
					\footnotesize
					Notes: The column \textit{True Effect} indicates the average DME at the true parameter values. PE refers to peer effects, whereas PE$^{gg^{\prime}}$ is the effect of peers in group $\mathcal G_{g^{\prime}}$ on agents in group $\mathcal G_g$. The columns \textit{Semiparametric cost} and \textit{Quadratic cost} indicate the means and standard deviations (Sd.) of the estimated marginal effects corresponding to the optimal $\bar R_c$ and $\bar R_c = 1$, respectively. The column \textit{Tobit} refers to the estimates using the SAR-Tobit approach.
				\end{tablenotes}
		\end{threeparttable}}
	\end{table}

	\section{Empirical Application \label{sec_appli}}
	In this section, I present an empirical illustration of the model. I estimate peer effects on the number of extracurricular activities students are enrolled in using the National Longitudinal Study of Adolescent to Adult Health (Add Health) data.
	
	\subsection{Data \label{sec_appli_data}}
	The Add Health data provide nationally representative information on \nth{7}--\nth{12} graders in the United States (US). I use the Wave I in-school data, which were collected between September 1994 and April 1995. The surveyed sample comprises 80 high schools and 52 middle schools. The data provide information on students' social and demographic characteristics, including their friendship links (best friends, up to 5 females and up to 5 males), education level, and parents' occupation, etc. The network is restricted at the school level. Two students from different schools are not friends.

	After cleaning the dataset by removing observations with missing values, the sample used in this empirical study comprises $n = \text{72,291}$ students from $S = 120$ schools.\footnote{Many of the nominated friend identifiers in the dataset are missing. Numerous papers have developed methods for estimating peer effects using partial network data \citep[e.g.,][]{boucher2020estimating}. To focus on the main purpose of this paper, I do not address that issue here.} The largest school has 2,156 students, and about 50\% of the schools have more than 500 students. The average number of friends per student is 3.8 (1.8 male and 2.0 female friends). 
	
	The studied count variable is the number of extracurricular activities in which students are enrolled. Students were presented with a list of 33 clubs, organizations, and teams that are found in many schools. The students were asked to identify any of these activities in which they participated during the current school year or in which they planned to participate later in the school year ($R = 33$).\footnote{Since the cut points are group-dependent, estimating an ordered model here would require estimating $64$ additional parameters. Moreover, some cut points cannot be estimated because fewer than 1\% of observations exceed 10. In such a case, the semiparametric cost function is a suitable alternative.} I study the influence of social interactions on the number of extracurricular activities identified by the students. Figure \ref{fig_outcome}  displays a histogram of the number of extracurricular activities, which varies from 0 to 33. The distribution has a long tail as in DGPs B--D in the simulation study, and more than 99\% of observed values fall below 10. This feature suggests that the thresholds would not be evenly spaced. Therefore, peer effect estimates on these data using the SAR-Tobit model would be biased, as is the case for DGPs B, C, and D in the simulation study.
	
	\begin{figure}[!h]
		\centering
		\includegraphics[scale = 0.75]{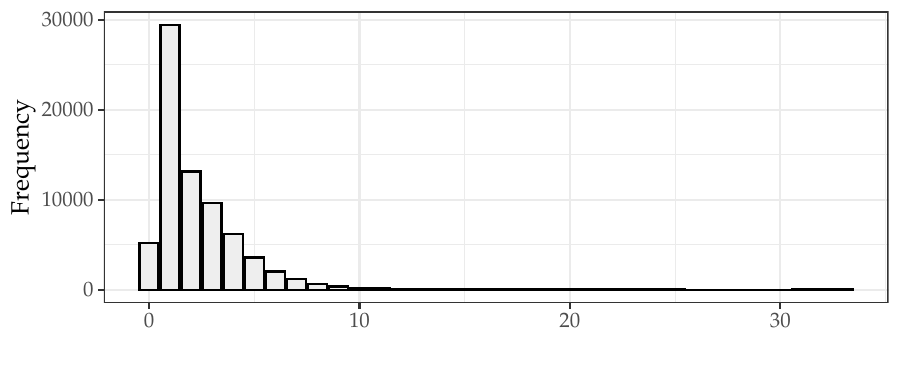}
		
		\vspace{-0.5cm}
		\caption{Distribution of the number of extracurricular activities \label{fig_outcome}}
	\end{figure}

	I introduce gender-based heterogeneity in the peer effects. Specifically, I estimate the influence of both male and female friends separately for male and female students. On average, students participate in 2.4 extracurricular activities, with females participating in 2.5 and males in 2.2. I control for several demographic and socioeconomic factors and contextual variables associated with these factors (see Table \ref{data_sum} in Online Appendix \ref{oa_data} for a data summary).

	\subsection{Empirical Results \label{sec_appli_model}}
	Table \ref{app_res} summarizes the estimated average marginal effects.\footnote{Comprehensive results, including all parameter estimates, are provided in Online Appendix \ref{oa_fullresults}.} Models 1--6 do not account for heterogeneity in peer effects. Model 1 is based on a semiparametric cost function, Model 2 imposes a quadratic cost function, and Model 3 is a SAR-Tobit model. Models 4--6 correspond to the school fixed-effect versions of Models 1--3, respectively. These fixed effects are included as dummy variables to control for unobserved school characteristics, such as pupil-teacher ratio and school climate.\footnote{Given the large sample size ($n = 72{,}291$) and relatively few schools $(S = 120)$, including school fixed effects as dummy variables does not raise an incidental parameter problem \citep[see][]{lee2014binary, guerra2020multinomial}.} Models 7–8 correspond to count data models with gender-dependent peer effects and a semiparametric cost function. Model 7 assumes a semiparametric cost function, whereas Model 8 considers a quadratic cost function.

	The estimation results for Model 1 indicate that a one-unit increase in the expected number of activities in which a student’s friends are enrolled increases the expected number of activities in which the student is enrolled by 0.082. However, the quadratic cost approach (Model 2) overestimates the marginal peer effect at 0.543. As in the simulation study, this model leads to the same estimates as the SAR-Tobit approach.

	Controlling for school fixed effects does not significantly change the peer effect estimate for the nonparametric cost approach (Model 4). Yet, the increase in the log-likelihood (for 119 additional regressors) indicates that unobserved school characteristics influence student participation. In contrast, controlling for school-fixed effects partially addresses the bias of the quadratic cost and SAR-Tobit approaches (Models 5 and 6). The marginal peer effects decrease by more than 30\%, but remain 4.6 times higher than the estimate from the nonparametric approach.\footnote{To assess model performance in replicating observed data, I generate extracurricular activities using the estimated coefficients as parameter values (as in the Monte Carlo simulations). The quadratic cost function (Model 5) fails to match the observed distribution, whereas predictions based on semiparametric cost functions provide a closer fit (see Online Appendix \ref{oa_replication}).}
	
	Furthermore, the estimation results reveal gender-based heterogeneity in peer effects, with same-sex students being less responsive to each other. In Model~7, the marginal effects of female peers on males are $0.071$, similar to the marginal effects of male peers on females ($0.069$). However, the marginal effects of female peers on females are estimated at $0.034$, while those of male peers on males are $-0.022$. The negative effects of male peers on males are perhaps surprising and suggest a decrease in males' participation when the participation of their male friends increases, holding the participation of female friends fixed. Moreover, the results show that females are more responsive and also exert a stronger influence on their friends. Indeed, the total marginal peer effects (from both males and females) on females are $0.091$, but are only half of those ($0.044$) on males. Similarly, the peer effects of females on their friends (males and females) are $0.096$, while those of males on their friends are $0.039$.

	In Model~8, the quadratic cost specification overestimates the peer effects in each group. The total marginal peer effects (from both males and females) on males are $0.187$, which is on average twice as large as the effects estimated using the semiparametric cost specification. Additionally, the total marginal peer effects (from both males and females) on females are $0.116$, which is more than twice the effects estimated using the semiparametric cost specification.

	\begin{table}[!ht]
		\centering 
		\footnotesize
		\caption{Empirical results (marginal effects)} 
		\label{app_res} 
		\resizebox{\columnwidth}{!}{\begin{threeparttable}
				\begin{tabular}{ld{7}d{7}cd{7}d{7}cd{7}d{7}c}
					\toprule
					& \multicolumn{6}{c}{\textbf{School fixed effects:   No, Heterogeneity: No}}                                                                                                       \\[.5ex]
					& \multicolumn{2}{c}{\textit{Model 1:   Semiparametric cost}} & \multicolumn{2}{c}{\textit{Model 2: Quadratic   cost}}   & \multicolumn{2}{c}{\textit{Model 3: Tobit}}             \\[.5ex]
					Endogenous Peer Effects                                   & 0.082                        & (0.026)                      & 0.543                      & (0.028)                     & 0.553                      & (0.017)                    \\
					Male (direct)                                             & -0.275                       & (0.018)                      & -0.195                     & (0.017)                     & -0.200                     & (0.010)                    \\
					Male (contextual direct)                                  & -0.025                       & (0.029)                      & 0.087                      & (0.030)                     & 0.086                      & (0.016)                    \\
					Male (indirect)                                           & -0.046                       & (0.025)                      & -0.034                     & (0.041)                     & -0.045                     & (0.026)                    \\
					Male (total)                                              & -0.321                       & (0.025)                      & -0.229                     & (0.037)                     & -0.245                     & (0.027)                    \\[.5ex]
					$\bar R_c$                                                & \multicolumn{2}{c}{13}                                      & \multicolumn{2}{c}{1}                                    & \multicolumn{2}{c}{}                                    \\
					Log likelihood                                            & \multicolumn{2}{c}{$-127,626$}                              & \multicolumn{2}{c}{$-159,924$}                           & \multicolumn{2}{c}{$-161,225$}                          \\ \midrule
					& \multicolumn{6}{c}{\textbf{School fixed effects: Yes, Heterogeneity: No}}                                                                                                        \\[.5ex]
					& \multicolumn{2}{c}{\textit{Model 4:   Semiparametric cost}} & \multicolumn{2}{c}{\textit{Model 5: Quadratic   cost}}   & \multicolumn{2}{c}{\textit{Model 6: Tobit}}             \\ [.5ex]
					Endogenous Peer Effects                                   & 0.079                        & (0.023)                      & 0.369                      & (0.029)                     & 0.362                      & (0.020)                    \\
					Male (direct)                                             & -0.296                       & (0.019)                      & -0.208                     & (0.017)                     & -0.213                     & (0.011)                    \\
					Male (contextual direct)                                  & -0.059                       & (0.028)                      & 0.029                      & (0.029)                     & 0.021                      & (0.020)                    \\
					Male (indirect)                                           & -0.080                       & (0.026)                      & -0.080                     & (0.026)                     & -0.071                     & (0.021)                    \\
					Male (total)                                              & -0.376                       & (0.030)                      & -0.376                     & (0.030)                     & -0.284                     & (0.021)                    \\[.5ex]
					$\bar R_c$                                                & \multicolumn{2}{c}{13}                                      & \multicolumn{2}{c}{1}                                    & \multicolumn{2}{c}{}                                    \\
					Log likelihood                                            & \multicolumn{2}{c}{$-126,111$}                              & \multicolumn{2}{c}{$-158,964$}                           & \multicolumn{2}{c}{$-160,258$}                          \\ \midrule
					& \multicolumn{6}{c}{\textbf{School fixed effects: Yes, Heterogeneity: Yes}}                                                                                \\[.5ex]
					& \multicolumn{2}{l}{\textit{Model 7:  Semiparametric cost}} & \multicolumn{2}{l}{\textit{Model 8: Quadratic cost}} &                      &                      \\
					Endogenous Peer Effects$^{\text{Male},\text{Male}}$     & -0.022                       & (0.007)                     & 0.033                    & (0.008)                   &                      &                      \\
					Endogenous Peer Effects$^{\text{Male},\text{Female}}$   & 0.066                        & (0.008)                     & 0.083                    & (0.007)                   &                      &                      \\
					Endogenous Peer Effects$^{\text{Female},\text{Male}}$   & 0.061                        & (0.006)                     & 0.066                    & (0.006)                   &                      &                      \\
					Endogenous Peer Effects$^{\text{Female},\text{Female}}$ & 0.030                        & (0.008)                     & 0.121                    & (0.008)                   &                      &                      \\
					Male (direct)                                           & -0.417                       & (0.042)                     & -0.224                   & (0.034)                   &                      &                      \\
					Male (indirect)                                         & -0.034                       & (0.044)                     & 0.089                    & (0.056)                   &                      &                      \\
					Male (total)                                            & -0.451                       & (0.061)                     & -0.135                   & (0.064)                   &                      &                      \\
					Male (contextual direct)                                & 0.011                        & (0.045)                     & 0.142                    & (0.053)                   &                      &                      \\[.5ex]
					$\bar R_c$                                              & \multicolumn{2}{c}{10}                                     & \multicolumn{2}{c}{1}                               &                      &                      \\
					Log likelihood                                          & \multicolumn{2}{c}{$-125,667$}                             & \multicolumn{2}{c}{$-157,908$}                       &                      &                                    \\\bottomrule
				\end{tabular} 
				\begin{tablenotes}[para,flushleft]
					\footnotesize
					Notes: Estimates without parentheses are marginal effects; standard errors are reported in parentheses. For semiparametric cost models, $\bar R_c$ is determined by BIC minimization. The estimates for "Male (direct)" and "Male (contextual direct)" indicate the direct marginal effects (and associated standard errors) of the variable “Male” and of “Share male friends,” respectively, while holding the game equilibrium fixed. However, changes in students’ sex also affect the game equilibrium and the share of male friends. The additional change in the expected outcome resulting from the indirect effect is reported as "Male (indirect)." "Male (total)" is the sum of the direct and indirect effects. 
				\end{tablenotes}
		\end{threeparttable}}
	\end{table}

	\subsection{Counterfactual Analysis}
	In this section, I conduct a counterfactual analysis to examine the effects of the proportion of male students on participation in extracurricular activities. Genders are randomly assigned to students to achieve target proportions of males. For each proportion, I compute the expected number of extracurricular activities by solving the fixed-point equation \eqref{REEY_nomatrix_recall}, while keeping other variables and the network constant. The parameters are fixed at their estimated values. Figure \ref{Ah_cf_plot} presents the average number of extracurricular activities as a function of the proportion of male students. The simulations are based on specifications that control for unobserved school heterogeneity, and the confidence intervals are obtained using the Delta method.

	In Panel~(A), the predictions are based on a model with a semiparametric cost function (Model~4), under the assumption of no social interactions ($a_{ij}^{gg^{\prime}} = 0$ for all $i$, $j$, $g$ and $g^{\prime}$). As the proportion of male students varies from $0\%$ to $100\%$, predicted participation declines from 2.330 to 2.056 activities. This decrease arises because male students participate in fewer activities than female students. Since there are no interactions in this scenario, there are no social multiplier effects. For example, when the proportion of males is $50\%$, which corresponds to the observed data, the average predicted outcome is 2.193 activities---lower than the observed sample mean of 2.353.

	Panel~(B) considers the model with a quadratic cost function but without heterogeneity in peer effects (Model~5). The results show that varying the proportion of males from $0\%$ to $100\%$ leads to a decline in participation from 2.670 to 2.364 activities. The larger decrease relative to Panel~(A) is driven by social multiplier effects. Indeed, increasing the number of females raises participation through both direct effects of students' sex and indirect peer effects. However, these estimates are potentially biased because the quadratic cost specification tends to overestimate social multiplier effects. For example, when the proportion of males is $50\%$, the average predicted outcome exceeds the observed sample mean. Panel~(C) is based on the model with a semiparametric cost function and without heterogeneity in peer effects (Model~4), thereby addressing the concern of biased estimates raised for Panel~(B). The results show that varying the proportion of males from $0\%$ to $100\%$ leads to a decrease in participation from 2.526 to 2.205 activities. 
	
	Panel~(D) incorporates heterogeneous peer effects (Model~7). Owing to this heterogeneity, the average expected outcome is not a monotone function of the proportion of males. The highest participation level is observed in coed schools with $32\%$ male representation, where students attend 2.431 extracurricular activities on average. Interestingly, although male presence generally reduces overall participation, the highest participation level is not observed in all-female schools. This arises because male students generate social multiplier effects by exerting stronger influence on female peers. However, when the proportion of males becomes too large, these effects are offset, as males tend to participate in fewer activities than females.

	\begin{figure}[!h]
		\centering
		\includegraphics[scale = .7]{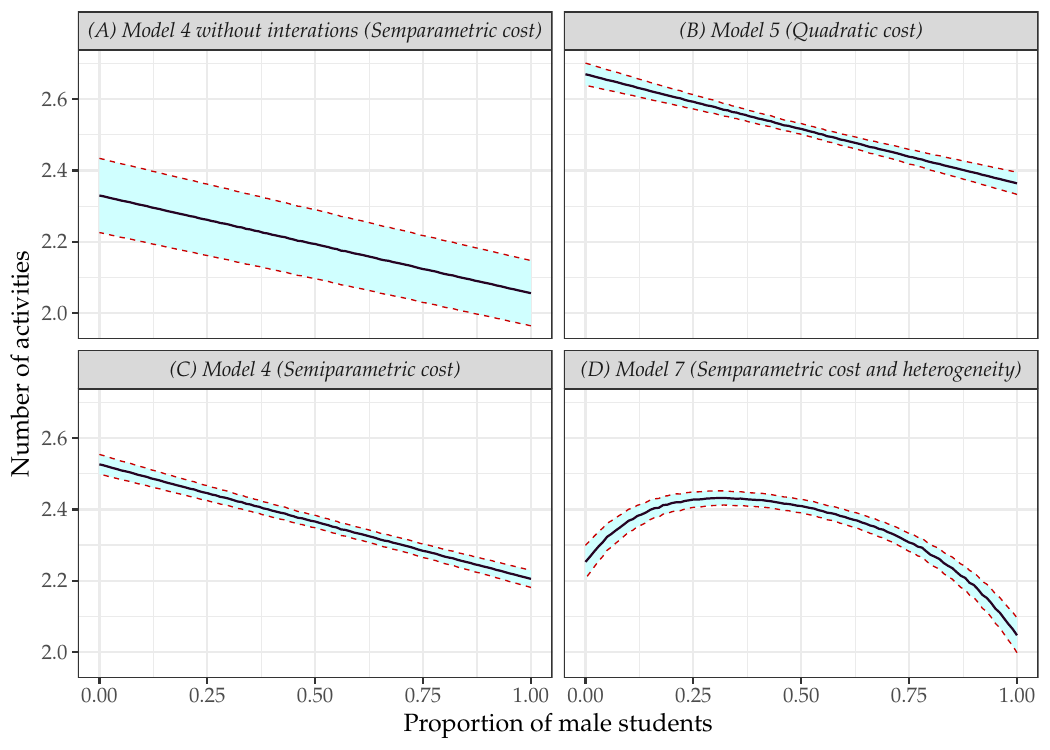}
		
		\vspace{-.2cm}
		\caption{Counterfactual analysis results \label{Ah_cf_plot}}
		
		\justify
		\footnotesize
		Notes: This figure presents the average expected number of extracurricular activities students attend. The confidence intervals are obtained using the Delta method.
	\end{figure}

	\section{Conclusion \label{sec_conclu}}
	This paper develops a peer effect model for count data using a static game of incomplete information. I introduce heterogeneity in peer effects through agents' groups (e.g., gender-based groups). This allows for peer effects to vary across different peer and agent groups. Moreover, the counting nature of the outcome offers flexibility in specifying payoffs. Unlike the conventional linear-quadratic payoff that is imposed in linear-in-means models, I opt for a semiparametric payoff.  I demonstrate that restricting to linear-quadratic payoffs leads to inconsistent estimators of peer effects on count outcomes. This result suggests that estimating peer effects on count outcomes using the linear-in-means SAR or Tobit model would be inconsistent, as these models assume linear-quadratic payoffs.
	
	Furthermore, I provide a general identification analysis. Specifically, I demonstrate that the argument supporting the use of friends' friends who are not directly friends to identify linear-in-means models extends to a large class of nonlinear models.

	The bias in peer effects estimated using the SAR-Tobit model is illustrated through a Monte Carlo study. In my empirical application, I find that the estimates are four times as high as those provided by the proposed approach. This result underscores a significant concern regarding the assumption of a linear-quadratic payoff. This assumption may be violated even for continuous outcomes, potentially leading to biased estimates of peer effects. The linear-quadratic payoff approach is adopted in the literature because it yields estimable econometric models. Introducing flexibility into the payoff structure can be challenging for continuous outcomes. The resulting econometric specification becomes semiparametric, thereby posing significant challenges for identification and estimation. Addressing these challenges is an active area of research in the current literature.

	\appendix
	\begin{center}
		{\Large\bf Appendices}
	\end{center}
	\renewcommand{\theequation}{\thesection.\arabic{equation}}
	\setcounter{equation}{0}
	\renewcommand{\thetable}{\thesection.\arabic{table}}
	\renewcommand{\thefigure}{\thesection.\arabic{figure}}
	\setcounter{table}{0}
	\setcounter{figure}{0}
	
	\vspace{-0.7cm}
	\section{Proofs}\label{app_proof}
	\subsection{Proof of Proposition \ref{prop:xtrema}\label{prop:xtrema:proof}}
	First, I state and prove the following lemma, which adapts \cite{murota1998discrete} to the case of univariate concave discrete functions.
	\begin{lemma} \label{lemma:conc}
		Let $\bar D$ be a convex subset of $\mathbb{R}$, and let $h$ be a discrete concave function defined on $D_h = \bar D \cap \mathbb{Z}$. Let also $r_0 \in D_h$, such that $r_0 - 1, r_0 + 1 \in D_h$.  Then, $h(r_0)$ is the global maximum of $h$ if and only if $h(r_0) \geq \max\left\{h(r_0 - 1), h(r_0 + 1)\right\}$.
	\end{lemma}
	
	\begin{proof}
		If $h(r_0)$ is the global maximum of $h$, then  $h(r_0) \geq h(r_0 + 1)$ and $h(r_0) \geq h(r_0 - 1)$. Thus, $h(r_0) \geq \max\{h(r_0 - 1),h(r_0 + 1)\}$.
		Assume now that $h(r_0) \geq \max\{h(r_0 - 1),h(r_0 + 1)\}$.
		As pointed out by \cite{murota1998discrete}, a discrete function is concave if and only if it can be extended to a continuous concave function. Let $\bar h$ be a concave extension of $h$ on $\bar D$, where $\bar h(r) = h(r)$ for all $r\in D_h$. One can construct $\bar h$ by linearly joining  $h(r_0 -  1)$ to $h(r_0)$ and then $h(r_0)$ to $h(r_0 + 1)$. This implies that $\bar h(r_0)$ is a local maximum of $\bar h$  on $[r_0-1,r_0+1]$. As $\bar h$ is concave,  $\bar h(r_0)$ is also the global maximum of $h$.
	\end{proof}
	
	\noindent \textbf{Proof of Proposition \ref{prop:xtrema}}\\
	\noindent The expected payoff is $U_i^e(y_i) = (\phi_i + \varepsilon_i)y_i - c_{g_i}(y_i) - \sum_{g^{\prime}\in G}\dfrac{\alpha^{g_ig^{\prime}}}{2}\mathbb{E}\left[(\textstyle y_i - \bar y_i^{g^{\prime}})^2|\mathcal{A},\boldsymbol{\phi}\right]$, where $\bar y_i^{g^{\prime}} = \sum_{j= 1}^n w_{ij}^{g_ig^{\prime}} y_j$. First, $U_i^e(.)$ has a global maximum as it is defined on the bounded discrete space $\mathbb{N}_R$.\footnote{Note that even if $R = \infty$, the global maximum exists because $U_i^e(y_i)$ tends to $-\infty$ as $y_i$ grows to $\infty$.}  The global maximum is reached at a single point almost surely (a.s.). Indeed, if it were reached at two points $y_i^{\ast}$ and $y_i^{\ast\prime}$, with $y_i^{\ast}\ne y_i^{\ast\prime}$, then the condition $U_i^e\left(y_i^{\ast}\right) =  U_i^e\left(y_i^{\ast\prime}\right)$ would imply: 
	$$\textstyle\varepsilon_i = \dfrac{c_{g_i}(y_i^{\ast}) - c_{g_i}(y_i^{\ast\prime})}{y_i^{\ast} - y_i^{\ast\prime}} - \phi_i + \sum_{g^{\prime}\in G}\dfrac{\alpha^{g_ig^{\prime}}}{2}(y_i^{\ast} + y_i^{\ast\prime})  -  \sum_{g^{\prime}\in G}\alpha^{g_ig^{\prime}} \mathbb{E}(\Bar y_i^{g^{\prime}}|\mathcal{A},\boldsymbol{\phi}).$$ 
	This equation has zero probability because $\varepsilon_i$ is continuous and the quantity on the right side of the equation is deterministic. As a result, 
	$U_i^e(.)$ has a unique maximizer a.s. 
	
	The second part of Proposition \ref{prop:xtrema} is given by Lemma \ref{lemma:conc}. The cost function $c_{g_i}(y_i)$ is strictly convex in $y_i$ (Assumption \ref{ass_cost}) and $(\phi_i + \varepsilon_i) y_i$ is linear in $y_i$. Moreover, Assumption \ref{ass_pe} implies that $\textstyle\sum_{g^{\prime}\in G}\dfrac{\alpha^{g_ig^{\prime}}}{2}\mathbb{E}\left[(\textstyle y_i - \bar y_i^{g^{\prime}})^2|\mathcal{A},\boldsymbol{\phi}\right]$ is convex. Thus, $U_i^e(.)$ is strictly concave. As a result, the global maximum is reached at $y_i^{\ast}$ if and only if $U_i(y_i^{\ast}) \geq \max\left\{ U_i(y_i^{\ast} - 1),~ U_i(y_i^{\ast} + 1)\right\}$.

	\subsection{Proof of Proposition \ref{prop:expby}\label{append_expby}}
	\noindent
	If $\mathbf{p}$ and $\mathbb{E}(\mathbf{y}|\mathcal{A},\boldsymbol{\phi})$ satisfy Equation  \eqref{REE_nomatrix}, then: $\mathbb{P}(y_i = r|\mathcal{A},\boldsymbol{\phi}) = F_{\varepsilon}(\bar y_i^e + \phi_{i} - \gamma_{g_i}(r)) -  
	F_{\varepsilon}(\bar y_i^e + \phi_i - \gamma_{g_i}(r+ 1))$, where $y_i^e  = \sum_{g^{\prime}\in G}\alpha^{g_ig^{\prime}} \boldsymbol{w}_{i}^{g_ig^{\prime}} \mathbb{E}(\mathbf{y}|\mathcal{A},\boldsymbol{\phi})$. As $\mathbb{E}(y_i|\mathcal{A},\boldsymbol{\phi}) = \sum_{t = 1}^R t \mathbb{P}(y_i = t|\mathcal{A},\boldsymbol{\phi})$, it follows that:
	\begingroup
	\allowdisplaybreaks
	\begin{align*}
		\mathbb{E}(y_i|\mathcal{A},\boldsymbol{\phi}) &\textstyle= \sum_{t = 1}^R t F_{\varepsilon}(\bar y_i^e + \phi_{i} - \gamma_{g_i}(t)) - \sum_{t = 1}^R t F_{\varepsilon}(\bar y_i^e + \phi_i - \gamma_{g_i}(t+ 1)),\\
		\mathbb{E}(y_i|\mathcal{A},\boldsymbol{\phi}) &\textstyle= F_{\varepsilon}(\bar y_i^e + \phi_{i} - \gamma_{g_i}(1)) +  \sum_{t = 2}^R t F_{\varepsilon}(\bar y_i^e + \phi_{i} - \gamma_{g_i}(t)) \\
		& \quad \quad \textstyle -\sum_{t = 1}^R (t + 1) F_{\varepsilon}(\bar y_i^e + \phi_i - \gamma_{g_i}(t+ 1)) + \sum_{t = 1}^R  F_{\varepsilon}(\bar y_i^e + \phi_i - \gamma_{g_i}(t+ 1)),\\
		\mathbb{E}(y_i|\mathcal{A},\boldsymbol{\phi}) &\textstyle= F_{\varepsilon}(\bar y_i^e + \phi_{i} - \gamma_{g_i}(1)) + \sum_{t = 2}^{R + 1}  F_{\varepsilon}(\bar y_i^e + \phi_i - \gamma_{g_i}(t)) \\
		& \quad \quad \textstyle  +\sum_{t = 2}^R t F_{\varepsilon}(\bar y_i^e + \phi_{i} - \gamma_{g_i}(t)) -   \sum_{t = 2}^{R+1} t F_{\varepsilon}(\bar y_i^e + \phi_i - \gamma_{g_i}(t)).
	\end{align*}
	\endgroup
	In the last equation, both $F_{\varepsilon}(\bar y_i^e + \phi_i - \gamma_{g_i}(t))$ and $t F_{\varepsilon}(\bar y_i^e + \phi_i - \gamma_{g_i}(t))$ are equal to zero if $t = R+1$ because $\gamma_{g_i}(R+1) = \infty$.\footnote{Even when $R = \infty$, the same conclusion holds in the limit as $t \to \infty$, since $\sum_{t=2}^{\infty} t F_{\varepsilon}(\bar y_i^e + \phi_i - \gamma_{g_i}(t)) < \infty$ (see Lemma~\ref{lemma_finite} in Online Appendix~\ref{append_unbounded}).} Thus
	\begin{align*}
		\mathbb{E}(y_i|\mathcal{A},\boldsymbol{\phi}) &\textstyle= \sum_{t = 1}^{R}  F_{\varepsilon}(\bar y_i^e + \phi_i - \gamma_{g_i}(t)) +\sum_{t = 2}^R t F_{\varepsilon}(\bar y_i^e + \phi_{i} - \gamma_{g_i}(t))\\
		& \quad \quad \textstyle   -   \sum_{t = 2}^{R} t F_{\varepsilon}(\bar y_i^e + \phi_i - \gamma_{g_i}(t)).
	\end{align*}
	As a result, $\mathbb{E}(y_i|\mathcal{A},\boldsymbol{\phi}) \textstyle= \sum_{t = 1}^R  F_{\varepsilon}(\bar y_i^e + \phi_i - \gamma_{g_i}(t))$.
	
	\subsection{Proof of Proposition \ref{prop:eunique}}
	\label{append:prop:eunique}
	
	\noindent 
	Proposition \ref{prop:xtrema} guarantees the uniqueness of the BNE  since the expected outcome has a unique maximizer. I now show that there is a unique rational belief system. 
	If $\mathbf{p}$ is a rational belief system, Proposition \ref{prop:expby} states that its associated expected outcome, $\mathbf{y}^e = \mathbb{E}(\mathbf{y}|\mathcal{A},\boldsymbol{\phi})$, satisfies $\mathbf{y}^e = \mathbf{L}(\mathbf{y}^e)$, where the mapping $\mathbf L$ is defined for all $\boldsymbol u = (u_1, \dots, u_n) \in \mathbb{R}^n$ as $\mathbf{L}(\boldsymbol u) = (\ell_1(\boldsymbol u),\dots,\ell_n(\boldsymbol u))^{\prime}$, with
	$$\textstyle\ell_i(\boldsymbol u) = \sum_{t = 1}^{R}F_{\varepsilon}\big(\sum_{g^{\prime}\in G}\alpha^{g_ig^{\prime}} \boldsymbol{w}_{i}^{g_ig^{\prime}} \boldsymbol u+ \phi_{i} - \gamma_{g_i}(t)\big).$$
	I show that $\mathbf{L}$ is contracting, i.e., $\forall~\boldsymbol u \in \mathbb{R}^n$, 
	$\lVert \partial \mathbf{L}(\boldsymbol u)/\partial \boldsymbol u^{\prime} \rVert_{\infty} \leq \Bar{\kappa}_c$ for some $\Bar{\kappa}_c < 1$ not depending on $\boldsymbol u$.\footnote{For any $n\times n$-matrix $\textstyle\mathbf{Q} = (q_{ij})_{\substack{ij}}$, $\textstyle\lVert \mathbf{Q} \rVert_{\infty} = \max_i \sum_{i = 1}^{\infty} \lvert q_{ij} \rvert$.}  
	For all $i,j\in \mathcal P$, we have $\textstyle\dfrac{\partial \ell_i(\boldsymbol u)}{\partial u_j} = \sum_{g^{\prime}\in G}\alpha^{g_ig^{\prime}} w_{ij}^{g_ig^{\prime}} f_i(\boldsymbol u)$, where $\textstyle f_i(\boldsymbol u) = \sum_{t = 1}^{R} f_{\varepsilon}\big(\sum_{g^{\prime}\in G}\alpha^{g_ig^{\prime}} \boldsymbol{w}_{i}^{g_ig^{\prime}} \boldsymbol u+ \phi_{i} - \gamma_{g_i}(t)\big)$. \\Thus, $\textstyle\lVert \partial \mathbf{L}(\boldsymbol u)/\partial \boldsymbol u^{\prime} \rVert_{\infty} = \displaystyle\max_{i\in\mathcal P} \textstyle\big\{\sum_{g^{\prime}\in G}\alpha^{g_ig^{\prime}} f_i(\boldsymbol u) \sum_{j = 1}^n w_{ij}^{g_ig^{\prime}}\big\}$. As the matrix $\mathbf{W}^{g_ig^{\prime}}$ is row-normalized, $\sum_{j = 1}^n w_{ij}^{g_ig^{\prime}} \leq 1$. Therefore, $\textstyle\lVert \partial \mathbf{L}(\boldsymbol u)/\partial\boldsymbol u^{\prime} \rVert_{\infty}\leq \max_{i\in \mathcal P}\big\{\sum_{g^{\prime}\in G}\alpha^{g_ig^{\prime}} f_i(\boldsymbol u)\big\}$.
	Moreover, $\displaystyle f_i(\boldsymbol u)  \leq \max_{u \in \mathbb{R}} \textstyle\sum_{t = 1}^{R}  f_{\varepsilon}\left(u -\gamma_{g_i}(t)\right)$. \\Thus, $\sum_{g^{\prime}\in G}\alpha^{g_ig^{\prime}} f_i(\boldsymbol u) \leq \sum_{g^{\prime}\in G}\alpha^{g_ig^{\prime}} \max_{u \in \mathbb{R}} \sum_{t = 1}^{R}  f_{\varepsilon}\left(u -\gamma_{g_i}(t)\right)$, which is strictly lower than one by Assumption \ref{ass_eqcond}. It turns out that $\max_{i\in \mathcal P}\big\{\sum_{g^{\prime}\in G}\alpha^{g_ig^{\prime}} f_i(\boldsymbol u)\big\} < 1$. As a result, there is a unique rational expected outcome that solves the equation $\mathbf{y}^e = \mathbf{L}(\mathbf{y}^e)$. The unique rational expected outcome yields a unique belief system via Equation \eqref{REE_nomatrix}.

	\subsection{Proof of Proposition \ref{prop_ident}\label{append_Ident:Fknown}}
	Let $\tilde{\boldsymbol{\theta}} = (\tilde{\boldsymbol{\alpha}}^{\prime}, \tilde{\boldsymbol{\beta}}^{\prime},\tilde{\boldsymbol{\gamma}}^{\prime})^{\prime}$, where  $\tilde{\boldsymbol{\alpha}} = (\tilde\alpha^{gg^{\prime}}: ~g,g^{\prime}\in G)^{\prime}$ and $\tilde{\boldsymbol{\gamma}} = (\tilde\gamma_g(t):~g\in G, t = 2,\dots,R)^{\prime}$. If $\boldsymbol{\theta}$ and $\tilde{\boldsymbol{\theta}}$ are observationally equivalent, then $p(\mathbf y_s|\mathcal{A}_s,\mathbf Z_s) = \tilde p(\mathbf y_s|\mathcal{A}_s,\mathbf Z_s)$, where $p(\mathbf y_s|\mathcal{A}_s,\mathbf Z_s)$ and $\tilde p(\mathbf y_s|\mathcal{A}_s,\mathbf Z_s)$ are the conditional distributions of $\mathbf y_s$ at $\boldsymbol\theta$ and $\tilde{\boldsymbol{\theta}}$, respectively. 
	
	The equivalence of the distributions also implies the same expected outcome $\mathbb E(\mathbf y_s|\mathcal{A}_s,\mathbf Z_s)$. 
	The condition $p(\mathbf y_s|\mathcal{A}_s,\mathbf Z_s) = \tilde p(\mathbf y_s|\mathcal{A}_s,\mathbf Z_s)$ implies that $p_i(0|\mathcal{A}_s,\mathbf Z_s) = \tilde p_i(0|\mathcal{A}_s,\mathbf Z_s)$ for all $i$ in any $\mathcal P_s$, where $p_i(.|\mathcal{A}_s,\mathbf Z_s)$ and $\tilde p_i(.|\mathcal{A}_s,\mathbf Z_s)$ are the conditional probability mass functions of $y_i$ given $\mathcal{A}_s$ and $\mathbf Z_s$ at $\boldsymbol\theta$ and $\tilde{\boldsymbol{\theta}}$, respectively. It follows that, 
	$F_{\varepsilon}\big(\bar{\mathcal{Y}}_i \boldsymbol\alpha + \boldsymbol{z}_i^{\prime}\boldsymbol{\beta}\big) = F_{\varepsilon}\big(\bar{\mathcal{Y}}_i \tilde{\boldsymbol\alpha}+ \boldsymbol{z}_i^{\prime}\tilde{\boldsymbol{\beta}}\big)$, which means that $\tilde{\boldsymbol{z}}_i^{\prime} \boldsymbol{\Gamma} = \tilde{\boldsymbol{z}}_i^{\prime}\tilde{\boldsymbol{\Gamma}}$, where $\boldsymbol{\Gamma} = (\boldsymbol\alpha^{\prime},\boldsymbol{\beta}^{\prime})^{\prime}$ and $\tilde{\boldsymbol{\Gamma}} = (\tilde{\boldsymbol\alpha}^{\prime},\tilde{\boldsymbol{\beta}}^{\prime})^{\prime}$. Thus, $\frac{1}{S}\sum_{i = 1}^n \tilde{\boldsymbol{z}}_i \tilde{\boldsymbol{z}}_i^{\prime} \boldsymbol{\Gamma} = \frac{1}{S}\sum_{i = 1}^n \tilde{\boldsymbol{z}}_i \tilde{\boldsymbol{z}}_i^{\prime}\tilde{\boldsymbol{\Gamma}}$, which implies that $\plim \left(\frac{1}{S}\sum_{i = 1}^n \tilde{\boldsymbol{z}}_i \tilde{\boldsymbol{z}}_i^{\prime}\right) \boldsymbol{\Gamma} =  \plim \left(\frac{1}{S}\sum_{i = 1}^n \tilde{\boldsymbol{z}}_i \tilde{\boldsymbol{z}}_i^{\prime}\right)\tilde{\boldsymbol{\Gamma}}$. 
	
	As $\plim \left(\frac{1}{S}\sum_{i = 1}^n \tilde{\boldsymbol{z}}_i \tilde{\boldsymbol{z}}_i^{\prime}\right)$ is full rank, $\boldsymbol\alpha = \tilde{\boldsymbol\alpha}$ and $\boldsymbol{\beta} = \tilde{\boldsymbol{\beta}}$. The identification of the thresholds also follows from the condition $p_i(r|\mathcal{A}_s,\mathbf Z_s) = \tilde p_i(r|\mathcal{A}_s,\mathbf Z_s)$ for $r \geq 1$.
	
	\subsection{Proof of Proposition \ref{prop_ident:suff}}\label{append_Ident_Suf}
	Let  $\mathbf y^e_s = \mathbb E(\mathbf y_s|\mathcal{A}_s,\mathbf Z_s)$, $y_i^e = \mathbb E(y_i|\mathcal{A}_s,\mathbf Z_s)$, and $f^{\ast}_i = \sum_{t = 1}^{R}f_{\varepsilon}\big(\sum_{g^{\prime}\in G}\alpha^{g_ig^{\prime}} \boldsymbol{w}_{s,i}^{g_ig^{\prime}} \mathbf y^e_s+ \boldsymbol{z}_i^{\prime}\boldsymbol{\beta} - \gamma_{g_i}(t)\big)$. Let $\mathbf{D}_s = \diag(f^{\ast}_1,\dots,f^{\ast}_{n_s})$ be a diagonal matrix with the diagonal elements being $f^{\ast}_1,\dots,f^{\ast}_{n_s}$. The subscripts $1$, \dots, $n_s$ stand for the agents in $\mathcal{P}_s$. Let $\mathbf{I}_{n_s}$ be the $n_s$-dimensional identity matrix. 
	
	For any $i,j\in\mathcal{P}_s$, the total differential of $y_i^e$ with respect to (w.r.t.) $x_{j,\bar\kappa}$ can be derived from Equation \eqref{REEY_nomatrix_recall} as follows:
	\begingroup
	\allowdisplaybreaks
	\begin{align}
		\textstyle\dfrac{d(y_i^e)}{d(x_{j,\bar\kappa})} &\textstyle= f^{\ast}_i \sum_{g^{\prime}\in G}\alpha^{g_ig^{\prime}} \boldsymbol{w}_{s,i}^{g_ig^{\prime}} \dfrac{d(\mathbf y^e_s)}{d(x_{j,\bar\kappa})} + \beta_{2,\bar\kappa}f^{\ast}_iw_{ij},\nonumber\\
		\textstyle\dfrac{d(y_i^e)}{d(x_{j,\bar\kappa})} &\textstyle=  f^{\ast}_i\sum_{g^{\prime}\in G}\alpha^{g_ig^{\prime}} \boldsymbol{w}_{s,i}^{g_ig^{\prime}} (\mathbf{I}_{n_s} - \mathbf{D}_s\sum_{\hat g, \hat g^{\prime}\in G}\alpha^{\hat g \hat g^{\prime}} \mathbf{W}_s^{\hat g \hat g^{\prime}})^{-1}  \mathbf e_j  + \beta_{2,\bar\kappa}f^{\ast}_iw_{ij}, \label{eq_dye}
	\end{align}
	\endgroup
	where $\mathbf e_j = (e_{j,1},\dots, e_{j,n_s})^{\prime}$ is an $n_s$-vector such that $e_{j,j} = \beta_{1,\bar\kappa}f^{\ast}_j$ and $e_{j,q} = \beta_{2,\bar\kappa} w_{j,q}f^{\ast}_{q}$ if $q \ne j$.

	I employ a proof by contradiction. If Assumption \ref{ass_fullrank} is not satisfied, then for all $g \in G$ and for all $s$, there exists constants $\check \alpha^{g1}$, \dots, $\check \alpha^{gM}$, $\check\beta_0$, $\textstyle\check{\boldsymbol{\beta}}_1$, and $\check{\boldsymbol{\beta}}_2$ such that, for all $i\in\mathcal G_{g} \cap \mathcal{P}_s$:
	\begin{equation}\label{eq_barylin_rec}
		\textstyle\sum_{g^{\prime}\in G}\check\alpha^{gg^{\prime}} \boldsymbol{w}_{s,i}^{g g^{\prime}} \mathbf y^e_s =  \check\beta_0 + \boldsymbol{x}_i^{\prime}\check{\boldsymbol{\beta}}_1 + \Bar{\boldsymbol{x}}_i^{\prime}\check{\boldsymbol{\beta}}_2.
	\end{equation}
	Let $i_1 \in \mathcal G_{g}\cap \mathcal{P}_s$ and $i_3 \in\mathcal{P}_s$ such that $i_3$ is a friend's friend of $i_1$ but $i_3$ is not a direct friend of $i_1$. By writing \eqref{eq_barylin_rec} for $i = i_1$, the right-hand side terms do not depend on $i_3$ because $i_3 \ne i_1$ and $i_3$ is not a direct friend of $i_1$. Thus  the total variation of $\sum_{g^{\prime}\in G}\check\alpha^{gg^{\prime}} \boldsymbol{w}_{s,i_1}^{g g^{\prime}} \mathbf y^e_s$  w.r.t. $x_{i_3,\bar\kappa}$ is zero:
	\begin{equation}
		\dfrac{d(\sum_{g^{\prime}\in G}\check\alpha^{gg^{\prime}} \boldsymbol{w}_{s,i_1}^{g g^{\prime}} \mathbf y^e_s)}{d(x_{i_3,\bar\kappa})} = \sum_{g^{\prime}\in G}\check\alpha^{gg^{\prime}} \boldsymbol{w}_{s,i_1}^{g g^{\prime}} \dfrac{d(\mathbf y^e_s)}{d(x_{i_3,\bar\kappa})} = 0. \label{eq_dyeconstraint}
	\end{equation}
	
	As  $\beta_{1,\bar \kappa}\beta_{2,\bar \kappa}  \geq 0$ and $\alpha^{g g^{\prime}}\geq 0$, \eqref{eq_dye} implies that $\dfrac{d(y_{i_2}^e)}{d(x_{i_3,\bar\kappa})}$ has the same sign or is zero for all $i_2 \in \mathcal P_s$. Specifically, if $i_3$ is a friend of $i_2$, then $\dfrac{d(y_{i_2}^e)}{d(x_{i_3,\bar\kappa})} \ne 0$, with the same sign for all $i_3$.  Since $\boldsymbol \pi = (\boldsymbol\pi_{i}^{\prime}, \dots, \boldsymbol\pi_{n})^{\prime}$ is a full-rank matrix and $M \leq 2$, there exists $i_1 \in \mathcal G_{g}\cap \mathcal{P}_s$ (for certain $\mathcal{P}_s$)  who has their friends in only one group, $\mathcal G_{g_1}$.\footnote{This is trivial when there is only one group. If $M = 2$ and all agents have friends in all groups, then the rows of $\boldsymbol \pi$ will comprise either only zeros or only ones, and the full rank condition will not hold.} For these agents $i_1$, $\boldsymbol{w}_{s,i_1}^{g g^{\prime}} = 0$ if $g^{\prime} \ne g_1$.   By writing \eqref{eq_dyeconstraint} for these $i_1$, I obtain  $\check\alpha^{g g_1}\boldsymbol{w}_{s,i_1}^{gg_1} \dfrac{d(\mathbf y^e_s)}{d(x_{i_3,\bar\kappa})} = 0$, which implies that $\check\alpha^{g g_1} = 0$. Moreover, there exists $i_1 \in \mathcal G_{g}\cap \mathcal{P}_s$ who has friends in the second group, $\mathcal G_{g_2}$ (these agents  could also have friends in $\mathcal G_{g_1}$). For these $i_1$, Equation \eqref{eq_dyeconstraint} implies $\check\alpha^{g g_2}\boldsymbol{w}_{s,i_1}^{gg_2} \dfrac{d(\mathbf y^e_s)}{d(x_{i_3,\bar\kappa})} = 0$, which means that $\check\alpha^{g g_2} = 0$. As a result, \eqref{eq_barylin_rec} becomes  $\check\beta_0 + \boldsymbol{x}_i^{\prime}\check{\boldsymbol{\beta}}_1 + \Bar{\boldsymbol{x}}_i^{\prime}\check{\boldsymbol{\beta}}_2 = 0$ for all $i\in\mathcal P$, which is not possible because $\frac{1}{n}\sum_{i \in 1}^n \boldsymbol{z}_i\boldsymbol{z}_i^{\prime}$ is a full-rank matrix.

	\subsection{Proof of Proposition \ref{prop:asymp}\label{append:asymptotic}} 
	Proposition~\ref{prop:asymp} is a direct implication of \cite{aguirregabiria2007sequential}. The network consists of $S$ subnetworks, and each subnetwork has a bounded number of nodes. Under the regularity assumptions of the NPL method, $\hat{\boldsymbol{\theta}}$ converges in probability to $\check{\boldsymbol{\theta}}_{0}$ as $S$ tends to infinity.\footnote{See also \cite{guerra2020multinomial}, who apply the NPL approach to peer-effects models.}
	
	I now establish the asymptotic normality.\\
	Let $\displaystyle\mathcal{L}_{i}(\boldsymbol{\theta}, \boldsymbol u_{s(i)}) =\sum_{t = 0}^{R}d_{it}\log\left( p_{it}(\boldsymbol{\theta}, \boldsymbol u_{s(i)}) \right)$, where $d_{it} = \mathbbm{1}\{y_i = t\}$, and $p_{it}(\boldsymbol{\theta}, \boldsymbol u_{s(i)}) = F_{\varepsilon}\big(\sum_{g^{\prime}\in G}\alpha^{g_ig^{\prime}} \boldsymbol{w}_{s,i}^{g_ig^{\prime}} \boldsymbol u_{s(i)}+ \boldsymbol{z}_i^{\prime}\boldsymbol{\beta} - \gamma_{g_i}(t)\big) -  F_{\varepsilon}\big(\sum_{g^{\prime}\in G}\alpha^{g_ig^{\prime}} \boldsymbol{w}_{s,i}^{g_ig^{\prime}} \boldsymbol u_{s(i)}+ \boldsymbol{z}_i^{\prime}\boldsymbol{\beta} - \gamma_{g_i}(t+ 1)\big)$. The pseudo-likelihood is $\mathcal{L}(\boldsymbol{\theta}, \boldsymbol u) = \displaystyle\dfrac{1}{S}\sum_{i = 1}^n\mathcal{L}_{i}.$ 
	
	Let $\boldsymbol{\Theta}$ be the space of $\boldsymbol{\theta}$, $\boldsymbol{\chi} = \{\mathcal{A}, \mathbf Z\}$, and $\mathbf{y}^e = \mathbb E(\mathbf y|\boldsymbol{\chi})$. The NPL estimator verifies:  $\partial_{\boldsymbol{\theta}} \mathcal{L}(\hat{\boldsymbol{\theta}} , \hat{\mathbf{y}}^e) = 0$ and $\hat{\mathbf{y}}^e = \mathbf{L}(\hat{\boldsymbol{\theta}}, \hat{\mathbf{y}}^e)$, where $\hat{\mathbf{y}}^e = \hat{\mathbb{E}}(\mathbf{y}|\boldsymbol{\chi})$ and $\partial_{\boldsymbol{\theta}}$ denotes the partial derivative with respect to $\boldsymbol \theta$.  By applying the mean value theorem to $\partial_{\boldsymbol{\theta}} \mathcal{L}(\hat{\boldsymbol{\theta}} , \hat{\mathbf{y}}^e)$ between $\hat{\boldsymbol{\theta}} $ and $\boldsymbol{\theta}_0 $, I obtain:
	$$\sqrt{S}(\hat{\boldsymbol{\theta}}  - \boldsymbol{\theta}_0 )  = -\left(\mathbf{H}_{1} + \mathbf{H}_{2}\right)^{-1}\dfrac{ \sum_{i = 1}^n \partial_{\boldsymbol{\theta}} \mathcal{L}_{i}(\boldsymbol{\theta}_0 , \mathbf{y}^e_s)}{\sqrt{S}},$$
	
	\noindent  For some $\dot{\boldsymbol{\theta}}$ lying between $\hat{\boldsymbol{\theta}} $ and $\boldsymbol{\theta}_0 $, where $\mathbf{H}_{1} = \partial_{\boldsymbol{\theta} \boldsymbol{\theta}^{\prime}} \mathcal{L}(\dot{\boldsymbol{\theta}}, \dot{\mathbf{y}}^e)$, \\
	$\mathbf{H}_{2} = \partial_{\boldsymbol{\theta} \boldsymbol u^{\prime}} \mathcal{L}(\dot{\boldsymbol{\theta}}, \dot{\mathbf{y}}^e) \partial_{\boldsymbol{\theta}^{\prime}}\dot{\mathbf{y}}^{e\prime}$, and $\dot{\mathbf{y}}^e = \mathbf{L}(\dot{\boldsymbol{\theta}}, \dot{\mathbf{y}}^e)$. The symbols $\partial_{\boldsymbol{\theta} \boldsymbol{\theta}^{\prime}}$ and $\partial_{\boldsymbol{\theta} \boldsymbol u^{\prime}}$ represent the second-order partial derivatives with respect to $\boldsymbol{\theta} \boldsymbol{\theta}^{\prime}$ and $\boldsymbol{\theta} \boldsymbol u^{\prime}$, respectively. 
	Since $\boldsymbol{\theta}_0  = \displaystyle \arg\max_{\boldsymbol{\theta} \in \boldsymbol{\Theta}} \mathcal{L}_0(\boldsymbol{\theta}, \mathbf{y}^e)$ and $\partial_{\boldsymbol{\theta}} \mathcal{L}_{i}(\boldsymbol{\theta}_0 , \mathbf{y}^e_s)$ is bounded, it follows that $\mathbb{E}(\partial_{\boldsymbol{\theta}} \mathcal{L}_{i}(\boldsymbol{\theta}_0 , \mathbf{y}^e_s)|\boldsymbol{\chi}) = 0$ and $\mathbb{V}(\partial_{\boldsymbol{\theta}} \mathcal{L}_{i}(\boldsymbol{\theta}_0 , \mathbf{y}^e_s)|\boldsymbol{\chi})$ is bounded. Consequently, I can apply the central limit theorem (CLT) to $\dfrac{ \sum_{i = 1}^n \partial_{\boldsymbol{\theta}} \mathcal{L}_{i}(\boldsymbol{\theta}_0, \mathbf{y}^e_s)}{\sqrt{S}}$, conditional on $\boldsymbol{\chi}$.\footnote{I use a conditional CLT because $\partial_{\boldsymbol{\theta}} \mathcal{L}_{i}(\boldsymbol{\theta}_0, \hat{\mathbf{y}}^e_s)$ is independent across $i$, conditional on $\boldsymbol{\chi}$. See examples of conditional CLT in \cite{van2000asymptotic} and \cite{houndetoungan2024inference}. Note that Lindeberg's condition is satisfied because $\partial_{\boldsymbol{\theta}} \mathcal{L}_{i}(\boldsymbol{\theta}_0 , \mathbf{y}^e_s)$ is bounded.} Let $\boldsymbol{\Sigma} = \dfrac{ \sum_{i = 1}^n \mathbb{V}(\partial_{\boldsymbol{\theta}} \mathcal{L}_{i}(\boldsymbol{\theta}_0 , \hat{\mathbf{y}}^e_s)|\boldsymbol{\chi})}{S}$. By assuming that $\plim\boldsymbol{\Sigma} = \boldsymbol{\Sigma}_0$, $\plim \mathbf{H}_{1} = \mathbf{H}_{1,0}$, and $\plim \mathbf{H}_{2} = \mathbf{H}_{2,0}$ exist, where $\plim$ is the limit in probability as $S$ grows to infinity, I obtain: 
	$$\sqrt{S}(\hat{\boldsymbol{\theta}}  - \boldsymbol{\theta}_0 ) \overset{d}{\to} \mathcal{N}\left(0, (\mathbf{H}_{1,0} + \mathbf{H}_{2,0})^{-1}\boldsymbol{\Sigma}_0 (\mathbf{H}_{1,0}^{\prime} + \mathbf{H}_{2,0}^{\prime})^{-1}\right).$$

	{\fontsize{11}{10}\selectfont
		\bibliography{References}

@article{brock2001discrete,
  title={Discrete choice with social interactions},
  author={Brock, William A and Durlauf, Steven N},
  journal={The Review of Economic Studies},
  volume={68},
  number={2},
  pages={235--260},
  year={2001},
  publisher={Wiley-Blackwell}
}

@article{bramoulle2009identification,
  title={Identification of peer effects through social networks},
  author={Bramoull{\'e}, Yann and Djebbari, Habiba and Fortin, Bernard},
  journal={Journal of Econometrics},
  volume={150},
  number={1},
  pages={41--55},
  year={2009},
  publisher={Elsevier}
}

@inproceedings{de2017econometrics,
  title={Econometrics of network models},
  author={De Paula, {\'A}ureo},
  booktitle={Advances in economics and econometrics: Theory and applications, eleventh world congress},
  pages={268--323},
  year={2017},
  organization={Cambridge University Press Cambridge}
}

@article{manski1993identification,
  title={Identification of endogenous social effects: The reflection problem},
  author={Manski, Charles F},
  journal={The Review of Economic Studies},
  volume={60},
  number={3},
  pages={531--542},
  year={1993},
  publisher={Wiley-Blackwell}
}

@article{boucher2022peer,
  title={Peer-induced beliefs regarding college participation},
  author={Boucher, Vincent and Dedewanou, F Antoine and Dufays, Arnaud},
  journal={Economics of Education Review},
  volume={90},
  pages={102307},
  year={2022},
  publisher={Elsevier}
}

@article{wan2014semiparametric,
  title={Semiparametric identification of binary decision games of incomplete information with correlated private signals},
  author={Wan, Yuanyuan and Xu, Haiqing},
  journal={Journal of Econometrics},
  volume={182},
  number={2},
  pages={235--246},
  year={2014},
  publisher={Elsevier}
}

@article{peng2019heterogeneous,
  title={Heterogeneous endogenous effects in networks},
  author={Peng, Sida},
  journal={arXiv preprint arXiv:1908.00663},
  year={2019}
}

@article{beugnot2019gender,
  title={Gender and peer effects on performance in social networks},
  author={Beugnot, Julie and Fortin, Bernard and Lacroix, Guy and Villeval, Marie Claire},
  journal={European Economic Review},
  volume={113},
  pages={207--224},
  year={2019},
  publisher={Elsevier}
}

@article{rothenberg1971identification,
  title={Identification in parametric models},
  author={Rothenberg, Thomas J},
  journal={Econometrica: Journal of the Econometric Society},
  pages={577--591},
  year={1971},
  publisher={JSTOR}
}

@article{comola2022heterogeneous,
  title={Heterogeneous peer effects and gender-based interventions for teenage obesity},
  author={Comola, Margherita and Dieye, Rokhaya and Fortin, Bernard},
  journal={Journal of Health Economics},
  pages={103023},
  year={2025},
  publisher={Elsevier}
}

@article{xu2015maximum,
  title={Maximum likelihood estimation of a spatial autoregressive Tobit model},
  author={Xu, Xingbai and Lee, Lung-fei},
  journal={Journal of Econometrics},
  volume={188},
  number={1},
  pages={264--280},
  year={2015},
  publisher={Elsevier}
}

@article{bramoulle2020peer,
  title={Peer effects in networks: A survey},
  author={Bramoull{\'e}, Yann and Djebbari, Habiba and Fortin, Bernard},
  journal={Annual Review of Economics},
  volume={12},
  pages={603--629},
  year={2020},
  publisher={Annual Reviews}
}

@article{boucher2016some,
  title={Some challenges in the empirics of the effects of networks},
  author={Boucher, Vincent and Fortin, Bernard},
  journal={The Oxford Handbook on the Economics of Networks},
  pages={277--302},
  year={2016},
  publisher={Oxford University Press Oxford}
}

@book{van2000asymptotic,
  title={Asymptotic statistics},
  author={Van der Vaart, Aad W},
  volume={3},
  year={2000},
  publisher={Cambridge University Press}
}

@article{yang2017social,
  title={Social interactions under incomplete information with heterogeneous expectations},
  author={Yang, Chao and Lee, Lung-fei},
  journal={Journal of Econometrics},
  volume={198},
  number={1},
  pages={65--83},
  year={2017},
  publisher={Elsevier}
}

@article{bajari2010estimating,
  title={Estimating static models of strategic interactions},
  author={Bajari, Patrick and Hong, Han and Krainer, John and Nekipelov, Denis},
  journal={Journal of Business \& Economic Statistics},
  volume={28},
  number={4},
  pages={469--482},
  year={2010},
  publisher={Taylor \& Francis}
}

@article{xu2018social,
  title={Social interactions in large networks: A game theoretic approach},
  author={Xu, Haiqing},
  journal={International Economic Review},
  volume={59},
  number={1},
  pages={257--284},
  year={2018},
  publisher={Wiley Online Library}
}

@article{houndetoungan2024inference,
  title={Inference for Two-Stage Extremum Estimators},
  author={Houndetoungan, Aristide and Maoude, Abdoul Haki},
  journal={arXiv preprint arXiv:2402.05030},
  year={2024}
}

@article{manski1988identification,
  title={Identification of binary response models},
  author={Manski, Charles F},
  journal={Journal of the American Statistical Association},
  volume={83},
  number={403},
  pages={729--738},
  year={1988},
  publisher={Taylor \& Francis}
}

@article{manski1988identificationoa,
  title={Identification of binary response models},
  author={Manski, Charles F},
  journal={Journal of the American Statistical Association},
  volume={83},
  number={403},
  pages={729--738},
  year={1988},
  publisher={Taylor \& Francis}
}

@article{blume2015linear,
  title={Linear social interactions models},
  author={Blume, Lawrence E and Brock, William A and Durlauf, Steven N and Jayaraman, Rajshri},
  journal={Journal of Political Economy},
  volume={123},
  number={2},
  pages={444--496},
  year={2015},
  publisher={University of Chicago Press Chicago, IL}
}

@article{li2009binary,
  title={Binary choice under social interactions: an empirical study with and without subjective data on expectations},
  author={Li, Ji and Lee, Lung-fei},
  journal={Journal of Applied Econometrics},
  volume={24},
  number={2},
  pages={257--281},
  year={2009},
  publisher={Wiley Online Library}
}

@article{brock2007identification,
  title={Identification of binary choice models with social interactions},
  author={Brock, William A and Durlauf, Steven N},
  journal={Journal of Econometrics},
  volume={140},
  number={1},
  pages={52--75},
  year={2007},
  publisher={Elsevier}
}

@article{boucher2020estimating,
  title={Estimating peer effects using partial network data},
  author={Boucher, Vincent and Houndetoungan, Aristide},
  journal={Review of Economics and Statistics},
  pages={1--48},
  year={2025},
  publisher={MIT Press 255 Main Street, 9th Floor, Cambridge, Massachusetts 02142, USA}
}

@article{de2010identification,
  title={Identification of social interactions through partially overlapping peer groups},
  author={De Giorgi, Giacomo and Pellizzari, Michele and Redaelli, Silvia},
  journal={American Economic Journal: Applied Economics},
  volume={2},
  number={2},
  pages={241--275},
  year={2010},
  publisher={American Economic Association}
}

@article{liu2017social,
  title={A social interaction model with ordered choices},
  author={Liu, Xiaodong and Zhou, Jiannan},
  journal={Economics Letters},
  volume={161},
  pages={86--89},
  year={2017},
  publisher={Elsevier}
}

@article{murota1998discrete,
  title={Discrete convex analysis},
  author={Murota, Kazuo},
  journal={Mathematical Programming},
  volume={83},
  number={1},
  pages={313--371},
  year={1998},
  publisher={Springer}
}

@article{brock2002multinomial,
  title={A multinomial-choice model of neighborhood effects},
  author={Brock, William A and Durlauf, Steven N},
  journal={American Economic Review},
  volume={92},
  number={2},
  pages={298--303},
  year={2002}
}

@article{aradillas2010semiparametric,
  title={Semiparametric estimation of a simultaneous game with incomplete information},
  author={Aradillas-Lopez, Andres},
  journal={Journal of Econometrics},
  volume={157},
  number={2},
  pages={409--431},
  year={2010},
  publisher={Elsevier}
}

@article{ballester2006s,
  title={Who's who in networks. Wanted: The key player},
  author={Ballester, Coralio and Calv{\'o}-Armengol, Antoni and Zenou, Yves},
  journal={Econometrica},
  volume={74},
  number={5},
  pages={1403--1417},
  year={2006},
  publisher={Wiley Online Library}
}

@article{guerra2020multinomial,
  title={Multinomial choice with social interactions: Occupations in Victorian London},
  author={Guerra, Jos{\'e}-Alberto and Mohnen, Myra},
  journal={Review of Economics and Statistics},
  volume={104},
  number={4},
  pages={736--747},
  year={2022},
  publisher={MIT Press One Rogers Street, Cambridge, MA 02142-1209, USA journals-info~…}
}

@article{calvo2009peer,
  title={Peer effects and social networks in education},
  author={Calv{\'o}-Armengol, Antoni and Patacchini, Eleonora and Zenou, Yves},
  journal={The Review of Economic Studies},
  volume={76},
  number={4},
  pages={1239--1267},
  year={2009},
  publisher={Wiley-Blackwell}
}

@article{lee2014binary,
  title={Binary choice models with social network under heterogeneous rational expectations},
  author={Lee, Lung-fei and Li, Ji and Lin, Xu},
  journal={Review of Economics and Statistics},
  volume={96},
  number={3},
  pages={402--417},
  year={2014},
  publisher={MIT Press}
}

@article{aguirregabiria2007sequential,
  title={Sequential estimation of dynamic discrete games},
  author={Aguirregabiria, Victor and Mira, Pedro},
  journal={Econometrica},
  volume={75},
  number={1},
  pages={1--53},
  year={2007},
  publisher={Wiley Online Library}
}

@article{kasahara2012sequential,
  title={Sequential estimation of structural models with a fixed point constraint},
  author={Kasahara, Hiroyuki and Shimotsu, Katsumi},
  journal={Econometrica},
  volume={80},
  number={5},
  pages={2303--2319},
  year={2012},
  publisher={Wiley Online Library}
}

@article{lin2024binary,
  title={Binary choice with misclassification and social interactions, with an application to peer effects in attitude},
  author={Lin, Zhongjian and Hu, Yingyao},
  journal={Journal of Econometrics},
  volume={238},
  number={1},
  pages={105551},
  year={2024},
  publisher={Elsevier}
}
		\bibliographystyle{ecta}}
	
	\clearpage
	\setcounter{page}{1}
	\begin{center}
		\Large{Online Appendix to "\TITLE"}\\
		\normalsize Aristide Houndetoungan
	\end{center}
	
	\bigskip
	
	\section{Unbounded Count Variables}\label{append_unbounded}
	In the main text, I assume that $y_i \in \mathbb{N}_{R} = \{0,1,\dots,R\}$, where $R$ is a finite strictly positive integer. In this section, I allow $R$ to be infinite, as is common in many count data models with a Poisson or Negative Binomial distribution. The main complication when $R = \infty$ is that many summations in the paper run over infinitely many terms, so one must ensure that they are convergent. An example is the expected outcome in Proposition \ref{prop:expby} and the uniqueness condition in Assumption \ref{ass_eqcond}.  
	
	To ensure that these sums are convergent, I impose the following restrictions.
	\begin{assumption}\label{ass_unbounded} \hfill
		\begin{enumerate}[label=(\alph*), ref=(\alph*), nosep, leftmargin=2\parindent]
			\item The probability density function (pdf) of the agent type distribution satisfies $f_{\varepsilon}(x) = o\left(\lvert x\rvert^{-\kappa}\right)$ at $\infty$, for some $\kappa > 4$.\label{ass_unbounded:e}
			\item  For any $g \in G$, $\displaystyle\inf_{r \in \mathbb N_{\infty}} \left(\Delta c_{g}(r+1) - \Delta c_{g}(r)\right) > 0$.\label{ass_unbounded_cost}
		\end{enumerate}
	\end{assumption}
	\noindent The condition $f_{\varepsilon}(x) = o\left(\lvert x\rvert^{-\kappa}\right)$ at $\infty$ for some $\kappa > 4$ in Statement \ref{ass_unbounded:e} ensures that the probability that $y_i$ takes the value $r$ converges to zero at a certain rate as $r$ tends to infinity, as is typical in most count data models. Most common distributions with finite moments of order greater than four satisfy this condition (e.g., normal, logistic, and Student’s $t$ distributions). Condition \ref{ass_unbounded_cost} guarantees that the difference in the cost function remains strictly increasing as $r$ approaches infinity. It serves as the asymptotic counterpart of Assumption~\ref{ass_cost}. Since $\textstyle \gamma_{g_i}(r) = \Delta c_{g_i}(r) +  \sum_{g^{\prime}\in G}(r - 1/2)\alpha^{g_ig^{\prime}}$, it follows that $\gamma_{g_i}(r)$ grows to infinity as $r$ tends to infinity.

	\begin{lemma} \label{lemma_finite}
		I establish the following results.
		\begin{enumerate}[label=(\alph*), ref=(\alph*), nosep, leftmargin=2\parindent]
			\item For any $g\in G$, $b \in [0, 2]$, and $u\in\mathbb R$, $\displaystyle h_1(u) = \sum_{t=0}^{+\infty} t^{b}F_{\varepsilon}(u - \gamma_g(t))$ is finite.\label{lemma_finite_finite}
			\item For any $g\in G$ and $u \in \mathbb R$, let $h_2$ be the function defined by $\displaystyle h_2(u) = \sum_{t=0}^{+\infty} f_{\varepsilon}(u - \gamma_g(t))$. Then $h_2$ is bounded on $\mathbb{R}$. \label{lemma_finite_bound}
		\end{enumerate}
	\end{lemma}
	\noindent Before presenting the proof of Lemma \ref{lemma_finite}, I first discuss how it ensures that the sums over infinitely many terms in the paper do not diverge. Statement \ref{lemma_finite_bound} guarantees that the upper bound $\big(\max_{u \in \mathbb{R}} \sum_{t = 1}^{R} f_{\varepsilon}\left(u - \gamma_g(t)\right)\big)^{-1}$ in Assumption \ref{ass_eqcond} is finite. Statement \ref{lemma_finite_finite} ensures the convergence of other sums involving infinitely many terms. For example, the proof of the expected outcome in Appendix \ref{append_expby} relies on this statement with $b = 1$, while the derivation of the conditional variance of the outcome in Online Appendix \ref{oa_comparison} applies the statement with $b = 2$. Consequently, all results in the paper extend to the case where $R$ is unbounded under Assumption~\ref{ass_unbounded}.

	As discussed in Remark \ref{rem_semi}, when $R$ is large or unbounded, the model involves many cut points and cannot be estimated without additional restrictions. The approach in Remark \ref{rem_semi} can be used to reduce the number of parameters. The idea is to specify a semiparametric cost function that switches from a fully nonparametric form to a quadratic cost function. This implies that the cut points are evenly spaced from this threshold. The threshold can be selected using information criteria such as BIC.

	\paragraph{Proof of Lemma \ref{lemma_finite}} \hfill\\
	I first show Statement \ref{lemma_finite_finite}.

	\noindent Since $f_{\varepsilon}(x) = o(\lvert x\rvert^{-\kappa})$ as $\lvert x\rvert\to\infty$ and $f_{\varepsilon}$ is symmetric, for large negative $x$ we have $f_{\varepsilon}(x) < \lvert x\rvert^{-\kappa}$. Integrating $f_{\varepsilon}$ over $(-\infty, x]$ then yields $F_{\varepsilon}(x) < \dfrac{\lvert x\rvert^{-\kappa+1}}{\kappa-1}$. As $\gamma_g(t)$ is strictly increasing to $\infty$, for $t$ sufficiently large $u-\gamma_g(t)$ becomes large and negative, and hence
	\[
	F_{\varepsilon}(u-\gamma_g(t)) < \frac{\lvert u-\gamma_g(t)\rvert^{-\kappa+1}}{\kappa-1}.
	\]
	Therefore, to show that $h_1(u)<\infty$, it suffices to prove that
	\[
	\sum_{t=t_0}^{\infty} t^{b}\lvert u-\gamma_g(t)\rvert^{-\kappa+1}<\infty,
	\]
	where $t_0$ is a large integer such that $\lvert u-\gamma_g(t)\rvert>0$ for all $t\ge t_0$.
	
	Let $\gamma^{\ast}=\inf_{t\in\mathbb N}(\gamma_g(t+1)-\gamma_g(t))$. By Assumption \ref{ass_unbounded}, $\gamma^{\ast}>0$. Thus, for all $t\ge t_0$, and since $-\kappa+1<0$, we have
	\begingroup
	\allowdisplaybreaks
	\begin{align}
		&\gamma_g(t) \ge (t-t_0)\gamma^{\ast}+\gamma_g(t_0), \nonumber \\
		&\lvert u-\gamma_g(t)\rvert = \gamma_g(t)-u \ge (t-t_0)\gamma^{\ast}+\gamma_g(t_0)-u > 0, \nonumber \\
		&\lvert u-\gamma_g(t)\rvert^{-\kappa+1} \le \big((t-t_0)\gamma^{\ast}+\gamma_g(t_0)-u\big)^{-\kappa+1}, \nonumber \\
		&\lvert u-\gamma_g(t)\rvert^{-\kappa+1} = O(t^{-\kappa+1}), \nonumber \\
		&t^{b}\lvert u-\gamma_g(t)\rvert^{-\kappa+1} = O(t^{-\kappa+b+1}). \nonumber
	\end{align}
	\endgroup
	
	Consequently, $\sum_{t=t_0}^{\infty} t^{b}(\lvert u-\gamma_g(t)\rvert+1)^{-\kappa+1}<\infty$ whenever $\kappa-b-1>1$, which holds since $\kappa>4$ and $b\in[0,2]$. Therefore, for any $g\in G$, $b\in[0,2]$, and $u\in\mathbb R$, $h_1(u)<\infty$.

	\bigskip
	\noindent I now show Statement \ref{lemma_finite_bound}.\\
	Since $f_{\varepsilon}(x) = o(\lvert x\rvert^{-\kappa})$ as $\lvert x\rvert\to\infty$, we also have $f_{\varepsilon}(x) = o((\lvert x\rvert + 1)^{-\kappa})$ as $\lvert x\rvert\to\infty$. Thus, $f_{\varepsilon}(x) \leq (\lvert x\rvert + 1)^{-\kappa}$ whenever $\lvert x\rvert$ is sufficiently large. In fact, since $\lvert x\rvert + 1 > 0$ for all $\lvert x\rvert$, there thus exists $M > 0$ such that $f_{\varepsilon}(x) \leq M(\lvert x\rvert + 1)^{-\kappa}$ for all $x$. Consequently, to show that $h_2$ is bounded on $\mathbb{R}$, it is sufficient to show that the function $f^{\ast}$ defined by
	$$f^{\ast}(u) =  \sum_{t = 0}^{\infty} \left(\lvert u-\gamma_g(t)\rvert+1\right)^{-\kappa}$$
	is bounded on $\mathbb R$. 
	
	Using a similar argument as in the proof of Statement \ref{lemma_finite_finite}, we have $\left(\lvert u-\gamma_g(t)\rvert+1\right)^{-\kappa} = O(t^{-\kappa})$. Thus, $f^{\ast}(u) < \infty$ for all $u$. However, this is not sufficient to conclude that $f^{\ast}$ is bounded. I must also show that $\displaystyle\lim_{u\to -\infty} f^{\ast}(u) < \infty$ and that $\displaystyle\lim_{u\to +\infty} f^{\ast}(u) < \infty$.
	
	If $u \leq 0$, then $\left(\lvert u-\gamma_g(t)\rvert+1\right)^{-\kappa} = \left(\gamma_g(t) - u +1\right)^{-\kappa} \leq \left(\gamma_g(t) +1\right)^{-\kappa}$. Thus, $f^{\ast}(u) \leq f^{\ast}(0)$. Since $f^{\ast}$ is a positive function, this means that $\displaystyle\lim_{u\to -\infty} f^{\ast}(u) < \infty$.
	
	For the second limit, assume that $u$ is positive and sufficiently large. Then, there exists $t_1\in\mathbb{N}$ such that $\forall~t \leq t_1 - 1$, $u > \gamma_g(t)$, and $\forall~t \geq t_1$, $u \leq \gamma_g(t)$. Thus, $f^{\ast}(u)$ can be decomposed as
	\begingroup
	\allowdisplaybreaks
	\begin{align*}
		f^{\ast}(u) &= \sum_{t = 0}^{t_1 - 1}\left(\lvert u-\gamma_g(t)\rvert+1\right)^{-\kappa} + \sum_{t = t_1}^{\infty}\left(\lvert u-\gamma_g(t)\rvert+1\right)^{-\kappa},\\
		f^{\ast}(u) &\leq \sum_{t = 0}^{t_1 - 1}\left(u-\gamma_g(t)+1\right)^{-\kappa} + \sum_{t = t_1}^{\infty}\left(\gamma_g(t)-u+1\right)^{-\kappa},\\
		f^{\ast}(u) &\leq \sum_{t = 0}^{t_1 - 1}\left(\gamma_g(t_1 - 1)-\gamma_g(t)+1\right)^{-\kappa} + \sum_{t = t_1}^{\infty}\left(\gamma_g(t)-\gamma_g(t_1)+1\right)^{-\kappa}.
	\end{align*}
	\endgroup
	The last equation holds because $u > \gamma_g(t_1 - 1)$ and $u \leq \gamma_g(t_1)$.
	
	\noindent Let $\gamma^{\ast}=\inf_{t\in\mathbb N}(\gamma_g(t+1)-\gamma_g(t))$. By Assumption \ref{ass_unbounded}, $\gamma^{\ast}>0$. Thus, if $t \leq t_1 -1$, then $\gamma_g(t_1 - 1)-\gamma_g(t) \geq(t_1 - 1 - t)\gamma^{\ast}$, and hence $(\gamma_g(t_1 - 1)-\gamma_g(t) + 1)^{-\kappa} \leq \left((t_1 - 1 - t)\gamma^{\ast} + 1\right)^{-\kappa}$. Similarly, if $t \geq t_1$, then $(\gamma_g(t) - \gamma_g(t_1) + 1)^{-\kappa} \leq \left((t - t_1)\gamma^{\ast} + 1\right)^{-\kappa}$. Consequently,
	\begingroup
	\allowdisplaybreaks
	\begin{align*}
		f^{\ast}(u) &\leq \sum_{t = 0}^{t_1 - 1}\left((t_1 - 1 - t)\gamma^{\ast} + 1\right)^{-\kappa} + \sum_{t = t_1}^{\infty}\left((t - t_1)\gamma^{\ast} + 1\right)^{-\kappa},\\
		f^{\ast}(u) &\leq \sum_{t = 0}^{t_1 - 1}\left(t\gamma^{\ast} + 1\right)^{-\kappa} + \sum_{t = 0}^{\infty}\left(t\gamma^{\ast} + 1\right)^{-\kappa},\\
		f^{\ast}(u) &\leq 2\sum_{t = 0}^{\infty}\left(t\gamma^{\ast} + 1\right)^{-\kappa}.
	\end{align*}
	\endgroup
	
	\noindent The quantity $\displaystyle 2\sum_{t = 0}^{\infty}\left(t\gamma^{\ast} + 1\right)^{-\kappa}$ does not depend on $u$ and is finite because $\kappa > 1$. Hence, $\displaystyle\lim_{u\to +\infty} f^{\ast}(u) < \infty$. As a result, $f^{\ast}$ is bounded.  This completes the proof of the lemma.

	\section{Supplemental Results on the Econometric Model}
	\subsection{Comparison with Other Count Data Models} \label{oa_comparison}
	In this section, I compare the model with the Poisson and Negative Binomial (NB) models, first by examining the outcome's dispersion. It is well known that the Poisson model implies equality between the expectation and the variance of the outcome, whereas the NB model includes a dispersion parameter and is therefore more flexible. 
	
	Analyzing the dispersion is challenging because the expected outcome does not have a closed form, and consequently, neither does the variance. Using a similar argument as in Proposition \ref{prop:expby} to compute the expected outcome, we have:
	\begingroup
	\allowdisplaybreaks
	\begin{align*}
		\mathbb{E}(y_i^2|\mathcal{A}, \boldsymbol{\phi}) &\textstyle= \sum_{t = 1}^R t^2 F_{\varepsilon}(\bar y_i^e + \phi_{i} - \gamma_{g_i}(t)) - \sum_{t = 1}^R t^2 F_{\varepsilon}(\bar y_i^e + \phi_i - \gamma_{g_i}(t+ 1)),\\
		\mathbb{E}(y_i^2|\mathcal{A}, \boldsymbol{\phi}) &\textstyle= \sum_{t = 1}^R t^2 F_{\varepsilon}(\bar y_i^e + \phi_{i} - \gamma_{g_i}(t)) - \sum_{t = 1}^R (t + 1)^2  F_{\varepsilon}(\bar y_i^e + \phi_i - \gamma_{g_i}(t+ 1))\\
		& \quad \quad \textstyle  + \sum_{t = 1}^R (2t + 1) F_{\varepsilon}(\bar y_i^e + \phi_i - \gamma_{g_i}(t+ 1))\\
		\mathbb{E}(y_i^2|\mathcal{A}, \boldsymbol{\phi}) &\textstyle= F_{\varepsilon}(\bar y_i^e + \phi_{i} - \gamma_{g_i}(1)) + \sum_{t = 1}^R (2t + 1) F_{\varepsilon}(\bar y_i^e + \phi_i - \gamma_{g_i}(t + 1)),\\
		\mathbb{E}(y_i^2|\mathcal{A}, \boldsymbol{\phi}) &\textstyle= \sum_{t = 1}^R (2t - 1) F_{\varepsilon}(\bar y_i^e + \phi_i - \gamma_{g_i}(t)).
	\end{align*}
	\endgroup
	
	The dispersion nature thus depends on the relationship between 
	$SF_i := \sum_{t = 1}^R (2t - 1) F_{\varepsilon}(\bar y_i^e + \phi_i - \gamma_{g_i}(t))$ 
	and $\bar y_i^e  = \mathbb{E}(y_i|\mathcal{A}, \boldsymbol{\phi})$. Specifically, if $SF_i = (\bar y_i^e)^2$, then the outcome is equidispersed; if $SF_i > (\bar y_i^e)^2$, it is overdispersed; and if $SF_i < (\bar y_i^e)^2$, it is underdispersed. Since $SF_i$ cannot be further simplified, this relationship can only be assessed numerically. 
	
	A simple simulation exercise shows that the model imposes no strong restrictions on the dispersion of the outcome. I consider a single network with three individuals from the same group. I set $\bar R = 1$ (evenly spaced cut points), $\lambda = 0.4c$, $\bar \delta = c$, and $\phi_i/c \sim \mathcal{U}(-5, 5)$, where $c$ is a tuning parameter that influence both the expectation and the variance of the outcome. Agent type is assumed to follow a standard normal distribution. Table \ref{tab_dispersion} indicates that the outcome is underdispersed when $c = 1$, equidispersed for agent 1 when $c = 1.83$, and overdispersed when $c = 5$.
	
	\begin{table}[h!]
		\centering
		\caption{Dispersion}
		\label{tab_dispersion}
		\small
		\begin{tabular}{lcccccc}
			\toprule
			& \multicolumn{2}{c}{c = 1} & \multicolumn{2}{c}{c = 1.83} & \multicolumn{2}{c}{c = 5} \\
			& $\mathbb{E}(y_i|\mathcal{A}, \boldsymbol{\phi})$       & $\mathbb{V}(y_i|\mathcal{A}, \boldsymbol{\phi})$         & 
			$\mathbb{E}(y_i|\mathcal{A}, \boldsymbol{\phi})$       & $\mathbb{V}(y_i|\mathcal{A}, \boldsymbol{\phi})$    & 
			$\mathbb{E}(y_i|\mathcal{A}, \boldsymbol{\phi})$       & $\mathbb{V}(y_i|\mathcal{A}, \boldsymbol{\phi})$          \\ \midrule
			Agent 1 & 1.1290      & 0.5449      & 1.2720        & 1.2733       & 2.2550      & 6.4624      \\
			Agent 2 & 3.2960      & 0.5629      & 3.3100        & 1.6976       & 3.4240      & 8.5888      \\
			Agent 3 & 1.4520      & 0.5803      & 1.5240        & 1.4589       & 2.4430      & 7.1439      \\\bottomrule
		\end{tabular}
	\end{table}
	
	The results in Table \ref{tab_dispersion} are obtained under the assumption of evenly spaced cut points. Allowing the spacing between cut points to vary enhances the model's flexibility. For instance, the probability of the outcome taking a specific value can be made large by increasing the gap between the cut points around that value. Such flexibility is not possible with the Poisson or NB models, where the outcome distribution is determined only by a location parameter in the Poisson case and by location and dispersion parameters in the NB case. The proposed model provides additional flexibility by incorporating parameters that can shape the probability at specific values. The cost of this flexibility is the larger number of parameters to estimate. The semiparametric assumption on the cost function reduces the number of parameters while preserving flexibility.
	
	Because of this flexibility, it is not necessary in practice to assume a richer distribution, such as the Student distribution with additional degrees of freedom, to obtain a flexible model. For example, instead of assuming that $\varepsilon_i$ follows a $t$-distribution, which allows $\varepsilon_i$ to take large values, one can use a normal distribution with tightly spaced cut points to generate outcomes with large values as well. This feature is important because inference with the Student distribution can be challenging due to the complicated form of its pdf.
	
	\subsection{Identification of the Distribution of Agents' Private Characteristics\label{SM_Ident_F}}
	This section delves into the identifiability of the distribution function of agents' private characteristics $F_{\varepsilon}$. I impose the following assumptions. 
	\begin{assumption}\label{ass_indlin1}
		There is a variable $x_{i,\hat \kappa}$ in $\boldsymbol{x}_i$ taking any value in $\mathbb{R}$ given every value of $\boldsymbol{x}_{i,-\hat\kappa} = (x_{i,1}, \dots, x_{i,\hat\kappa-1}, x_{i,\hat\kappa+1}, \dots, x_{i,K})^{\prime}$, with $\beta_{1,\hat \kappa} \beta_{2,\hat \kappa} \geq 0$ and $\beta_{1,\hat \kappa}\ne0$, where $\beta_{1,\hat \kappa}$ is the coefficient of $x_{i,\hat \kappa}$ and $\beta_{2,\hat \kappa}$ is the coefficient of $\bar x_{i,\hat \kappa}$, the contextual variable associated with $x_{i,\hat\kappa}$.
	\end{assumption}
	\begin{assumption}\label{ass_eps}
		The location and scale parameters of $F_{\varepsilon}$ are set fixed such that $\mathbb E(\varepsilon_i) = 0$ and  $\mathbb V(\varepsilon_i) = 1$.
	\end{assumption}

	\noindent Assumption \ref{ass_indlin1} adapts Condition X3 of \citeoa{manski1988identificationoa} to my framework. It implies that $\beta_{1,\hat \kappa} x_{i,\hat \kappa} + \beta_{2,\hat \kappa} \bar{x}_{i,\hat \kappa}$ can take any value in $\mathbb R$ given $\boldsymbol{x}_{i,-\hat\kappa}$. Assumption \ref{ass_eps} is a standard restriction imposed in discrete models. As the intercept in $\mathbf X$ can absorb the mean of $\varepsilon_i$, the assumption of zero mean for $\varepsilon_i$ is inconsequential. Additionally, imposing a unit variance for $\varepsilon_i$ is not overly restrictive. The model remains unchanged when both $\varepsilon_i$ and $\boldsymbol \theta$ are multiplied by any positive scalar.
	
	\begin{proposition}\label{prop_ident:general}
		Under Assumptions \ref{ass_cost}--\ref{ass_density_e}, \ref{ass_indlin1}, and \ref{ass_eps}, $(\boldsymbol \theta, F_{\varepsilon})$ is identified.
	\end{proposition}
	\begin{proof}
		\noindent Assume that $(\boldsymbol \theta, F_{\varepsilon})$ and $(\tilde{\boldsymbol \theta}, \tilde F_{\varepsilon})$ are observationally equivalent. Following Definition \ref{def_equivalence},  this means that 
		$p (\mathbf y_s|\mathcal{A}_s, \mathbf Z_s) = \tilde p(\mathbf y_s|\mathcal{A}_s, \mathbf Z_s)$, where $p (\mathbf y_s|\mathcal{A}_s, \mathbf Z_s)$ and $\tilde p (\mathbf y_s|\mathcal{A}_s, \mathbf Z_s)$ are the distributions of $\mathbf y_s$ conditional on $\mathcal{A}_s$ and $\mathbf Z_s$  at $(\boldsymbol \theta, F_{\varepsilon})$ and $(\tilde{\boldsymbol \theta}, \tilde F_{\varepsilon})$, respectively. Both distributions also yield the same expected outcome $\mathbb E(\mathbf y_s|\mathcal{A}_s, \mathbf Z_s)$.
		
		The condition $p(\mathbf y_s|\mathcal{A}_s, \mathbf Z_s) = \tilde p(\mathbf y_s|\mathcal{A}_s, \mathbf Z_s)$ implies that $p_i(0|\mathcal{A}_s, \mathbf Z_s) = \tilde p_i(0|\mathcal{A}_s, \mathbf Z_s)$ for all $i$, where $p_i(.|\mathcal{A}_s, \mathbf Z_s)$ and $\tilde p_i(.|\mathcal{A}_s, \mathbf Z_s)$ are the conditional probability mass functions of $y_i$ given $\mathcal{A}_s$ and $\mathbf Z_s$ at $\boldsymbol\theta$ and $\tilde{\boldsymbol{\theta}}$, respectively. Thus,
		$F_{\varepsilon}\big(\bar{\mathcal{Y}}_i \boldsymbol\alpha + \boldsymbol{z}_i^{\prime}\boldsymbol{\beta}\big) = \tilde F_{\varepsilon}\big(\bar{\mathcal{Y}}_i \tilde{\boldsymbol\alpha}+ \boldsymbol{z}_i^{\prime}\tilde{\boldsymbol{\beta}}\big)$. As $\frac{1}{n}\sum_{i = 1}^n \tilde{\boldsymbol z}_i^{\prime} \tilde{\boldsymbol z}_i$ is a full-rank matrix and $F_{\varepsilon}(0) = \tilde F_{\varepsilon}(0) = 0.5$, $\boldsymbol \theta$ can be identified relative to $\tilde{\boldsymbol \theta}$. This is a direct implication of Proposition 2 of \citeoa{manski1988identificationoa}. Indeed, assume that $(\bar \tau \boldsymbol \alpha,  \bar \tau \boldsymbol \beta) \ne (\tilde{\boldsymbol\alpha},  \tilde{\boldsymbol{\beta}})$ for any $\bar \tau > 0$. Assumption \ref{ass_indlin1} implies that the following condition holds for some agents in sufficiently large samples:
		\begin{equation}\label{eq_setQb}
			\bar{\mathcal{Y}}_i\boldsymbol\alpha + \boldsymbol{z}_i^{\prime}\boldsymbol{\beta} < 0 \leq \bar{\mathcal{Y}}_i \tilde{\boldsymbol\alpha}  + \boldsymbol{z}_i^{\prime}\tilde{\boldsymbol{\beta}} \quad \text{or} \quad \bar{\mathcal{Y}}_i \tilde{\boldsymbol\alpha} +  \boldsymbol{z}_i^{\prime}\tilde{\boldsymbol{\beta}} < 0 \leq \bar{\mathcal{Y}}_i \boldsymbol\alpha + \boldsymbol{z}_i^{\prime}\boldsymbol{\beta}.
		\end{equation} 
		This is because $\bar{\mathcal{Y}}_i\boldsymbol\alpha + \boldsymbol{z}_i^{\prime}\boldsymbol{\beta}$ can take any value in $\mathbb R$ (by Assumption \ref{ass_indlin1}) and is not proportional to $\bar{\mathcal{Y}}_i \tilde{\boldsymbol\alpha}  + \boldsymbol{z}_i^{\prime}\tilde{\boldsymbol{\beta}}$. However, since $F_{\varepsilon}(0) = \tilde F_{\varepsilon}(0) = 0.5$ and $F_{\varepsilon}$ is strictly increasing (Assumption \ref{ass_density_e}), Equation \eqref{eq_setQb} implies that $F_{\varepsilon}\big(\bar{\mathcal{Y}}_i \boldsymbol\alpha + \boldsymbol{z}_i^{\prime}\boldsymbol{\beta}\big) \ne \tilde F_{\varepsilon}\big(\bar{\mathcal{Y}}_i \tilde{\boldsymbol\alpha}+ \boldsymbol{z}_i^{\prime}\tilde{\boldsymbol{\beta}}\big)$. Consequently, the condition $(\bar \tau \boldsymbol \alpha,  \bar \tau \boldsymbol \beta) \ne (\tilde{\boldsymbol\alpha},  \tilde{\boldsymbol{\beta}})$ for any $\bar\tau > 0$ cannot hold. There exists $\bar \tau > 0$ such that $(\bar \tau \boldsymbol \alpha,  \bar \tau \boldsymbol \beta) = (\tilde{\boldsymbol\alpha},  \tilde{\boldsymbol{\beta}})$, which means that $F_{\varepsilon}(u) = \tilde F_{\varepsilon}(\bar\tau u)$ for all $u\in\mathbb R$. As a result, $F_{\varepsilon}$ can be identified relative to $\tilde F_{\varepsilon}$ up to scale. As the scale parameter is set to one, it follows that $\tau = 1$, $(\boldsymbol \alpha,  \boldsymbol \beta) = (\tilde{\boldsymbol\alpha},  \tilde{\boldsymbol{\beta}})$, and  $F_{\varepsilon} = \tilde F_{\varepsilon}$. Hence, $(\boldsymbol \alpha,  \boldsymbol \beta, F_{\varepsilon})$ is identified.  The threshold parameters are identified as in Proposition \ref{prop_ident}.
	\end{proof}

	\section{Supplemental Results on the Empirical Application}
	\subsection{Data Summary}\label{oa_data}
	
	\noindent Table \ref{data_sum} presents a data summary.  On average, students participate in 2.4 extracurricular activities, with females participating in 2.5 and males in 2.2.
	
	\begin{table}[!ht] 
		\footnotesize
		\centering 
		\caption{Data Summary \label{data_sum}} 
		\begin{threeparttable}
			\begin{tabular}{@{\extracolsep{5pt}} lld{3}d{3}rrrrr} 
				\toprule
				\multicolumn{2}{l}{Variable} & \multicolumn{1}{c}{Mean} & \multicolumn{1}{c}{Sd.} & Min & Max \\ 
				\midrule
				\multicolumn{2}{l}{\textbf{Number of activities}}  & 2.353 & 2.406 & 0 & 33 \\ 
				& For males & 2.240 & 2.596 & 0 & 33 \\ 
				& For females & 2.465 & 2.197 & 0 & 33 \\ [2ex]
				\multicolumn{2}{l}{\textbf{Control variables}}\\
				\multicolumn{2}{l}{Age}  & 15.010 & 1.709 & 10  & 19 \\ 
				\multicolumn{2}{l}{Sex} &  &  &  &    \\ 
				&\textit{Female}  & 0.503& 0.500& 0& 1  \\ 
				&Male  & 0.497& 0.500& 1  \\ 
				\multicolumn{2}{l}{Hispanic}  & 0.168 & 0.374 & 0 & 1 \\ 
				\multicolumn{2}{l}{Race} &  &  &  &   \\ 
				& \textit{White} & 0.625 & 0.484 & 0  & 1 \\ 
				& Black & 0.185 & 0.388 & 0 & 1 \\ 
				& Asian & 0.071 & 0.256 & 0 & 1 \\ 
				& Other & 0.097 & 0.296 & 0 & 1 \\ 
				\multicolumn{2}{l}{Years at school}  & 2.490 & 1.413 & 1 & 6 \\ 
				\multicolumn{2}{l}{Living with parents} & 0.727 & 0.445 & 0 & 1 \\ 
				\multicolumn{2}{l}{Mother education}  &  &  &   &  \\ 
				& \textit{High} & 0.175 & 0.380 & 0 & 1 \\ 
				& <High & 0.302 & 0.459 & 0 & 1 \\ 
				& >High & 0.406 & 0.491 & 0 & 1 \\ 
				& Missing & 0.117 & 0.322 & 0 & 1 \\ 
				\multicolumn{2}{l}{Mother job}  &  &  &  &  \\ 
				& \textit{Stay at home} & 0.204 & 0.403 & 0 & 1 \\ 
				& Professional & 0.199 & 0.400 & 0 & 1 \\ 
				& Other & 0.425 & 0.494 & 0 & 1 \\ 
				& Missing & 0.172 & 0.377 & 0 & 1 \\ 
				\bottomrule
			\end{tabular} 
			\begin{tablenotes}[para,flushleft]
				\footnotesize
				Note: For the categorical explanatory variables, the level in italics is set as the reference level in the econometric models.
			\end{tablenotes}
		\end{threeparttable}
	\end{table}
	
	\subsection{Data Replications Using the Estimated Models} \label{oa_replication}
	To assess the model's performance in replicating the observed data, I generate the number of extracurricular activities using the estimated coefficients as parameter values (as in the Monte Carlo simulations). The matrix of control variables $\mathbf{X}$ and the network are kept fixed at their observed values in the sample. Figure \ref{AH_data_plot} displays the distribution of the observed count outcome alongside simulated outcomes.

	\begin{figure}[!htbp]
		\centering
		\includegraphics[scale = .8]{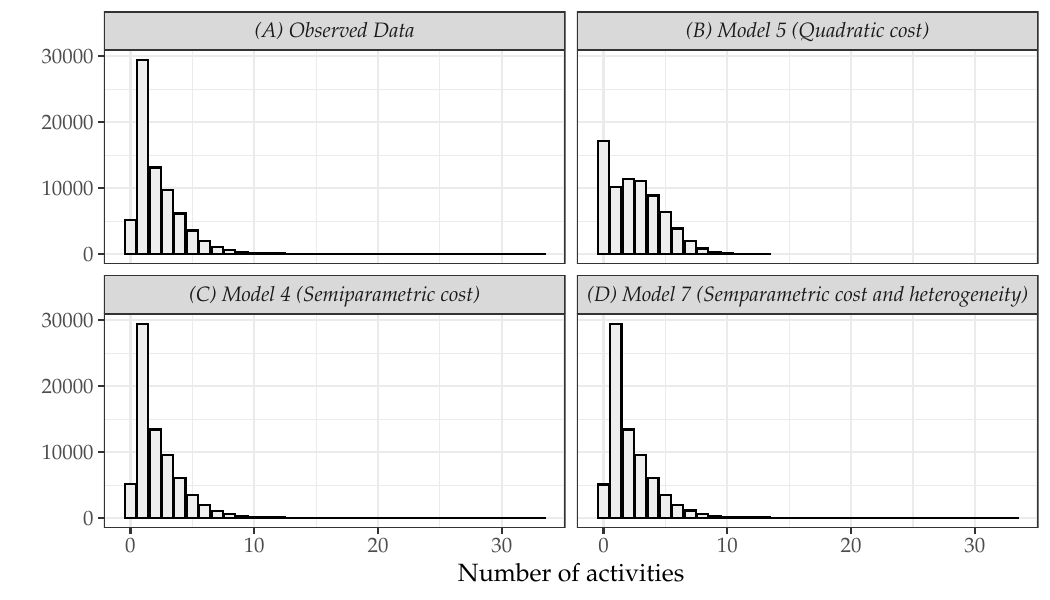}
		
		\vspace{-.2cm}
		\caption{Histograms of the observed and simulated dependent variables \label{AH_data_plot}}
		\justify
		\footnotesize
		Notes: This figure presents distributions of the number of extracurricular activities students attend, based on real data and simulations from estimated models. The simulations are obtained using the estimated coefficients from various models and by fixing the network and the control variables in $\mathbf{X}$ to their observed values from the empirical sample.
		
	\end{figure}
	
	\noindent Panel (A) represents the distribution of the observed data. The predictions are all based on the specifications controlling for unobserved school heterogeneity. Panel (B) shows the predicted data from the specification with a quadratic cost function; i.e., evenly spaced thresholds (Model 5). As in DGP A of the Monte Carlo study, the simulated data exhibit a left-truncated, symmetric distribution with a thin tail. This specification fails to replicate the observed data. Panel (C) presents the predicted data from the specification with a semiparametric cost function (Model 4), while Panel (D) displays the prediction using the specification with a semiparametric cost function and heterogeneity in peer effects. Both predictions are similar and appear to replicate the observed data because they rely on semiparametric costs.
	
	\subsection{Full Estimation Results}\label{oa_fullresults}
	This section presents the full results of the empirical application. Models 1--3 (Table \ref{app_noFE_full}) do not account for heterogeneity in the peer effects. Models 4--6 (Table \ref{app_FE_full}) are the fixed-effect versions of Models 1--3, respectively. I introduce heterogeneity in the peer effects in Models 7--8 (Table \ref{app_Het_full}).

	\begin{table}[!ht]
		\centering 
		\footnotesize
		\caption{Empirical results -- Homogeneous peer effects without fixed effects} 
		\label{app_noFE_full} 
		\resizebox{\columnwidth}{!}{\begin{threeparttable}
				\begin{tabular}{lld{3}d{3}cd{3}d{3}cd{3}d{3}c}
					\toprule
					&                     & \multicolumn{3}{c}{Model 1: Semiparametric cost} & \multicolumn{3}{c}{Model 2: Quadratic cost}             & \multicolumn{3}{c}{Model 3: Tobit}                       \\
					&                     & \multicolumn{1}{c}{Coef.}   & \multicolumn{2}{c}{Direct Marg. Eff.}   & \multicolumn{1}{c}{Coef.}      & \multicolumn{2}{c}{Direct Marg. Eff.}        & \multicolumn{1}{c}{Coef.}      & \multicolumn{2}{c}{Direct Marg. Eff.}         \\\midrule
					\multicolumn{2}{l}{Peer Effects}             & 0.043           & 0.082            & (0.026)           & 0.273         & 0.543             & (0.028)          & 0.681     & 0.553            & (0.017)           \\[1.5ex]
					\multicolumn{11}{l}{\textbf{Own effects}}                                                                                                                                                                                \\
					\multicolumn{2}{l}{Age}                      & -0.056          & -0.105           & (0.005)           & -0.008        & -0.017            & (0.007)          & -0.018    & -0.015           & (0.004)           \\
					\multicolumn{2}{l}{Male}                     & -0.147          & -0.275           & (0.018)           & -0.098        & -0.195            & (0.017)          & -0.246    & -0.200           & (0.010)           \\
					\multicolumn{2}{l}{Hispanic}                 & -0.062          & -0.116           & (0.024)           & 0.015         & 0.029             & (0.026)          & 0.025     & 0.020            & (0.017)           \\
					\multicolumn{2}{l}{Race}                     &                 &                  &                   &               &                   &                  &           &                  &                   \\
					& Black                      & 0.074           & 0.142            & (0.027)           & 0.103         & 0.206             & (0.032)          & 0.241     & 0.197            & (0.022)           \\
					& Asian                      & 0.192           & 0.386            & (0.036)           & 0.276         & 0.568             & (0.035)          & 0.668     & 0.561            & (0.025)           \\
					& Other                      & 0.054           & 0.103            & (0.029)           & 0.087         & 0.175             & (0.032)          & 0.211     & 0.173            & (0.019)           \\
					\multicolumn{2}{l}{Year in school}           & 0.050           & 0.094            & (0.007)           & 0.051         & 0.101             & (0.008)          & 0.123     & 0.100            & (0.005)           \\
					\multicolumn{2}{l}{Live with both   parents} & 0.085           & 0.157            & (0.015)           & 0.066         & 0.131             & (0.020)          & 0.161     & 0.130            & (0.011)           \\
					\multicolumn{2}{l}{Mother   education}       &                 &                  &                   &               &                   &                  &           &                  &                   \\
					& $<$ High                   & -0.077          & -0.141           & (0.018)           & -0.027        & -0.054            & (0.021)          & -0.068    & -0.055           & (0.014)           \\
					& $>$ High                   & 0.203           & 0.385            & (0.020)           & 0.155         & 0.311             & (0.023)          & 0.381     & 0.311            & (0.012)           \\
					& Missing                    & 0.027           & 0.051            & (0.029)           & 0.091         & 0.184             & (0.030)          & 0.211     & 0.173            & (0.022)           \\
					\multicolumn{2}{l}{Mother job}               &                 &                  &                   &               &                   &                  &           &                  &                   \\
					& Professional               & 0.125           & 0.242            & (0.025)           & 0.087         & 0.175             & (0.027)          & 0.216     & 0.177            & (0.015)           \\
					& Other                      & 0.037           & 0.069            & (0.020)           & 0.024         & 0.047             & (0.020)          & 0.061     & 0.049            & (0.013)           \\
					& Missing                    & -0.053          & -0.099           & (0.025)           & -0.034        & -0.067            & (0.028)          & -0.085    & -0.069           & (0.019)           \\[1.5ex]
					\multicolumn{11}{l}{\textbf{Contextual effects}}                                                                                                                                                                         \\
					\multicolumn{2}{l}{Age}                      & -0.006          & -0.011           & (0.004)           & -0.032        & -0.063            & (0.005)          & -0.077    & -0.062           & (0.003)           \\
					\multicolumn{2}{l}{Male}                     & -0.013          & -0.025           & (0.029)           & 0.044         & 0.087             & (0.030)          & 0.106     & 0.086            & (0.016)           \\
					\multicolumn{2}{l}{Hispanic}                 & -0.139          & -0.262           & (0.039)           & -0.064        & -0.127            & (0.038)          & -0.148    & -0.120           & (0.026)           \\
					\multicolumn{2}{l}{Race}                     &                 &                  &                   &               &                   &                  &           &                  &                   \\
					& Black                      & -0.019          & -0.036           & (0.031)           & -0.069        & -0.137            & (0.039)          & -0.161    & -0.131           & (0.025)           \\
					& Asian                      & -0.118          & -0.222           & (0.049)           & -0.242        & -0.482            & (0.047)          & -0.587    & -0.477           & (0.033)           \\
					& Other                      & -0.071          & -0.135           & (0.048)           & -0.115        & -0.229            & (0.050)          & -0.279    & -0.226           & (0.034)           \\
					\multicolumn{2}{l}{Year in school}           & 0.039           & 0.074            & (0.009)           & -0.011        & -0.022            & (0.009)          & -0.029    & -0.024           & (0.007)           \\
					\multicolumn{2}{l}{Live with both   parents} & 0.178           & 0.335            & (0.035)           & 0.030         & 0.059             & (0.038)          & 0.069     & 0.056            & (0.023)           \\
					\multicolumn{2}{l}{Mother   education}       &                 &                  &                   &               &                   &                  &           &                  &                   \\
					& $<$ High                   & -0.154          & -0.291           & (0.047)           & -0.092        & -0.182            & (0.037)          & -0.226    & -0.183           & (0.029)           \\
					& $>$ High                   & 0.195           & 0.368            & (0.037)           & 0.010         & 0.020             & (0.038)          & 0.015     & 0.012            & (0.024)           \\
					& Missing                    & -0.055          & -0.103           & (0.057)           & -0.101        & -0.201            & (0.058)          & -0.253    & -0.206           & (0.041)           \\
					\multicolumn{2}{l}{Mother job}               &                 &                  &                   &               &                   &                  &           &                  &                   \\
					& Professional               & 0.195           & 0.367            & (0.046)           & 0.040         & 0.080             & (0.048)          & 0.092     & 0.074            & (0.030)           \\
					& Other                      & 0.071           & 0.134            & (0.037)           & -0.002        & -0.003            & (0.037)          & -0.007    & -0.006           & (0.021)           \\
					& Missing                    & -0.012          & -0.023           & (0.053)           & -0.012        & -0.024            & (0.051)          & -0.026    & -0.021           & (0.036)           \\\midrule
					\multicolumn{2}{l}{$\bar R_c$}               & \multicolumn{3}{c}{13}                           & \multicolumn{3}{c}{1}                                    &            &                      &                       \\
					\multicolumn{2}{l}{Log likelihood}           & \multicolumn{3}{c}{$-127,626$}                      & \multicolumn{3}{c}{$-159,924$}                              & \multicolumn{3}{c}{$-161,225$}                               \\\bottomrule
				\end{tabular} 
				\begin{tablenotes}[para,flushleft]
					\footnotesize
					Notes: Columns \textit{Coef.} indicate the estimates for the coefficients, whereas Columns \textit{Direct Marg. Eff.} report the direct marginal effects and their corresponding standard errors in parentheses. Indirect and total marginal effects are presented in Table \ref{app_noFE_ME}. $\bar R_c$ indicates the value at which the cost function switches from nonparametric to quadratic. For a fully quadratic cost function, $\bar R_c = 1$. For the semiparametric cost approach, the best $\bar R_c$ is determined by minimizing the BIC. The sample used in this empirical application consists of $n = \text{72,291}$ students from $S = 120$ schools.
				\end{tablenotes}
		\end{threeparttable}}
	\end{table}

	\begin{table}[!ht]
		\centering 
		\footnotesize
		\caption{Empirical results -- Homogeneous peer effects without fixed effects (indirect and total marginal effects)} 
		\label{app_noFE_ME} 
		\begin{threeparttable}
			\begin{tabular}{lld{5}d{5}d{5}d{5}d{5}d{5}}
				\toprule
				&                              & \multicolumn{2}{c}{Model 1} & \multicolumn{2}{c}{Model 2} & \multicolumn{2}{c}{Model 3} \\\midrule
				\multicolumn{2}{l}{\textbf{Indirect Marginal   Effects}} &              &              &              &              &              &              \\
				\multicolumn{2}{l}{Age}                                  & -0.019       & (0.006)      & -0.121       & (0.008)      & -0.121       & (0.006)      \\
				\multicolumn{2}{l}{Male}                                 & -0.046       & (0.025)      & -0.034       & (0.041)      & -0.045       & (0.026)      \\
				\multicolumn{2}{l}{Hispanic}                             & -0.233       & (0.033)      & -0.181       & (0.049)      & -0.182       & (0.035)      \\
				\multicolumn{2}{l}{Race}                                 &              &              &              &              &              &              \\
				& Black                            & -0.022       & (0.028)      & -0.040       & (0.042)      & -0.034       & (0.028)      \\
				& Asian                            & -0.196       & (0.047)      & -0.306       & (0.059)      & -0.300       & (0.043)      \\
				& Other                            & -0.120       & (0.044)      & -0.223       & (0.073)      & -0.222       & (0.055)      \\
				\multicolumn{2}{l}{Year in school}                       & 0.077        & (0.008)      & 0.056        & (0.010)      & 0.056        & (0.008)      \\
				\multicolumn{2}{l}{Live with both   parents}             & 0.301        & (0.028)      & 0.211        & (0.050)      & 0.212        & (0.035)      \\
				\multicolumn{2}{l}{Mother   education}                   &              &              &              &              &              &              \\
				& $<$ High                         & -0.257       & (0.040)      & -0.335       & (0.056)      & -0.344       & (0.043)      \\
				& $>$ High                         & 0.386        & (0.038)      & 0.325        & (0.044)      & 0.326        & (0.038)      \\
				& Missing                          & -0.092       & (0.053)      & -0.169       & (0.090)      & -0.187       & (0.066)      \\
				\multicolumn{2}{l}{Mother job}                           &              &              &              &              &              &              \\
				& Professional                     & 0.409        & (0.052)      & 0.312        & (0.071)      & 0.316        & (0.053)      \\
				& Other                            & 0.133        & (0.034)      & 0.038        & (0.056)      & 0.038        & (0.034)      \\
				& Missing                          & -0.029       & (0.048)      & -0.099       & (0.074)      & -0.099       & (0.060)      \\[1.5ex]
				\multicolumn{2}{l}{\textbf{Total Marginal   Effects}}    &              &              &              &              &              &              \\
				\multicolumn{2}{l}{Age}                                  & -0.125       & (0.007)      & -0.138       & (0.006)      & -0.136       & (0.005)      \\
				\multicolumn{2}{l}{Male}                                 & -0.321       & (0.025)      & -0.229       & (0.037)      & -0.245       & (0.027)      \\
				\multicolumn{2}{l}{Hispanic}                             & -0.349       & (0.033)      & -0.152       & (0.047)      & -0.161       & (0.036)      \\
				\multicolumn{2}{l}{Race}                                 &              &              &              &              &              &              \\
				& Black                            & 0.120        & (0.025)      & 0.166        & (0.031)      & 0.163        & (0.023)      \\
				& Asian                            & 0.190        & (0.033)      & 0.262        & (0.050)      & 0.261        & (0.038)      \\
				& Other                            & -0.016       & (0.047)      & -0.048       & (0.074)      & -0.049       & (0.059)      \\
				\multicolumn{2}{l}{Year in school}                       & 0.171        & (0.009)      & 0.157        & (0.007)      & 0.156        & (0.007)      \\
				\multicolumn{2}{l}{Live with both   parents}             & 0.458        & (0.033)      & 0.342        & (0.049)      & 0.342        & (0.038)      \\
				\multicolumn{2}{l}{Mother   education}                   &              &              &              &              &              &              \\
				& $<$ High                         & -0.398       & (0.044)      & -0.388       & (0.057)      & -0.399       & (0.046)      \\
				& $>$ High                         & 0.771        & (0.043)      & 0.635        & (0.047)      & 0.637        & (0.037)      \\
				& Missing                          & -0.041       & (0.061)      & 0.015        & (0.093)      & -0.013       & (0.069)      \\
				\multicolumn{2}{l}{Mother job}                           &              &              &              &              &              &              \\
				& Professional                     & 0.652        & (0.058)      & 0.487        & (0.069)      & 0.493        & (0.056)      \\
				& Other                            & 0.203        & (0.037)      & 0.085        & (0.056)      & 0.088        & (0.038)      \\
				& Missing                          & -0.127       & (0.054)      & -0.166       & (0.072)      & -0.168       & (0.066)     \\\bottomrule
			\end{tabular}
			\begin{tablenotes}[para,flushleft]
				\footnotesize
				Notes: Estimates without parentheses are marginal effects; standard errors are shown in parentheses.
			\end{tablenotes}
		\end{threeparttable}
	\end{table}

	\begin{table}[!ht]
		\centering 
		\footnotesize
		\caption{Empirical results -- Homogeneous peer effects with fixed effects} 
		\label{app_FE_full} 
		\resizebox{\columnwidth}{!}{\begin{threeparttable}
				\begin{tabular}{lld{3}d{3}cd{3}d{3}cd{3}d{3}c}
					\toprule
					&                     & \multicolumn{3}{c}{Model 4: Semiparametric cost} & \multicolumn{3}{c}{Model 5: Quadratic cost}             & \multicolumn{3}{c}{Model 6: Tobit}                       \\
					&                     & \multicolumn{1}{c}{Coef.}   & \multicolumn{2}{c}{Marginal effects}   & \multicolumn{1}{c}{Coef.}      & \multicolumn{2}{c}{Marginal effects}        & \multicolumn{1}{c}{Coef.}      & \multicolumn{2}{c}{Marginal effects}         \\\midrule
					\multicolumn{2}{l}{Peer Effects}             & 0.043           & 0.079            & (0.023)           & 0.188         & 0.369             & (0.029)          & 0.446     & 0.362            & (0.020)           \\[1.5ex]
					\multicolumn{11}{l}{\textbf{Own effects}}                                                                                                                                                                                \\
					\multicolumn{2}{l}{Age}                      & -0.051          & -0.095           & (0.008)           & -0.019        & -0.038            & (0.008)          & -0.048    & -0.039           & (0.005)           \\
					\multicolumn{2}{l}{Male}                     & -0.160          & -0.296           & (0.019)           & -0.106        & -0.208            & (0.017)          & -0.262    & -0.213           & (0.011)           \\
					\multicolumn{2}{l}{Hispanic}                 & -0.007          & -0.013           & (0.024)           & 0.051         & 0.101             & (0.026)          & 0.114     & 0.093            & (0.018)           \\
					\multicolumn{2}{l}{Race}                     &                 &                  &                   &               &                   &                  &           &                  &                   \\
					& Black                      & 0.133           & 0.255            & (0.038)           & 0.129         & 0.256             & (0.033)          & 0.303     & 0.249            & (0.020)           \\
					& Asian                      & 0.221           & 0.443            & (0.042)           & 0.293         & 0.596             & (0.039)          & 0.699     & 0.588            & (0.024)           \\
					& Other                      & 0.065           & 0.122            & (0.028)           & 0.092         & 0.182             & (0.029)          & 0.220     & 0.180            & (0.019)           \\
					\multicolumn{2}{l}{Year in school}           & 0.043           & 0.080            & (0.008)           & 0.050         & 0.098             & (0.008)          & 0.120     & 0.097            & (0.005)           \\
					\multicolumn{2}{l}{Live with both   parents} & 0.083           & 0.152            & (0.016)           & 0.066         & 0.129             & (0.018)          & 0.159     & 0.128            & (0.014)           \\
					\multicolumn{2}{l}{Mother   education}       &                 &                  &                   &               &                   &                  &           &                  &                   \\
					& $<$ High                   & -0.061          & -0.111           & (0.019)           & -0.019        & -0.036            & (0.024)          & -0.047    & -0.038           & (0.015)           \\
					& $>$ High                   & 0.210           & 0.391            & (0.024)           & 0.163         & 0.323             & (0.017)          & 0.396     & 0.323            & (0.012)           \\
					& Missing                    & 0.032           & 0.060            & (0.029)           & 0.096         & 0.191             & (0.031)          & 0.220     & 0.180            & (0.023)           \\
					\multicolumn{2}{l}{Mother job}               &                 &                  &                   &               &                   &                  &           &                  &                   \\
					& Professional               & 0.140           & 0.267            & (0.026)           & 0.098         & 0.195             & (0.026)          & 0.241     & 0.197            & (0.014)           \\
					& Other                      & 0.047           & 0.087            & (0.018)           & 0.029         & 0.057             & (0.021)          & 0.072     & 0.058            & (0.013)           \\
					& Missing                    & -0.039          & -0.072           & (0.028)           & -0.027        & -0.052            & (0.029)          & -0.069    & -0.056           & (0.018)           \\[1.5ex]
					\multicolumn{11}{l}{\textbf{Contextual effects}}                                                                                                                                                                         \\
					\multicolumn{2}{l}{Age}                      & -0.007          & -0.013           & (0.004)           & -0.027        & -0.053            & (0.005)          & -0.062    & -0.051           & (0.003)           \\
					\multicolumn{2}{l}{Male}                     & -0.032          & -0.059           & (0.028)           & 0.015         & 0.029             & (0.029)          & 0.026     & 0.021            & (0.020)           \\
					\multicolumn{2}{l}{Hispanic}                 & -0.063          & -0.117           & (0.040)           & -0.020        & -0.039            & (0.039)          & -0.046    & -0.037           & (0.028)           \\
					\multicolumn{2}{l}{Race}                     &                 &                  &                   &               &                   &                  &           &                  &                   \\
					& Black                      & 0.005           & 0.009            & (0.042)           & -0.036        & -0.070            & (0.040)          & -0.075    & -0.061           & (0.024)           \\
					& Asian                      & -0.018          & -0.032           & (0.052)           & -0.140        & -0.274            & (0.055)          & -0.325    & -0.264           & (0.037)           \\
					& Other                      & -0.070          & -0.130           & (0.050)           & -0.102        & -0.201            & (0.047)          & -0.242    & -0.196           & (0.035)           \\
					\multicolumn{2}{l}{Year in school}           & 0.021           & 0.040            & (0.011)           & -0.007        & -0.014            & (0.012)          & -0.016    & -0.013           & (0.007)           \\
					\multicolumn{2}{l}{Live with both   parents} & 0.168           & 0.310            & (0.035)           & 0.068         & 0.134             & (0.037)          & 0.170     & 0.138            & (0.023)           \\
					\multicolumn{2}{l}{Mother   education}       &                 &                  &                   &               &                   &                  &           &                  &                   \\
					& $<$ High                   & -0.131          & -0.244           & (0.037)           & -0.075        & -0.147            & (0.047)          & -0.183    & -0.149           & (0.028)           \\
					& $>$ High                   & 0.212           & 0.393            & (0.033)           & 0.078         & 0.154             & (0.039)          & 0.193     & 0.156            & (0.024)           \\
					& Missing                    & -0.064          & -0.118           & (0.059)           & -0.075        & -0.147            & (0.059)          & -0.183    & -0.148           & (0.040)           \\
					\multicolumn{2}{l}{Mother job}               &                 &                  &                   &               &                   &                  &           &                  &                   \\
					& Professional               & 0.223           & 0.413            & (0.045)           & 0.107         & 0.210             & (0.048)          & 0.263     & 0.213            & (0.033)           \\
					& Other                      & 0.083           & 0.154            & (0.034)           & 0.032         & 0.062             & (0.036)          & 0.080     & 0.065            & (0.028)           \\
					& Missing                    & 0.015           & 0.028            & (0.050)           & 0.023         & 0.045             & (0.052)          & 0.062     & 0.050            & (0.035)           \\\midrule
					\multicolumn{2}{l}{$\bar R_c$}               & \multicolumn{3}{c}{13}                           & \multicolumn{3}{c}{1}                                    &            &                      &                       \\
					\multicolumn{2}{l}{Log likelihood}           & \multicolumn{3}{c}{$-126,111$}                      & \multicolumn{3}{c}{$-158,964$}                              & \multicolumn{3}{c}{$-160,258$}                               \\\bottomrule
				\end{tabular} 
				\begin{tablenotes}[para,flushleft]
					\footnotesize
					Notes: Columns \textit{Coef.} indicate the estimates for the coefficients, whereas Columns \textit{Direct Marg. Eff.} report the direct marginal effects and their corresponding standard errors in parentheses. Indirect and total marginal effects are presented in Table \ref{app_FE_ME}. $\bar R_c$ indicates the value at which the cost function switches from nonparametric to quadratic. For a fully quadratic cost function, $\bar R_c = 1$. For the semiparametric cost approach, the best $\bar R_c$ is determined by minimizing the BIC. The sample used in this empirical application consists of $n = \text{72,291}$ students from $S = 120$ schools.
				\end{tablenotes}
		\end{threeparttable}}
	\end{table}

	\begin{table}[!ht]
		\centering 
		\footnotesize
		\caption{Empirical results -- Homogeneous peer effects with fixed effects (indirect and total marginal effects)} 
		\label{app_FE_ME} 
		\begin{threeparttable}
			\begin{tabular}{lld{5}d{5}d{5}d{5}d{5}d{5}}
				\toprule
				&                              & \multicolumn{2}{c}{Model 4} & \multicolumn{2}{c}{Model 5} & \multicolumn{2}{c}{Model 6} \\\midrule
				\multicolumn{2}{l}{\textbf{Indirect Marginal   Effects}} &              &              &              &              &              &              \\
				\multicolumn{2}{l}{Age}                                  & -0.020       & (0.006)      & -0.020       & (0.006)      & -0.081       & (0.006)      \\
				\multicolumn{2}{l}{Male}                                 & -0.080       & (0.026)      & -0.080       & (0.026)      & -0.071       & (0.021)      \\
				\multicolumn{2}{l}{Hispanic}                             & -0.108       & (0.035)      & -0.108       & (0.035)      & -0.004       & (0.033)      \\
				\multicolumn{2}{l}{Race}                                 &              &              &              &              &              &              \\
				& Black                            & 0.033        & (0.041)      & 0.033        & (0.041)      & 0.040        & (0.026)      \\
				& Asian                            & 0.009        & (0.053)      & 0.009        & (0.053)      & -0.063       & (0.045)      \\
				& Other                            & -0.115       & (0.047)      & -0.115       & (0.047)      & -0.164       & (0.043)      \\
				\multicolumn{2}{l}{Year in school}                       & 0.044        & (0.011)      & 0.044        & (0.011)      & 0.029        & (0.008)      \\
				\multicolumn{2}{l}{Live with both   parents}             & 0.281        & (0.030)      & 0.281        & (0.030)      & 0.224        & (0.024)      \\
				\multicolumn{2}{l}{Mother   education}                   &              &              &              &              &              &              \\
				& $<$ High                         & -0.219       & (0.032)      & -0.219       & (0.032)      & -0.198       & (0.033)      \\
				& $>$ High                         & 0.412        & (0.037)      & 0.412        & (0.037)      & 0.348        & (0.027)      \\
				& Missing                          & -0.106       & (0.054)      & -0.106       & (0.054)      & -0.103       & (0.049)      \\
				\multicolumn{2}{l}{Mother job}                           &              &              &              &              &              &              \\
				& Professional                     & 0.470        & (0.057)      & 0.470        & (0.057)      & 0.374        & (0.042)      \\
				& Other                            & 0.155        & (0.032)      & 0.155        & (0.032)      & 0.108        & (0.033)      \\
				& Missing                          & 0.020        & (0.047)      & 0.020        & (0.047)      & 0.037        & (0.043)      \\[1.5ex]
				\multicolumn{2}{l}{\textbf{Total Marginal   Effects}}    &              &              &              &              &              &              \\
				\multicolumn{2}{l}{Age}                                  & -0.115       & (0.011)      & -0.115       & (0.011)      & -0.119       & (0.007)      \\
				\multicolumn{2}{l}{Male}                                 & -0.376       & (0.030)      & -0.376       & (0.030)      & -0.284       & (0.021)      \\
				\multicolumn{2}{l}{Hispanic}                             & -0.121       & (0.041)      & -0.121       & (0.041)      & 0.089        & (0.035)      \\
				\multicolumn{2}{l}{Race}                                 &              &              &              &              &              &              \\
				& Black                            & 0.288        & (0.038)      & 0.288        & (0.038)      & 0.289        & (0.028)      \\
				& Asian                            & 0.452        & (0.059)      & 0.452        & (0.059)      & 0.525        & (0.046)      \\
				& Other                            & 0.008        & (0.050)      & 0.008        & (0.050)      & 0.017        & (0.045)      \\
				\multicolumn{2}{l}{Year in school}                       & 0.125        & (0.012)      & 0.125        & (0.012)      & 0.126        & (0.009)      \\
				\multicolumn{2}{l}{Live with both   parents}             & 0.433        & (0.033)      & 0.433        & (0.033)      & 0.352        & (0.028)      \\
				\multicolumn{2}{l}{Mother   education}                   &              &              &              &              &              &              \\
				& $<$ High                         & -0.330       & (0.039)      & -0.330       & (0.039)      & -0.235       & (0.033)      \\
				& $>$ High                         & 0.803        & (0.049)      & 0.803        & (0.049)      & 0.671        & (0.030)      \\
				& Missing                          & -0.046       & (0.059)      & -0.046       & (0.059)      & 0.077        & (0.055)      \\
				\multicolumn{2}{l}{Mother job}                           &              &              &              &              &              &              \\
				& Professional                     & 0.737        & (0.069)      & 0.737        & (0.069)      & 0.571        & (0.044)      \\
				& Other                            & 0.241        & (0.037)      & 0.241        & (0.037)      & 0.166        & (0.036)      \\
				& Missing                          & -0.052       & (0.056)      & -0.052       & (0.056)      & -0.019       & (0.046)     \\\bottomrule
			\end{tabular}
			\begin{tablenotes}[para,flushleft]
				\footnotesize
				Notes: Estimates without parentheses are marginal effects; standard errors are shown in parentheses.
			\end{tablenotes}
		\end{threeparttable}
	\end{table}

	\begin{table}[!ht]
		\centering 
		\footnotesize
		\caption{Empirical results -- Heterogeneous peer effects with fixed effects} 
		\label{app_Het_full} 
		\begin{threeparttable}
			\begin{tabular}{lld{3}d{3}cd{3}d{3}cd{3}}
				\toprule
				&                                      & \multicolumn{3}{c}{Model 7: Semiparametric cost}        & \multicolumn{3}{c}{Model 8: Quadratic cost} \\
				&                     & \multicolumn{1}{c}{Coef.}   & \multicolumn{2}{c}{Marginal effects}   & \multicolumn{1}{c}{Coef.}      & \multicolumn{2}{c}{Marginal effects}    \\\midrule
				\multicolumn{2}{l}{Peer Effects$^{\text{Male},\text{Male}}$}     & -0.024  & -0.022            & (0.007)          & 0.032   & 0.033             & (0.008)          \\
				\multicolumn{2}{l}{Peer Effects$^{\text{Male},\text{Female}}$}   & 0.071   & 0.066             & (0.008)          & 0.080   & 0.083             & (0.007)          \\
				\multicolumn{2}{l}{Peer Effects$^{\text{Female},\text{Male}}$}   & 0.069   & 0.061             & (0.006)          & 0.072   & 0.066             & (0.006)          \\
				\multicolumn{2}{l}{Peer Effects$^{\text{Female},\text{Female}}$} & 0.034   & 0.030             & (0.008)          & 0.132   & 0.121             & (0.008)          \\[1.5ex]
				\multicolumn{8}{l}{\textbf{Own effects}}                                                                                                                           \\
				\multicolumn{2}{l}{Age}                                          & -0.050  & -0.091            & (0.008)          & -0.027  & -0.052            & (0.007)          \\
				\multicolumn{2}{l}{Male}                                         & -0.224  & -0.417            & (0.042)          & -0.115  & -0.224            & (0.034)          \\
				\multicolumn{2}{l}{Hispanic}                                     & -0.008  & -0.015            & (0.022)          & 0.050   & 0.098             & (0.029)          \\
				\multicolumn{2}{l}{Race}                                         &         &                   &                  &         &                   &                  \\
				& Black                                & 0.138   & 0.260             & (0.032)          & 0.137   & 0.270             & (0.031)          \\
				& Asian                                & 0.227   & 0.445             & (0.038)          & 0.295   & 0.594             & (0.035)          \\
				& Other                                & 0.067   & 0.125             & (0.026)          & 0.096   & 0.188             & (0.028)          \\
				\multicolumn{2}{l}{Year in school}                               & 0.042   & 0.077             & (0.008)          & 0.048   & 0.093             & (0.008)          \\
				\multicolumn{2}{l}{Live with both parents}                       & 0.081   & 0.145             & (0.018)          & 0.065   & 0.125             & (0.018)          \\
				\multicolumn{2}{l}{Mother education}                             &         &                   &                  &         &                   &                  \\
				& $<$ High                             & -0.059  & -0.106            & (0.022)          & -0.017  & -0.034            & (0.023)          \\
				& $>$ High                             & 0.209   & 0.382             & (0.020)          & 0.167   & 0.326             & (0.020)          \\
				& Missing                              & 0.031   & 0.057             & (0.031)          & 0.095   & 0.187             & (0.035)          \\
				\multicolumn{2}{l}{Mother job}                                   &         &                   &                  &         &                   &                  \\
				& Professional                         & 0.134   & 0.250             & (0.024)          & 0.094   & 0.185             & (0.025)          \\
				& Other                                & 0.042   & 0.078             & (0.018)          & 0.025   & 0.049             & (0.021)          \\
				& Missing                              & -0.039  & -0.070            & (0.027)          & -0.024  & -0.047            & (0.029)          \\[1.5ex]
				\multicolumn{8}{l}{\textbf{Contextual effects}}                                                                                                                    \\
				\multicolumn{2}{l}{Age}                                          & -0.009  & -0.016            & (0.004)          & -0.023  & -0.045            & (0.004)          \\
				\multicolumn{2}{l}{Male}                                         & 0.006   & 0.011             & (0.045)          & 0.073   & 0.142             & (0.053)          \\
				\multicolumn{2}{l}{Hispanic}                                     & -0.061  & -0.111            & (0.039)          & -0.027  & -0.053            & (0.043)          \\
				\multicolumn{2}{l}{Race}                                         &         &                   &                  &         &                   &                  \\
				& Black                                & 0.011   & 0.019             & (0.036)          & -0.014  & -0.027            & (0.039)          \\
				& Asian                                & -0.003  & -0.005            & (0.046)          & -0.090  & -0.175            & (0.053)          \\
				& Other                                & -0.066  & -0.121            & (0.046)          & -0.085  & -0.166            & (0.050)          \\
				\multicolumn{2}{l}{Year in school}                               & 0.018   & 0.034             & (0.012)          & -0.001  & -0.002            & (0.010)          \\
				\multicolumn{2}{l}{Live with both parents}                       & 0.166   & 0.303             & (0.033)          & 0.095   & 0.185             & (0.034)          \\
				\multicolumn{2}{l}{Mother education}                             &         &                   &                  &         &                   &                  \\
				& $<$ High                             & -0.122  & -0.222            & (0.042)          & -0.069  & -0.134            & (0.043)          \\
				& $>$ High                             & 0.206   & 0.375             & (0.034)          & 0.109   & 0.212             & (0.031)          \\
				& Missing                              & -0.057  & -0.103            & (0.060)          & -0.047  & -0.090            & (0.062)          \\
				\multicolumn{2}{l}{Mother job}                                   &         &                   &                  &         &                   &                  \\
				& Professional                         & 0.213   & 0.388             & (0.042)          & 0.130   & 0.253             & (0.046)          \\
				& Other                                & 0.076   & 0.139             & (0.034)          & 0.040   & 0.079             & (0.037)          \\
				& Missing                              & 0.014   & 0.025             & (0.054)          & 0.026   & 0.051             & (0.052)          \\\midrule
				\multicolumn{2}{l}{$\bar R_c$}                                   & \multicolumn{3}{c}{10}                         & \multicolumn{3}{c}{1}                         \\
				\multicolumn{2}{l}{Log likelihood}                               & \multicolumn{3}{c}{$-125,667$}                    & \multicolumn{3}{c}{$-125,096$}  \\\bottomrule
			\end{tabular} 
			\begin{tablenotes}[para,flushleft]
				\footnotesize
				Notes: Columns \textit{Coef.} indicate the estimates for the coefficients, whereas Columns \textit{Marginal effects} report the direct marginal effects and their corresponding standard errors in parentheses. Indirect and total marginal effects are presented in Table \ref{app_Het_ME}. $\bar R_c$ indicates the value at which the cost function switches from nonparametric to quadratic. The best $\bar R_c$ is determined by minimizing the BIC. Peer Effects$^{g,g^{\prime}}$ indicates the effects of peers in group $\mathcal G_{g^{\prime}}$ on students in group $\mathcal G_g$, where $g,g^{\prime}\in \{\text{Male},\text{Female}\}$. The sample used in this empirical application consists of $n = \text{72,291}$ students from $S = 120$ schools.
			\end{tablenotes}
		\end{threeparttable}
	\end{table}
	
	\clearpage
	\begin{table}[!ht]
		\centering 
		\footnotesize
		\caption{Empirical results -- Heterogeneous peer effects with fixed effects (indirect and total marginal effects)} 
		\label{app_Het_ME} 
		\begin{threeparttable}
			\begin{tabular}{lld{5}d{5}d{5}d{5}}
				\toprule
				&                              & \multicolumn{2}{c}{Model 7} & \multicolumn{2}{c}{Model 8}  \\\midrule
				\multicolumn{2}{l}{\textbf{Indirect Marginal   Effects}} &              &              &              &                   \\
				\multicolumn{2}{l}{Age}                                & -0.026       & (0.005)      & -0.064       & (0.006)      \\
				\multicolumn{2}{l}{Male}                               & -0.034       & (0.044)      & 0.089        & (0.056)      \\
				\multicolumn{2}{l}{Hispanic}                           & -0.107       & (0.035)      & -0.029       & (0.045)      \\
				\multicolumn{2}{l}{Race}                               &              &              &              &              \\
				& Black                           & 0.055        & (0.035)      & 0.052        & (0.038)      \\
				& Asian                           & 0.059        & (0.051)      & -0.017       & (0.057)      \\
				& Other                           & -0.105       & (0.044)      & -0.127       & (0.054)      \\
				\multicolumn{2}{l}{Year in school}                     & 0.042        & (0.011)      & 0.025        & (0.010)      \\
				\multicolumn{2}{l}{Live with both parents}             & 0.287        & (0.030)      & 0.233        & (0.036)      \\
				\multicolumn{2}{l}{Mother education}                   &              &              &              &              \\
				& $<$ High                        & -0.211       & (0.037)      & -0.154       & (0.044)      \\
				& $>$ High                        & 0.421        & (0.039)      & 0.333        & (0.034)      \\
				& Missing                         & -0.094       & (0.057)      & -0.044       & (0.067)      \\
				\multicolumn{2}{l}{Mother job}                         &              &              &              &              \\
				& Professional                    & 0.462        & (0.056)      & 0.344        & (0.052)      \\
				& Other                           & 0.146        & (0.034)      & 0.101        & (0.041)      \\
				& Missing                         & 0.015        & (0.051)      & 0.042        & (0.058)      \\[1.5ex]
				\multicolumn{2}{l}{\textbf{Total Marginal   Effects}}    &              &              &              &               \\
				\multicolumn{2}{l}{Age}                                & -0.118       & (0.010)      & -0.117       & (0.010)      \\
				\multicolumn{2}{l}{Male}                               & -0.451       & (0.061)      & -0.135       & (0.064)      \\
				\multicolumn{2}{l}{Hispanic}                           & -0.121       & (0.036)      & 0.069        & (0.047)      \\
				\multicolumn{2}{l}{Race}                               &              &              &              &              \\
				& Black                           & 0.315        & (0.034)      & 0.322        & (0.034)      \\
				& Asian                           & 0.504        & (0.065)      & 0.577        & (0.055)      \\
				& Other                           & 0.020        & (0.050)      & 0.061        & (0.060)      \\
				\multicolumn{2}{l}{Year in school}                     & 0.119        & (0.012)      & 0.119        & (0.012)      \\
				\multicolumn{2}{l}{Live with both parents}             & 0.432        & (0.033)      & 0.359        & (0.040)      \\
				\multicolumn{2}{l}{Mother education}                   &              &              &              &              \\
				& $<$ High                        & -0.317       & (0.042)      & -0.187       & (0.048)      \\
				& $>$ High                        & 0.803        & (0.048)      & 0.659        & (0.039)      \\
				& Missing                         & -0.036       & (0.066)      & 0.143        & (0.073)      \\
				\multicolumn{2}{l}{Mother job}                         &              &              &              &              \\
				& Professional                    & 0.712        & (0.068)      & 0.529        & (0.053)      \\
				& Other                           & 0.224        & (0.041)      & 0.150        & (0.050)      \\
				& Missing                         & -0.055       & (0.058)      & -0.005       & (0.070)     \\\bottomrule
			\end{tabular}
			\begin{tablenotes}[para,flushleft]
				\footnotesize
				Notes: Estimates without parentheses are marginal effects; standard errors are shown in parentheses.
			\end{tablenotes}
		\end{threeparttable}
	\end{table}

	{\fontsize{11}{10}\selectfont
		\bibliographyoa{References}
		\bibliographystyleoa{ecta} }
	
\end{document}